\documentclass[11pt]{article}
\usepackage{fullpage,amsfonts,amsmath,amssymb,amsthm,bm,xspace,thmtools,thm-restate,asymptote,mathrsfs,appendix,array,comment}
\usepackage{multirow,graphicx, algorithm,algorithmic,caption,subcaption,url,cite,float}
\usepackage{mathtools}
\usepackage[colorlinks=true,urlcolor=blue,citecolor=blue]{hyperref}
\usepackage{todonotes}
\usepackage{array} 
\usepackage[noabbrev]{cleveref}
\usepackage{listings}
\usepackage{xpatch}
\usepackage{thmtools}
\usepackage{apptools}

\let\phi\varphi
\makeatletter
\xpatchcmd{\thmt@restatable}{\csname #2\@xa\endcsname\ifx\@nx#1\@nx\else[{#1}]\fi}%
{\IfAppendix{\csname #2\@xa\endcsname}{\csname #2\@xa\endcsname[{#1}]}}
\makeatother

\newcolumntype{C}{ >{\centering\arraybackslash} m{7cm} }
\newcolumntype{D}{ >{\centering\arraybackslash} m{2cm} }
\newcolumntype{E}{ >{\centering\arraybackslash} m{1.5cm} }
\newcolumntype{F}{ >{\centering\arraybackslash} m{8.3cm} }
\newcolumntype{G}{ >{\centering\arraybackslash} m{8cm} }
\newcolumntype{H}{ >{\centering\arraybackslash} m{2.5cm} }
\newcolumntype{J}{ >{\centering\arraybackslash} m{3cm} }
\newcolumntype{K}{ >{\centering\arraybackslash} m{6cm} }
\newcolumntype{L}{ >{\centering\arraybackslash} m{1.75cm} }
\newcolumntype{M}{ >{\centering\arraybackslash} m{5cm} }

\newtheorem{theorem}{Theorem}[section]

\newtheorem{lemma}[theorem]{Lemma}
\newtheorem{corollary}[theorem]{Corollary}
\newtheorem{proposition}[theorem]{Proposition}
\newtheorem{definition}[theorem]{Definition}

\newtheorem{condition}[theorem]{Condition}
\Crefname{condition}{Condition}{Conditions}

\theoremstyle{remark}

\numberwithin{equation}{section}

\newcommand{\vphi}{\vec{\phi}}

\newcommand{\hrho}{\hat{\rho}}

\newcommand{\R}{\mathbb{R}}
\newcommand{\RR}{\mathbb{R}}

\newcommand{\cC}{\mathcal{C}}

\newcommand{\II}{\mathcal{I}}

\newcommand{\sign}{\operatorname*{sign}}

\newcommand{\cfg}[1]{\textcolor{blue}{#1}}
\newcommand{\bdb}[1]{#1} 

\newcommand{\MAT}[1]{\begin{bmatrix} #1 \end{bmatrix}}

\newcommand{\abs}[1]{\left| #1 \right|}
\newcommand{\keys}[1]{\left\{ #1 \right\}}

\newcommand{\brac}[1]{\left( #1 \right) }
\newcommand{\ml}[1]{\mathcal{ #1 } }
\newcommand{\op}[1]{ \operatorname{#1} }
\newcommand{\normInf}[1]{\left\| #1 \right\| _{\infty}}
\newcommand{\normTwo}[1]{\left\| #1 \right\| _{2}}
\newcommand{\normOne}[1]{\left\| #1 \right\| _{1}}
\newcommand{\normTV}[1]{\left\| #1 \right\| _{\op{TV}}}

\newcommand{\PROD}[2]{\left \langle #1, #2\right \rangle}
\newcommand{\diff}[1]{ \, \text{d} #1 }
\newcommand{\diffbrac}[1]{ (\text{d} #1) }

\usepackage{pgfplots}
\usetikzlibrary{spy}
\usepgfplotslibrary{external}
\usepgfplotslibrary{fillbetween}
\usepgfplotslibrary{groupplots}
\usetikzlibrary{arrows,automata}
\tikzsetexternalprefix{figs_tikz/}

\title{Sparse Recovery Beyond Compressed Sensing:\\ Separable Nonlinear Inverse Problems}

\author{Brett Bernstein\thanks{Courant Institute of Mathematical
    Sciences, New York University}\ \thanks{This paper was
  presented in part at the 2018 SIAM Annual Conference.}\hspace{0.4cm}
   Sheng Liu  \thanks{Center for Data Science,
    New York University} \hspace{0.4cm} Chrysa Papadaniil
   \thanks{Center for Neural Science, New York University} \hspace{0.4cm} Carlos
  Fernandez-Granda\footnotemark[1] \footnotemark[3] }

\date{May 2019}

\begin{document}

\maketitle

\vspace{-0.3in}

\begin{abstract}
Extracting information from nonlinear measurements is a fundamental challenge in data analysis. In this work, we consider separable inverse problems, where the data are modeled as a linear combination of functions that depend nonlinearly on certain parameters of interest. These parameters may represent neuronal activity in a human brain, frequencies of electromagnetic waves, fluorescent probes in a cell, or magnetic relaxation times of biological tissues. Separable nonlinear inverse problems can be reformulated as underdetermined sparse-recovery problems, and solved using convex programming. This approach has had empirical success in a variety of domains, from geophysics to medical imaging, but lacks a theoretical justification. In particular, compressed-sensing theory does not apply, because the measurement operators are deterministic and violate incoherence conditions such as the restricted-isometry property. Our main contribution is a theory for sparse recovery adapted to deterministic settings. We show that convex programming succeeds in recovering the parameters of interest, as long as their values are sufficiently distinct with respect to the correlation structure of the measurement operator. The theoretical results are illustrated through numerical experiments for two applications: heat-source localization and estimation of brain activity from electroencephalography data.  
\end{abstract}


\section{Introduction}

\subsection{Separable Nonlinear Inverse Problems}
\label{sec:snl}
The inverse problem of extracting information from nonlinear
measurements is a fundamental challenge in many applied domains,
including geophysics, microscopy, astronomy, medical imaging, and
signal processing. In this work, we focus on \emph{separable nonlinear} (SNL)
problems~\cite{golub1973differentiation,golub2003separable}, where the
data are modeled as samples from a linear combination of functions
that depend nonlinearly on certain quantities of interest. Depending
on the application, these quantities may represent neuronal activity
in a human brain, oscillation frequencies of electromagnetic waves,
locations of fluorescent probes in a cell, magnetic-resonance
relaxation times of biological tissues, or positions of celestial
objects in the sky.

\begin{figure}[tp]
  \centering
  \begin{tabular}{@{\vspace{0.75cm}}L@{\quad}CC}
    & &
    $y=\vphi(\theta_1)+2\vphi(\theta_2)+0.5\vphi(\theta_3)$\\ {\footnotesize
      Deconvolution} &
    \includegraphics{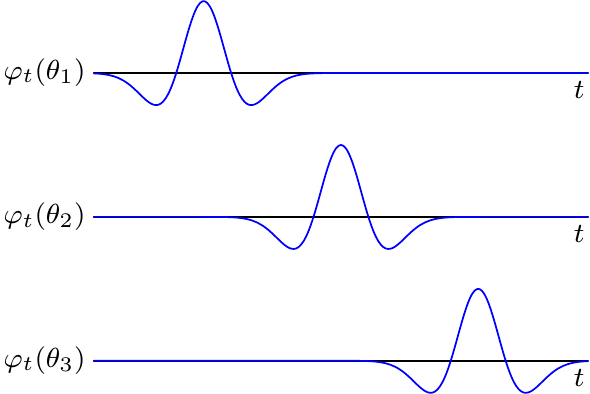}&\includegraphics{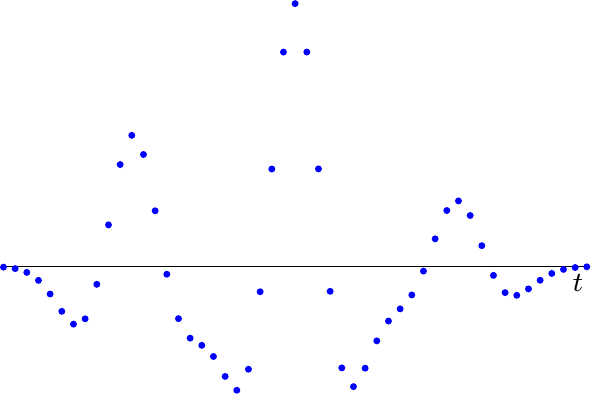}\\ {\footnotesize
      Super-resolution} &
    \includegraphics{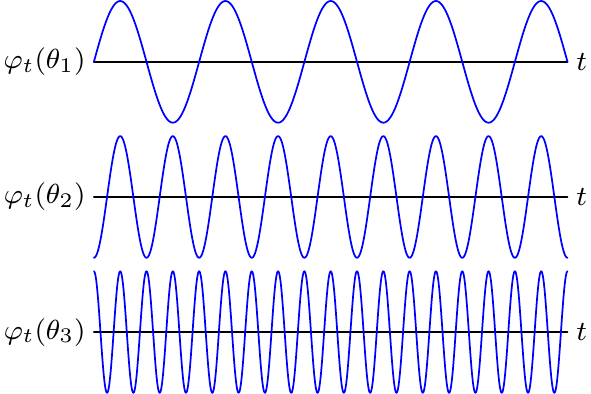}&\includegraphics{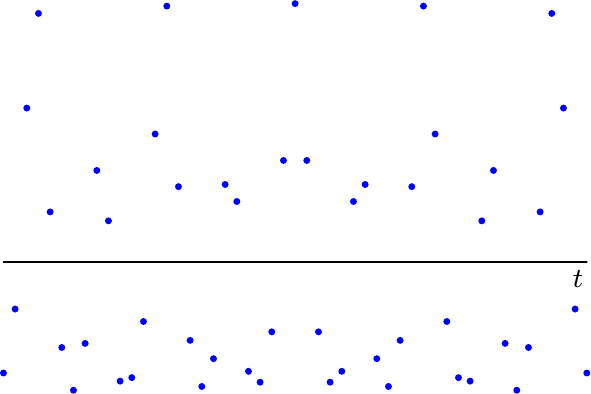}\\ {\footnotesize
      Heat-Source Localization} &
    \includegraphics{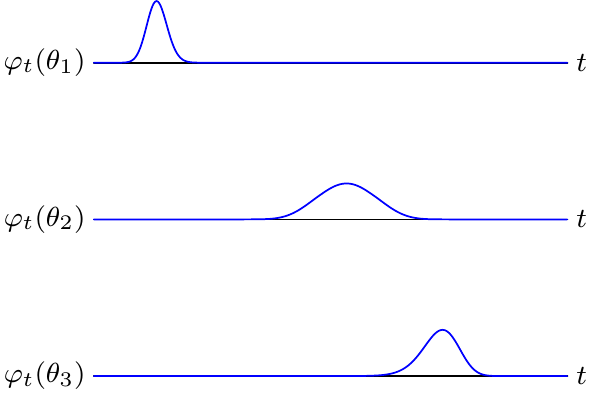}&\includegraphics{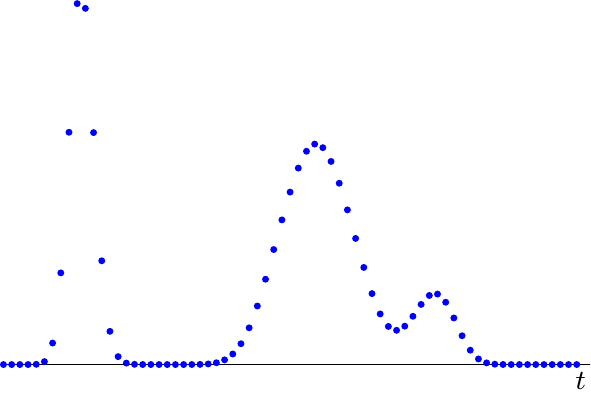}
  \end{tabular}
  \caption{Illustration of three inverse problems that can be modeled as SNL problems: deconvolution in reflection
    seismology (the convolution kernel is a Ricker wavelet~\cite{ricker1953form}), super-resolution of spectral lines, and
    heat-source localization.  The left column shows the
    continuous measurements $\phi_t$ for three parameters
    $\theta_1,\theta_2,\theta_3$.  The right column shows the data
    samples corresponding to an example where the coefficients are set to $c:=(1,2,0.5)$.  For
    super-resolution, only the real part of the data is shown. }
  \label{fig:applications}
\end{figure}

\begin{figure}[tp]
  \centering 
\begin{tabular}{ >{\centering\arraybackslash}m{0.4\linewidth} >{\centering\arraybackslash}m{0.4\linewidth} >{\centering\arraybackslash}m{0.05\linewidth}  }
    Parameter space \vspace{0.2cm} &
    $\vphi(\theta_1)$ \vspace{0.2cm}\\
    \vspace{0.5cm} \includegraphics[scale=0.5]{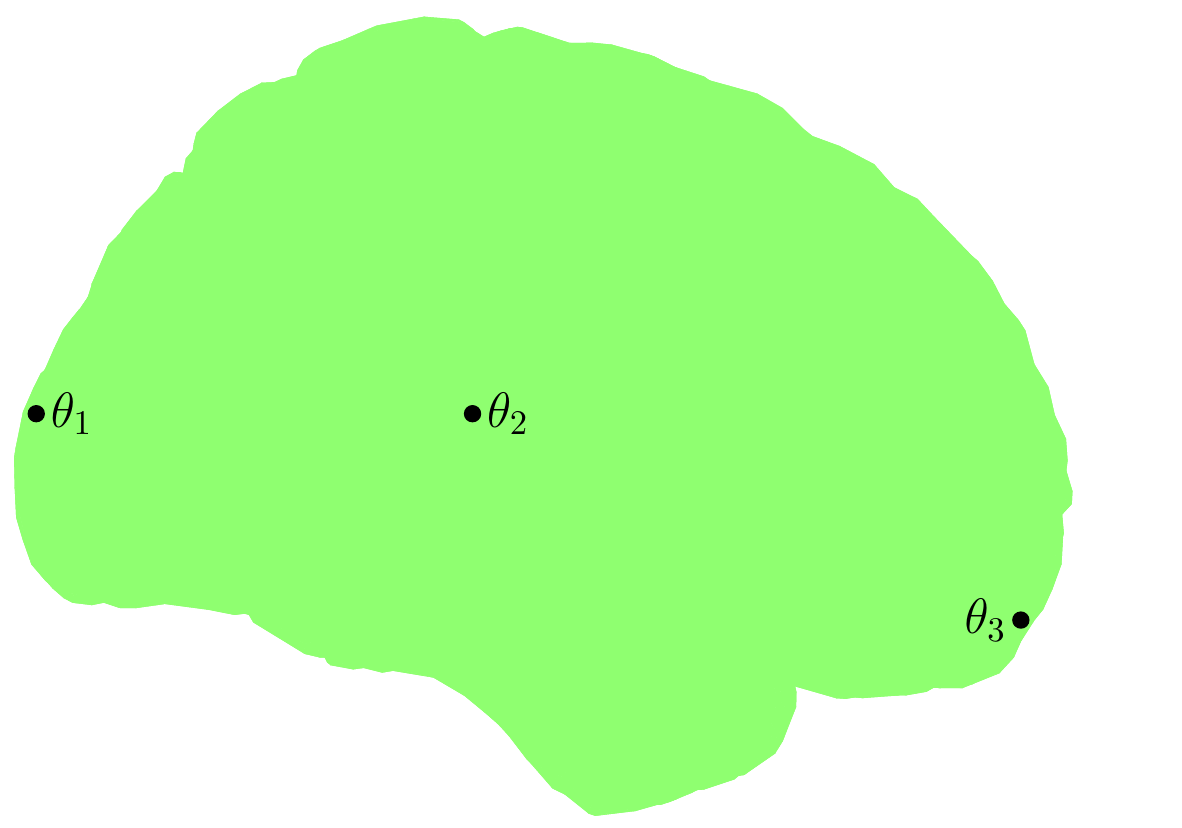}&
    \includegraphics[scale=0.3]{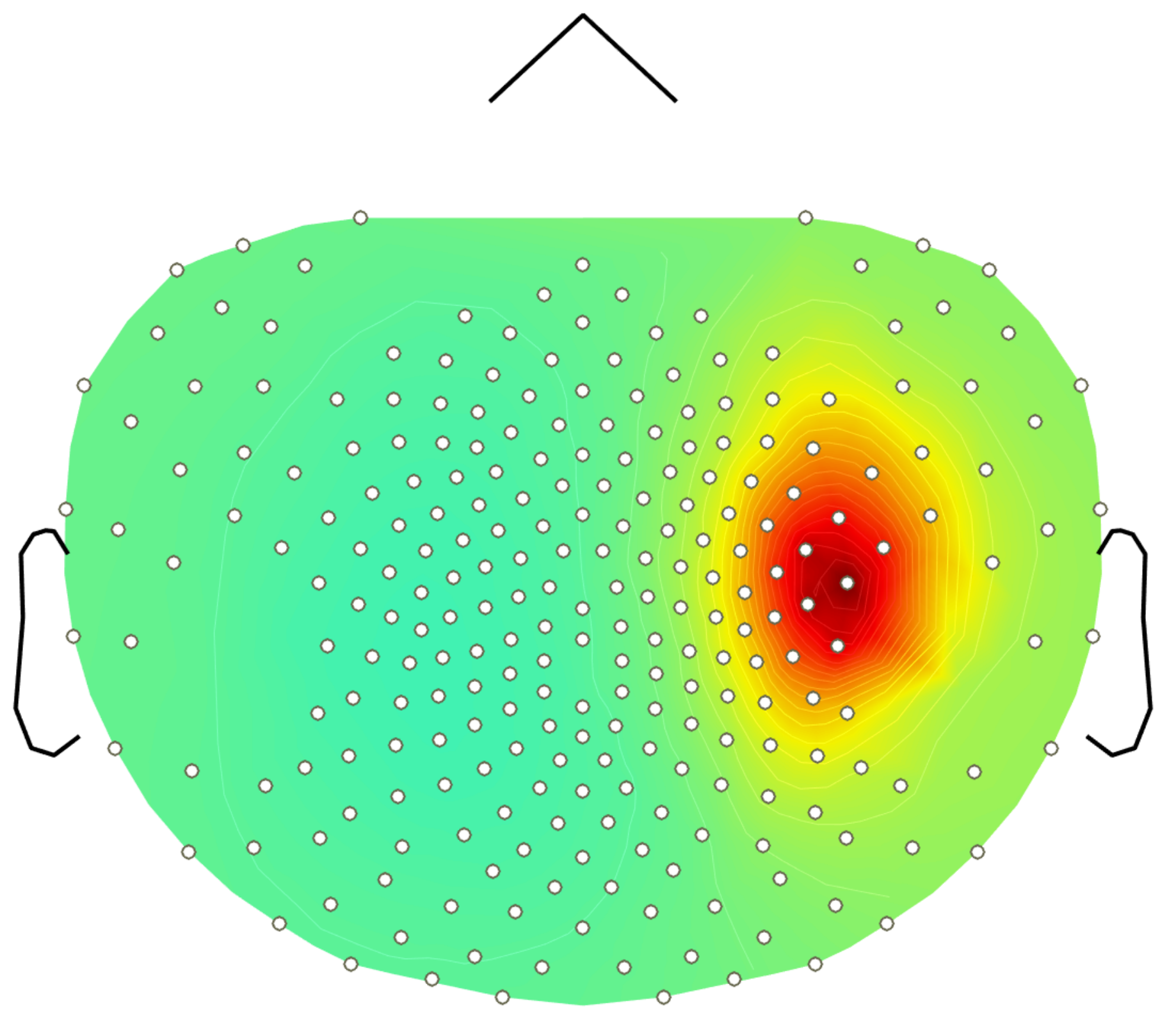} &
    \\ \includegraphics[scale=0.3]{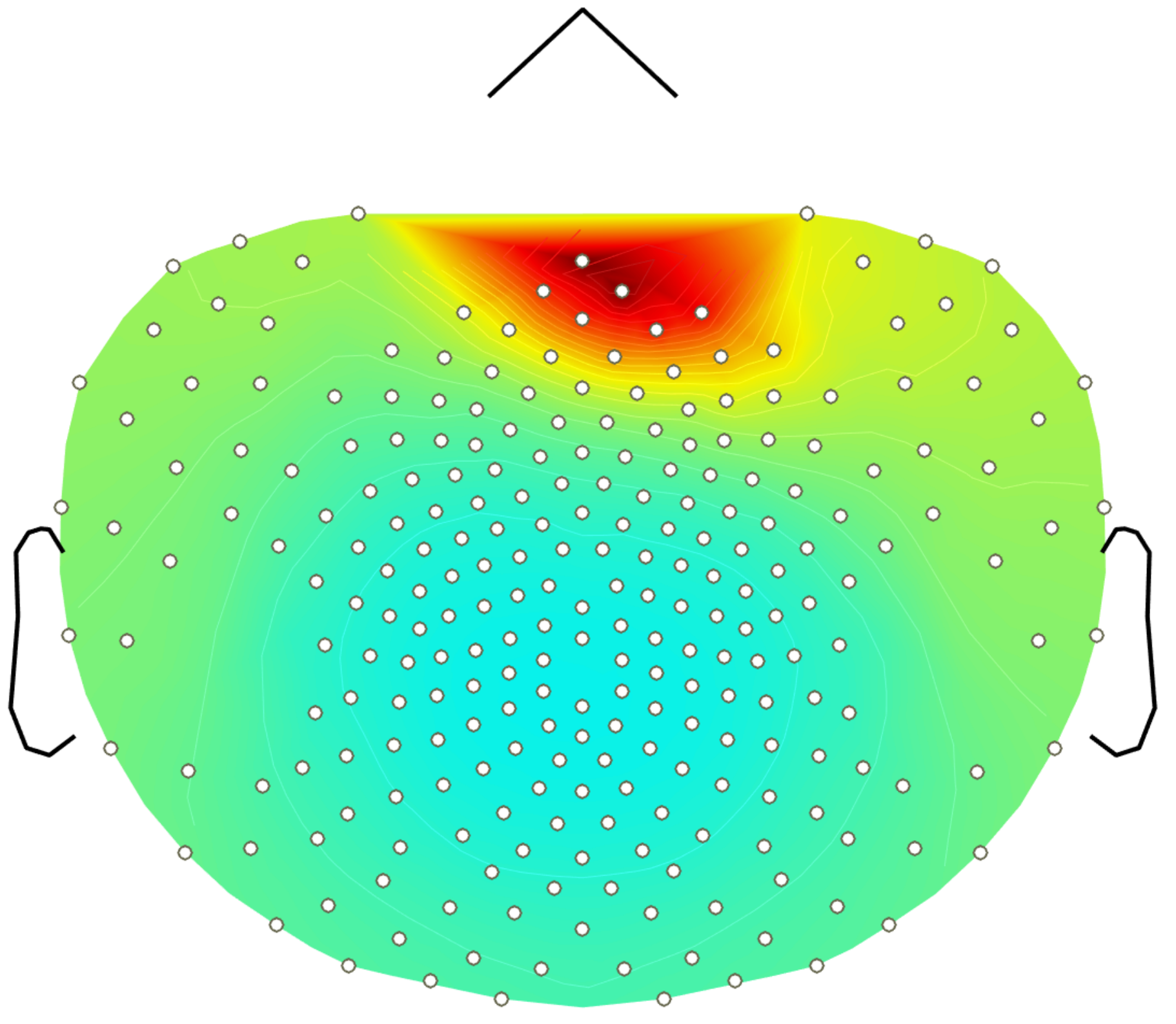} \vspace{0.2cm}
    &\includegraphics[scale=0.3]{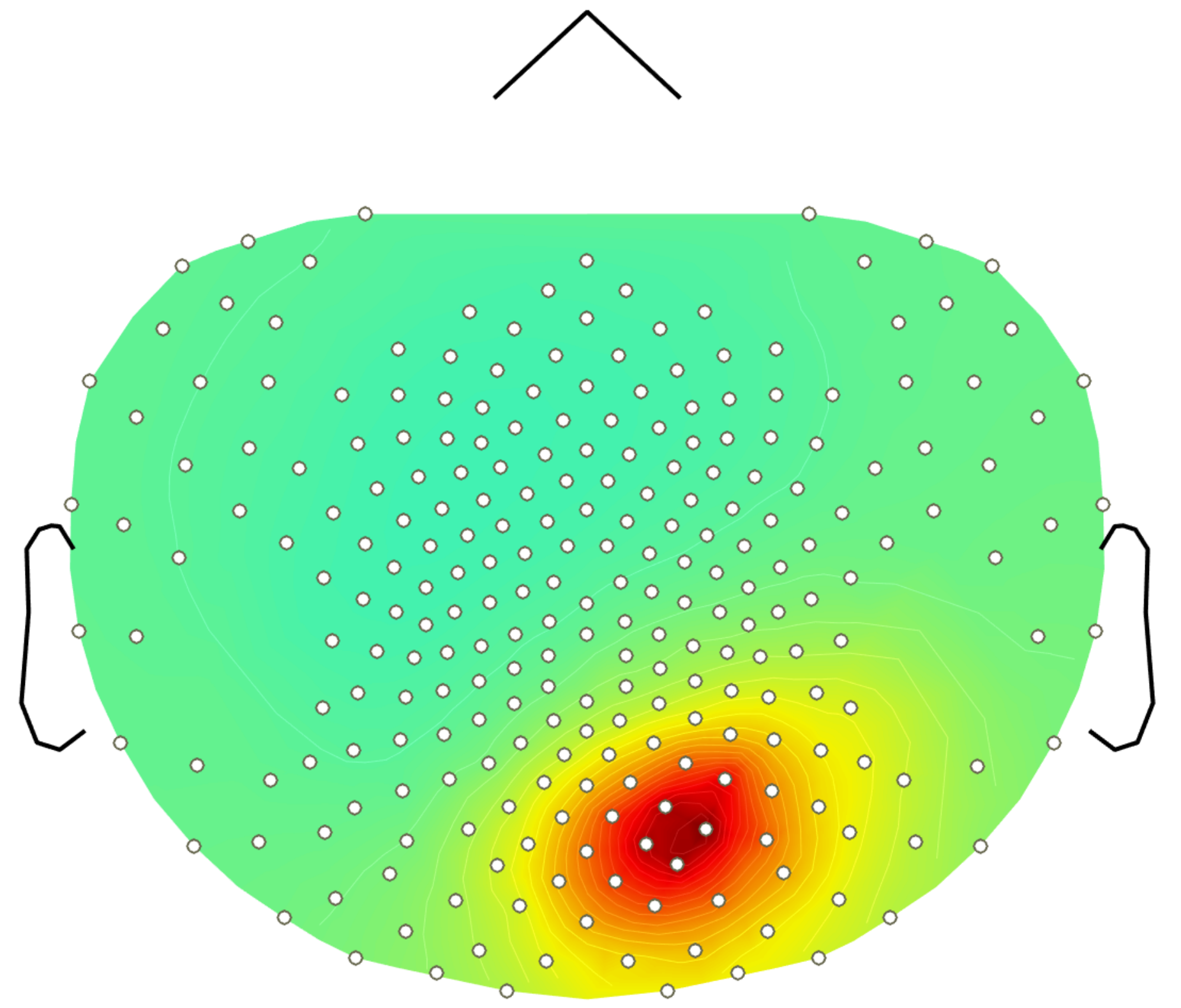} \vspace{0.2cm}&
    \includegraphics[scale=0.6]{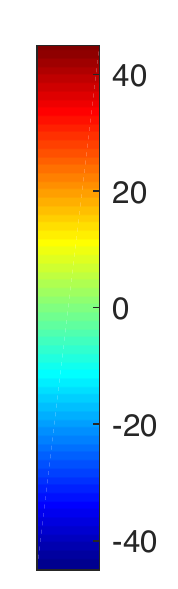}
    \\ $\vphi(\theta_2)$ & $\vphi(\theta_3)$
  \end{tabular}
  \caption{Localization of brain-activity sources from EEG data is an SNL problem. The image in the top left corner shows the position of three sources in a human brain, situated in the occipital ($\theta_1$), temporal ($\theta_2$) and frontal ($\theta_3$) lobes. The remaining three images show the EEG data corresponding to each of these sources, obtained from 256 sensors located on the surface of the head (see Section~\ref{sec:eeg} for more details).}
  \label{fig:eegatoms}
\end{figure}

Mathematically, the goal in an SNL problem is to estimate $k$
parameters $\theta_1$, \ldots, $\theta_k \in \R^{p}$ from samples of a
function
\begin{align}
  \label{eq:snl_cts}
  f(t) := \sum_{i=1}^k c_i \phi_t(\theta_i),
\end{align}
where $c_1$, \ldots, $c_k \in \R$ are unknown coefficients. The
dependence between each component and the corresponding parameter at a
particular value of $t$ is governed by a nonlinear map $\phi_t: \R^p
\rightarrow \R $. For simplicity of exposition, we assume that $t$ is
one-dimensional and $f$ is a real-valued function, but the framework
can be directly extended to multidimensional and complex-valued
measurements. The data are samples of $f$ at $n$ locations $s_1$,
\ldots, $s_n \in \R$
\begin{align}
  \label{eq:snl_model}
  y & := \MAT{ f\brac{s_1} \\ \vdots \\ f\brac{s_n} } = \sum_{i=1}^k
  c_i \vec{\phi}(\theta_i),
\end{align}
where $\vec{\phi}(\theta_i)_j:=\phi_{s_j}(\theta_i)$, $1\leq j \leq
n$. Each $\vec{\phi}(\theta_i) \in \R^n$ is a feature vector
associated to one of the parameters $\theta_i$. Without loss of generality, 
we assume that the feature vectors are normalized, i.e. $\|\vphi(\theta)\|_2=1$ for all $\theta\in\RR^p$. 
The following examples
illustrate the importance of SNL problems in a range of
applications.

\begin{itemize}
\item \emph{Deconvolution of point sources}: Deconvolution consists of
  estimating a signal from samples of its convolution with a fixed
  kernel $K$. When the signal is modeled as a superposition of point
  sources or spikes, representing fluorescent probes in
  microscopy~\cite{palm,fpalm}, celestial bodies in
  astronomy~\cite{astronomy_puschmann} or interfaces between
  geological layers in seismography~\cite{sheriff1995exploration},
  this is an SNL problem where $\theta_1$, \ldots, $\theta_k$ are the
  locations of the spikes. In that case, $\phi_t\brac{\theta}$ is a
  shifted copy of the convolution kernel $K \brac{ t - \theta}$, as
  illustrated in the top row of Figure~\ref{fig:applications}.
\item \emph{Spectral super-resolution}: Super-resolving the spectrum
  of a multisinusoidal signal from samples taken over a short period
  of time is an important problem in communications, radar, and
  signal processing~\cite{Stoica:2005wf}. This is an SNL problem where
  $\phi_t\brac{\theta}$ is a complex exponential $ \exp \brac{ - i 2
    \pi \theta t}$ with frequency $\theta$ (see the second row of
  Figure~\ref{fig:applications}).
\item \emph{Heat-source localization}: Finding the position of several
  heat sources in a material with known conductivity from temperature
  measurements is an SNL problem where $\phi_t\brac{\theta}$ is the
  Green's function of the heat equation parametrized by the location
  $\theta$ of a particular heat source~\cite{li2014heat}. The bottom
  row of Figure~\ref{fig:applications} shows an example (see \Cref{sec:numexact} for more details).
\item \emph{Estimation of neural activity}: Electroencephalography measurements of the electric potential field on the surface of
  the head can be used to detect regions of
  focalized activity in the brain~\cite{niedermeyer2005eeg}. The data are well approximated by an SNL model where the parameters are the locations of these regions~\cite{michel2004eeg}. The function $\phi_t\brac{\theta}$ represents the potential at a specific location $t$ on the scalp, which originates from neural activity at position $\theta$ in the brain. This function can be computed by solving the Poisson differential equation taking into account the
  geometry and electric properties of the
  head~\cite{nunez2006electric}. Figure~\ref{fig:eegatoms} shows an example. See Section~\ref{sec:eeg} for more details. 
  \item \emph{Quantitative magnetic-resonance imaging}: The magnetic-resonance relaxation times $T_1$ and $T_2$ of biological tissues govern the local fluctuations of the magnetic field measured by MR imaging systems~\cite{nishimura1996principles}. MR fingerprinting is a technique to estimate these parameters by fitting an SNL model where each component corresponds to a different tissue~\cite{ma2013magnetic,McGivney2017,mcmrf}. In this case, the parameter $\theta \in \R^2$ encodes the values of $T_1$ and $T_2$ and the function $\phi_t\brac{\theta}$ can be computed by solving the Bloch differential equations~\cite{bloch1946nuclear}.
\end{itemize}


\subsection{Reformulation as a Sparse-Recovery Problem}
\label{sec:sparse_recovery}
A natural approach to estimate the parameters of an SNL model is to
solve the nonlinear least-squares problem,
\begin{align}
\label{pr:varpro}
\underset{\tilde{\theta}_1,\ldots, \tilde{\theta}_k \in \R^p,\;
  \tilde{c} \in \R^{k}}{\op{minimize}} \quad \normTwo{ y -
  \sum_{i=1}^k \tilde{c}_i \vec{\phi}(\tilde{\theta}_i)}^2.
\end{align}
Unfortunately, the resulting cost function is typically nonconvex and
has local minima, as illustrated by the simple example in
Figure~\ref{fig:varproj}. Consequently, local-descent methods do not
necessarily recover the true parameters, even in the absence of noise,
and global optimization becomes intractable unless $k$ is very
small. 

Alternatively, we can reformulate the SNL problem as a sparse-recovery problem and leverage $\ell_1$-norm minimization to solve it. This approach was pioneered in the 1970s by geophysicists working on spike deconvolution in the context of reflection seismology~\cite{taylor1979deconvolution,claerbout,levy,santosa,debeye1990lp}. Since then, it has been applied to many SNL problems such as
acoustic sensing~\cite{zhao2011localization,bertin2015compressive},
radar~\cite{potter2010sparsity,tang2011aliasing},
electroencephalography (EEG)~\cite{silva2004evaluation,xu2007lp},
positron emission tomography (PET)
\cite{gunn2002positron,reader2007fully,heins2014locally},
direction of arrival \cite{Malioutov:2005jw,borcea2015resolution},
quantitative magnetic resonance imaging~
\cite{McGivney2017,mcmrf}, and source localization~\cite{li2014heat,mamonov2013point,pieper2018inverse}. Our goal is to provide a theory of sparse recovery via convex optimization explaining the empirical success of this approach.

\begin{figure}[tp]
  \centering \includegraphics{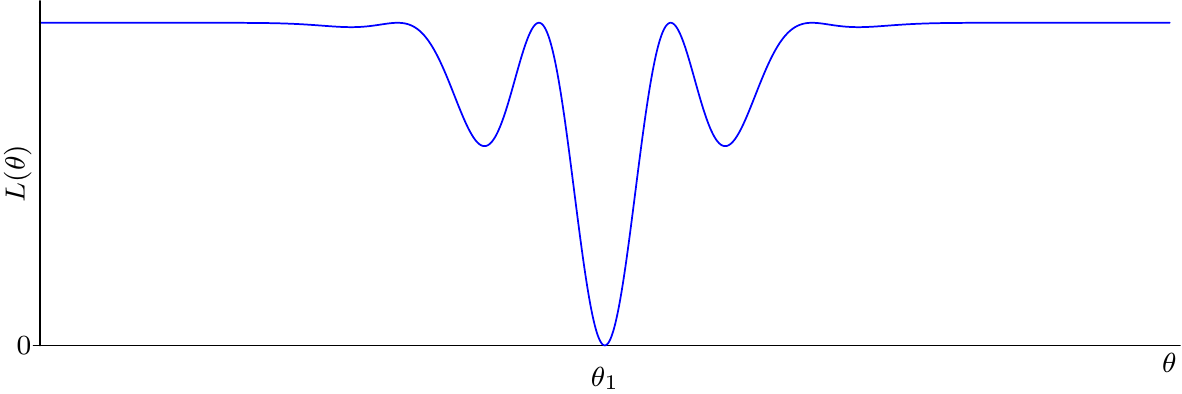}
  \caption{Nonlinear least-squares cost functions associated to deterministic SNL problems often have non-optimal local minima. The graph depicts the landscape of the nonlinear least squares cost function $L(\theta)=\min_{\tilde{c}\in\RR}\normTwo{\vphi(\theta_1)-\tilde{c}\vphi(\theta)}^2$ associated to a deconvolution problem, where the data is generated by convolving a Ricker wavelet~\cite{ricker1953form} with a single spike located $\theta_1$. In addition to the global minimum at $\theta_1$ there are several
    spurious local minima. 
  }
  \label{fig:varproj}
\end{figure}

Let us represent the parameters $\theta_1$, \ldots, $\theta_k \in
\R^{p}$ of an SNL model as a superposition of Dirac measures or
\emph{spikes} in $\R^{p}$, interpreted as a $p$-dimensional
\emph{parameter space}, 
\eqref{eq:snl_model},
\begin{align}
\label{eq:x_def}
\mu := \sum_{i=1}^k c_i \delta_{\theta_i},
\end{align}
where $\delta_{\theta_i}$ denotes a Dirac measure supported at
$\theta_i$. Intuitively, the atomic measure $\mu$ is a signal that
encodes the parameters of interest and their corresponding
coefficients. The data described by \Cref{eq:snl_model} can
now be expressed as
\begin{align}
\label{eq:y_model}
y = \int_{\R^p} \vec{\phi}(\theta) \mu \diffbrac{\theta}.
\end{align}
The SNL problem is equivalent to recovering $\mu$ from these linear
measurements. The price we pay for linearizing is that the
linear inverse problem is extremely underdetermined: $y$ has dimension
$n$, but $\mu$ lives in a continuous space of infinite
dimensionality! To solve the problem, we need to exploit the assumption that
the data only depends on a small number of parameters or,
equivalently, that $\mu$ is \emph{sparse}.

For an SNL problem to be well posed, $\theta_1$, \ldots, $\theta_k$
should be the only set of $k$ or less parameters such that
\Cref{eq:snl_model} holds. In that case, $\mu$ is the solution to the
sparse-recovery problem
\begin{equation}
\label{eq:min_cardinality}
  \begin{aligned}
    \underset{\tilde{\mu}}{\op{minimize}} \quad&
    \abs{\op{support}\brac{ \tilde{\mu} }} \\ \text{subject to} \quad&
    \int_{\R^p} \vec{\phi}(\theta) \tilde{\mu} \diffbrac{\theta} = y,
  \end{aligned}
\end{equation}
where minimization occurs over the set of measures in $\R^p$. The
cardinality of the support of an atomic measure is a nonconvex
function, which is notoriously challenging to minimize.
A fundamental insight underlying many approaches to sparse estimation in high-dimensional statistics and
signal processing is that one can bypass this difficulty by replacing the nonconvex function with a convex
counterpart. In particular, minimizing the $\ell_1$ norm instead of the cardinality function has proven to be very effective in many applications. 

In order to analyze the application of $\ell_1$-norm minimization
  to SNL problems, we consider a continuous setting, where the
  optimization variable is a measure supported on a continuous
  domain. The goal is to obtain an analysis that is valid for
  arbitrarily fine discretizations of the domain. This is important
  because, as we will see below, a fine discretization
  results in a highly-correlated linear operator, which violates the
  usual assumptions made in the literature on sparse recovery.
 
In the case of measures supported on a continuous domain, the
continuous counterpart of the $\ell_1$ norm is the total-variation
(TV) norm~\cite{rudin1987real,folland2013real}\footnote{Not to be
  confused with the total variation of a piecewise-constant function
  used in image processing.}.
Indeed, the TV norm of the atomic measure
$\mu$ in \Cref{eq:x_def} equals the $\ell_1$ norm of its
coefficients $\normOne{c}$. Just as the $\ell_1$ norm is the dual norm
of the $\ell_{\infty}$ norm, the TV norm is defined by
\begin{equation}
  \normTV{\mu} := \sup_{f\in\cC,~\normInf{f}\leq 1} \left|\int f
  \diff{\mu} \right|,
\end{equation}
where the supremum is taken over all continuous functions in the unit
$\ml{L}_{\infty}$-norm ball. Replacing the cardinality function by
this sparsity-promoting norm yields the following convex program
\begin{equation}
\label{pr:min_TV}
  \begin{aligned}
    \underset{\tilde{\mu}}{\op{minimize}} \quad& \normTV{ \tilde{\mu}
    } \\ \text{subject to} \quad& \int_{\R^p} \vec{\phi}(\theta)
    \tilde{\mu} \diffbrac{\theta} = y.
  \end{aligned}
\end{equation}
The goal of this paper is to understand when the solution to
  Problem~\eqref{pr:min_TV} exactly recovers the parameters
  of an SNL model.

\subsection{Compressed Sensing}
\label{sec:cs}



\Cref{sec:sparse_recovery} shows that solving an SNL problem is
equivalent to recovering a sparse signal from linear underdetermined
measurements. This is reminiscent of \emph{compressed
  sensing}~\cite{donoho2006compressed,candes2008introduction,foucart2013mathematical}. In
its most basic formulation, the goal of compressed sensing is to
estimate a signal $x \in \R^{m}$ with $k$ nonzeros from linear
measurements $y \in \R^n$ given by $y := Ax$, where $A \in \R^{n
  \times m}$ and $m > n$. Remarkably, exact recovery of $x$ is still possible under certain
conditions on the matrix $A$, even though the linear system is
underdetermined.

Overdetermined linear inverse problems are said to be ill posed when
the measurement matrix is ill conditioned. This occurs when there
exist vectors that lie close to the null space of the matrix, or
equivalently when a subset of its columns is highly
correlated. Analogously, the compressed-sensing problem is ill posed
if any \emph{sparse} subset of columns is highly correlated, because
this implies that sparse vectors lie close to the null space. Early
works on compressed sensing derive recovery guarantees assuming a
bound on the maximum correlation between the columns of the
measurement matrix $A$ (sometimes called \emph{incoherence}).
They prove that tractable algorithms such as
$\ell_1$-norm minimization and greedy techniques achieve exact
recovery as long as the maximum correlation is of order
$n^{-1/2}$ for sparsity levels $k$ of up to order $\sqrt{n}$~\cite{donoho2003optimally,donoho2001uncertainty,gribonval2003sparse,Tropp:2004gc}, even if the data
are corrupted by additive
noise~\cite{donoho2006stable,Tropp:2006jc}. These results were
subsequently strengthened to sparsity levels of order $n$ (up to
logarithmic
factors)~\cite{donoho2006compressed,needell2009cosamp,candes2005decoding,bickel2009simultaneous,candes2008restricted}
under stricter assumptions on the conditioning of sparse subsets of
columns in the measurement matrix, such as the restricted-isometry
property~\cite{candes2005decoding} or the restricted-eigenvalue
condition~\cite{bickel2009simultaneous}.

\begin{figure}[tp]
  \centering {
  \begin{tabular}{CC}
    \includegraphics{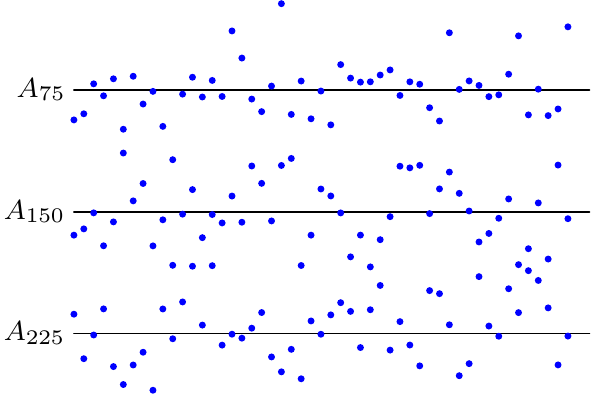}&\includegraphics{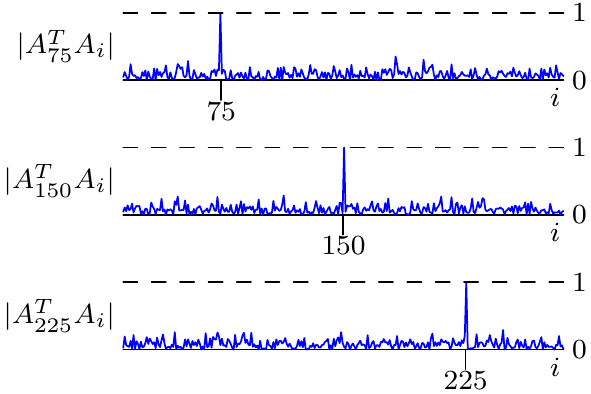}
  \end{tabular}
  }
  \caption{Columns and intercolumn correlations for a typical
    compressed-sensing matrix $A\in\RR^{100\times 300}$ with
    i.i.d.~standard Gaussian entries.  The notation $A_i$ denotes the
    $i$th column of $A$.  Note that $|A_i^TA_j|$ is small whenever
    $i\neq j$.  }
  \label{fig:cs_correlations}
\end{figure}

The question is whether compressed-sensing theory applies to SNL
problems. Let us consider an SNL problem where the parameter space is
discretized to yield a finite-dimensional version of the
sparse-recovery problem described in \Cref{sec:sparse_recovery}. The
measurement model in \Cref{eq:y_model} can then be expressed
as
\begin{align}
y & = \sum_{j=1}^{m} \vec{\phi}(\eta_j) x_j \\ & =\Phi
x. \label{eq:y_model_matrix}
\end{align}
$\Phi$ is a measurement matrix whose columns correspond to the feature
vectors of $\eta_1$, \ldots, $\eta_m \in \RR^{p}$, which denote the
$m$ points of the discretized parameter space. The signal $x \in
\RR^m$ is the discrete version of $\mu$ in \Cref{eq:x_def}:
a sparse vector, such that $x_i=c_j$ when $\eta_j=\theta_i$ for some
$i \in \keys{1,\ldots,k}$ and $x_j=0$ otherwise. For
compressed-sensing theory to apply here, the intercolumn correlations
of $\Phi$ should be very low. \Cref{fig:cs_correlations} shows the
intercolumn correlations of a typical compressed-sensing matrix: each
column is almost uncorrelated with every other
column. \Cref{fig:sccorrs,fig:eegcorr} show the correlation function $\rho_\theta: \R^p
\rightarrow \R$
\begin{equation}
\label{eq:corr_func}
  \rho_{\theta}(\eta) := \PROD{\vphi(\theta)}{\vphi(\eta)},
\end{equation}
for the different SNL problems discussed in
Section~\ref{sec:snl}. 
The intercolumn correlations of the measurement matrix $\Phi$ in the
corresponding discretized SNL problems are given by samples of
$\rho_{\eta_i}$ at the locations of the remaining grid points
$\eta_j$, $i \neq j$. The contrast with the intercolumn structure of
the compressed-sensing matrix is striking: nearby columns in all SNL
measurement matrices are very highly correlated. This occurs under any
reasonable discretization of the parameter space; discretizing very
coarsely would result in inaccurate parameter estimates, defeating the
whole point of solving the SNL problem.

\begin{figure}[tp]
  \centering {
  \begin{tabular}{@{\vspace{0.75cm}}L@{\quad}CC}
    {\footnotesize Deconvolution} &
    \includegraphics{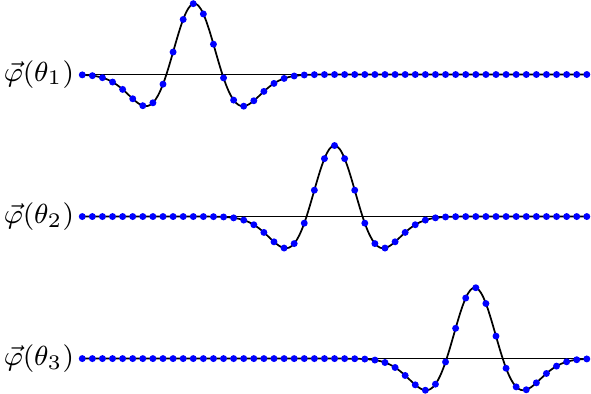}&\includegraphics{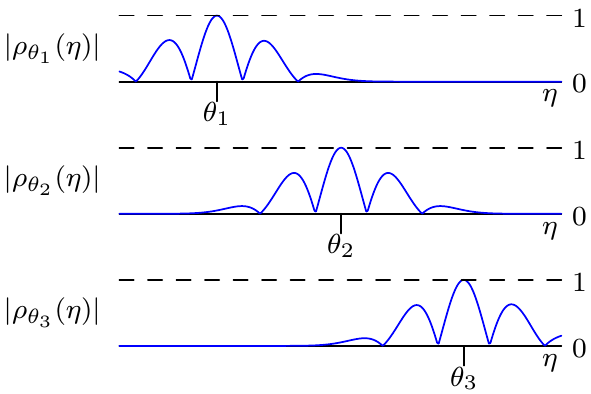}\\
                    {\footnotesize Super-resolution} &
                    \includegraphics{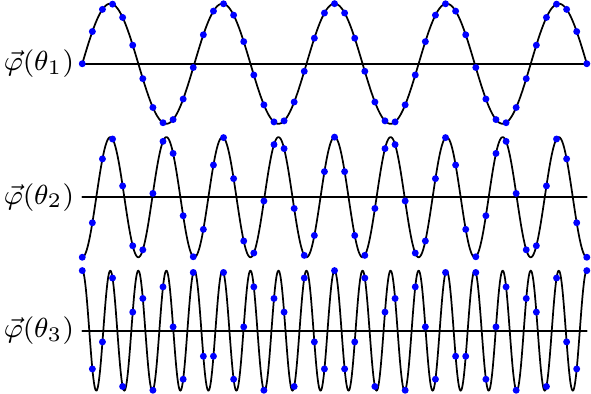}&\includegraphics{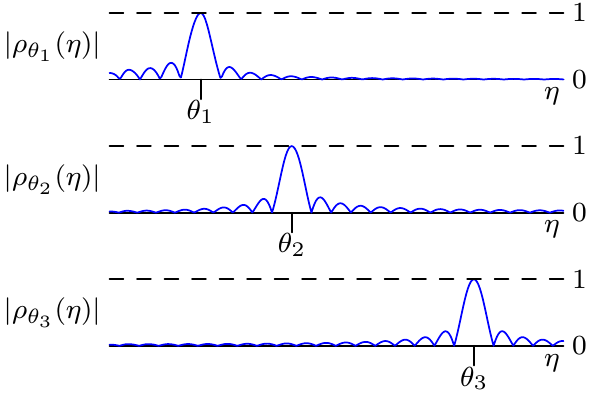}\\
                                    {\footnotesize Heat-Source
                                      Localization} &
                                    \includegraphics{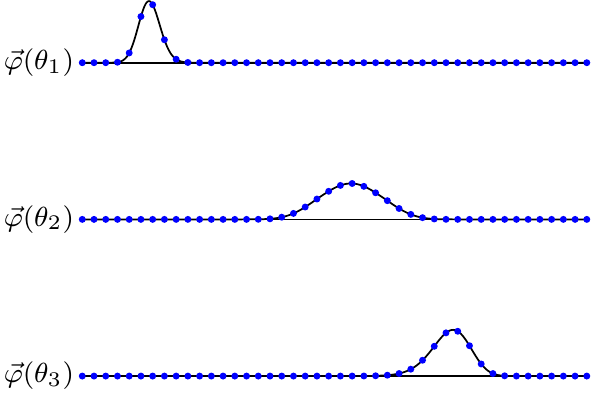}&\includegraphics{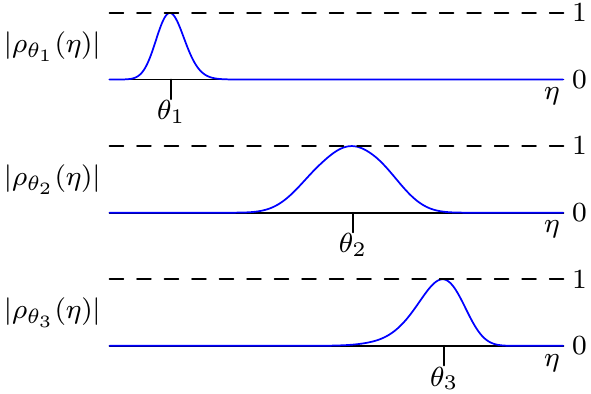}
  \end{tabular}
  }
  \caption{Correlation structure of the measurement operators arising in
    deconvolution, super-resolution, and heat-source localization.
    The left column shows the discrete data $\vphi(\theta_i)$ for
    three parameter values.  The right column gives the absolute
    values of the corresponding correlation functions
    $\rho_{\theta_i}(\eta)=\vphi(\theta_i)^T\vphi(\eta)$.  }
  \label{fig:sccorrs}
\end{figure}

\begin{figure}[tp]
  \hspace{1.8cm} $\rho_{\theta_1}$ \hspace{4.8cm}
  $\rho_{\theta_2}$ \hspace{4.8cm} $\rho_{\theta_3}$ \\ 
  \includegraphics[scale=0.6]{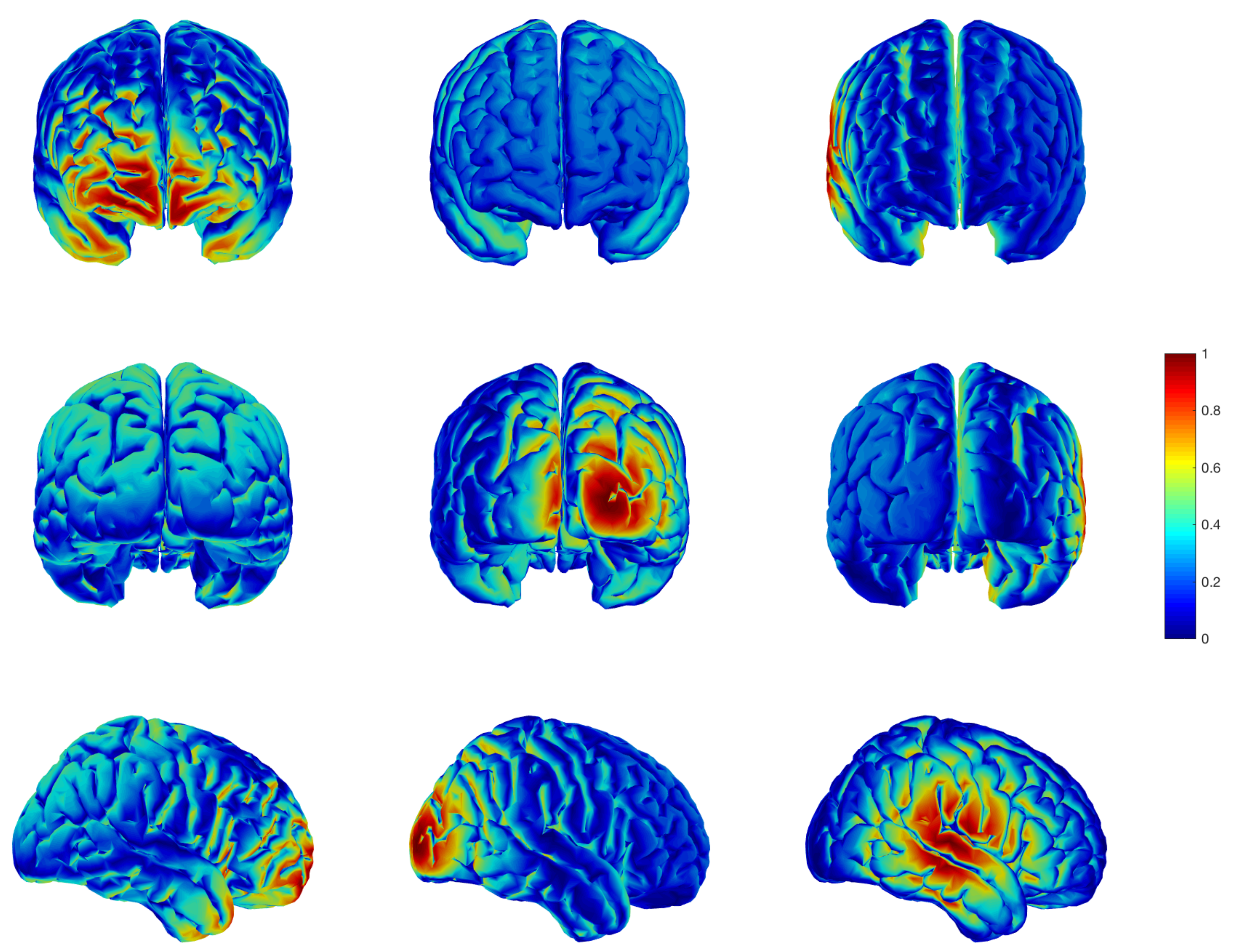}
%
  \caption{Example of correlation functions for the electroencephalography brain-activity localization problem. Each column shows three views of the correlation function $\rho_{\theta_i}(\eta)=\vphi(\theta_i)^T\vphi(\eta)$, $i=1,2,3$, corresponding to one of the neural-activity sources shown in Figure~\ref{fig:eegatoms}.}
  \label{fig:eegcorr}
\end{figure}

The difference between the correlation structure of the measurement
matrix in compressed-sensing and SNL problems is not surprising. The
entries of compressed-sensing matrices are \emph{random}. As a result,
small subsets of columns are almost uncorrelated with high
probability. In contrast, matrices in discretized SNL problems arise
from a deterministic model tied to an underlying continuous parameter
space and to a function $\phi_t$ that is typically smooth. Since
$\phi_t\brac{\theta} \approx \phi_t\brac{\theta'}$ when $\theta
\approx \theta'$, nearby columns are highly correlated. These matrices
do not satisfy any of the properties of the conditioning of sparse
submatrices commonly assumed in compressed sensing. In conclusion, the
answer to our previous question is a resounding \emph{no}:
compressed-sensing theory does not apply to SNL problems.

\subsection{Beyond Sparsity and Randomness: Separation and Correlation Decay}
\label{sec:separation}

The fact that compressed-sensing theory
does not apply to SNL problems involving deterministic measurements is not a theoretical
artifact. Sparsity is not a strong enough condition to ensure that such SNL
problems are well posed. If $\phi_t$ is smooth, which is usually the
case in applications, the features $\vphi\brac{\theta}$ corresponding
to parameters that are clustered in the parameter space are highly
correlated. This can be seen in the correlation plots
of \Cref{fig:sccorrs,fig:eegcorr}.
As a result, different sparse combinations of features may
yield essentially the same data. For a detailed analysis of
this issue in the context of super-resolution and deconvolution of
point sources we refer the reader to Section~3.2 in \cite{superres}
(see also \cite{moitra_superres} and \cite{slepian}) and Section~2.1
in~\cite{bernstein2017deconvolution}, respectively.

Additional assumptions beyond sparsity are necessary to establish
recovery guarantees for SNL problems. At the very least, the features
$\vphi\brac{\theta_1}$, \ldots, $\vphi\brac{\theta_k}$ in the data
cannot be too correlated. For arbitrary SNL problems it is challenging
to define simple conditions to preclude this from happening. However,
in most practical situations, SNL problems exhibit \emph{correlation
  decay}, meaning that the correlation function $\rho_{\theta}$
defined in \Cref{eq:corr_func} is bounded by a decaying
function away from $\theta$. This is a natural property: the more separated
two parameters $\theta$ and $\theta'$ are in the parameter space, the
less correlated we expect their features $\vphi\brac{\theta}$ and
$\vphi\brac{\theta'}$ to be. All the examples in Section~\ref{sec:snl}
have correlation decay (see Figures~\ref{fig:sccorrs}
and \ref{fig:eegcorr}). 

For SNL problems with correlation decay there
is a simple way of ensuring that the features corresponding to the
true parameters $\theta_1$, \ldots, $\theta_k$ are not highly
correlated: imposing a \emph{minimum separation} between them in the
parameter space. The main contribution of this paper is showing that this is in fact sufficient to guarantee that TV-norm minimization achieves exact recovery, under some additional conditions on the derivatives of the correlation function.

\subsection{Organization}

In \Cref{sec:framework} we propose a theoretical framework for the
analysis of sparse estimation in the context of SNL inverse problems.
We focus on the case $p=1$ for simplicity, but our results can be
extended to higher dimensions, as described in \Cref{sec:2d}.
Our main results are \Cref{thm:exact1d_simp,thm:exact1d}, which
establish exact-recovery results for SNL problems with correlation
decay under a minimum separation on the true
parameters. \Cref{sec:RecoveryCert} contains the proof of these
results, which are based on a novel dual-certificate
construction. \Cref{sec:numerical} illustrates the theoretical results
through numerical experiments for two applications: heat-source
localization and estimation of brain activity from
electroencephalography data.

\section{Main Results}
\label{sec:framework}

\subsection{Correlation decay}
\label{sec:correlation_decay}
In this section we formalize the notion of correlation decay by
defining several conditions on the correlation function
$\rho_{\theta}$ and on its derivatives.
Throughout we assume that the problem is one dimensional ($p:=1$). 

To alleviate notation we define
 \begin{equation}
 \label{eq:def_rho}
  \rho_{\theta}^{(q,r)}(\eta) :=
  \vphi^{\,(q)}(\theta)^T\vphi^{\,(r)}(\eta),
\end{equation}
for $q=0,1$ and $r=0,1,2$, where $\vphi^{\,(q)}$ is the $q$th derivative of
$\vphi$.  Recall that we assume $\|\vphi(\theta)\|_2=1$ for all
$\theta\in\RR$.  This implies $\rho_{\theta}(\theta)=1$ and 
$\rho^{(1,0)}_{\theta}(\theta) = \rho^{(0,1)}_{\theta}(\theta)=0$ for
all $\theta\in\RR$.
Plots of these derivatives are
shown in \Cref{fig:corrderivs} for the deconvolution,
super-resolution, and heat-source localization problems. Our conditions take the form of bounds in different regions of the parameter space: a \emph{near} region, an \emph{intermediate} region, and a \emph{far} region, as depicted in Figure~\ref{fig:decay}.

\begin{figure}[t]
  \centering \includegraphics{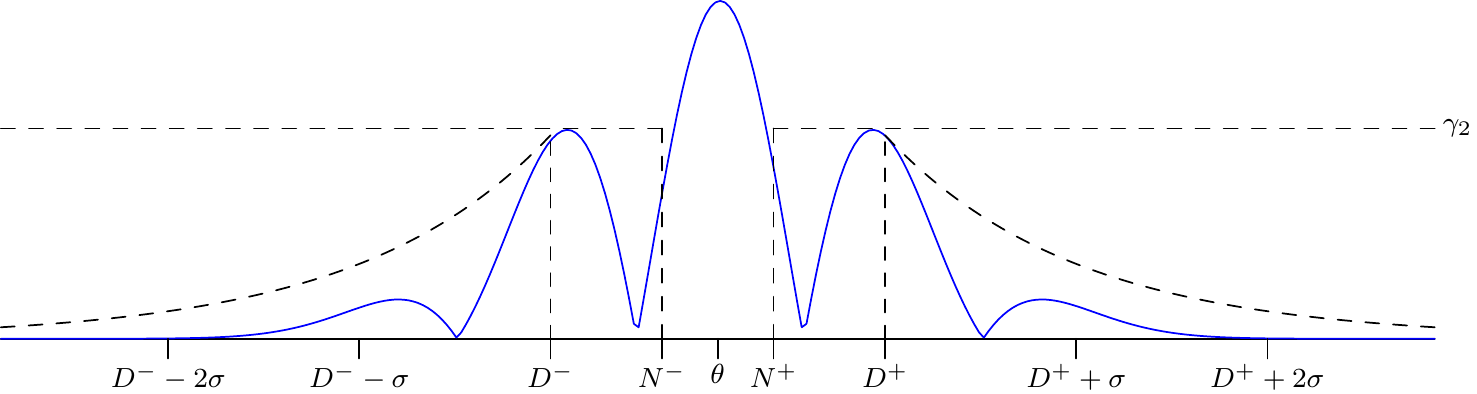}
  \caption{Illustration of the near, intermediate and decay conditions
    on the correlation function defined in
    Section~\ref{sec:correlation_decay}. This example shows a
    correlation function arising in deconvolution.
  }
  \label{fig:decay}
\end{figure}

In the \emph{near} region the correlation can be arbitrarily
  close to one, but is locally bounded by a quadratic function.
\begin{condition}[Near Condition]
  \label{def:near}
  The correlation function $\rho_\theta$ satisfies the near condition if
  \begin{align}
    \rho_{\theta}^{(0,2)}(\eta)&\leq -\gamma_0\quad\text{and}\label{eq:near1}\\
    |\rho_{\theta}^{(1,2)}(\eta)|&\leq \gamma_1\normTwo{\vphi^{\,(1)}(\theta)}^2
    \label{eq:near2}
  \end{align}
  hold for all $\eta$ in $[N^-,N^+]$, where $N^{\pm}:=\theta\pm
  N$, and $\gamma_0,\gamma_1$ are positive constants.
\end{condition}
Equation~\eqref{eq:near1} requires correlations to be concave locally,
which is natural since the maximum of $\rho_{\theta}$ is attained at
$\theta$.  Equation~\eqref{eq:near2} is a regularity condition that requires
$\rho_{\theta}^{(0,2)}(\eta)$ to vary smoothly as we change the center
$\theta$.  The normalization quantity
$\normTwo{\vphi^{\,(1)}(\theta)}^2$ captures how sensitive the
features $\vphi$ are to perturbations.  
If this quantity is small for
some $i$ then we require more regularity from $\rho_{\theta}$ because
$\theta$ is harder to distinguish from nearby points using the
measurements $\vphi$.

In the \emph{intermediate} region the correlation function
$\rho_{\theta}$ is bounded but can otherwise fluctuate arbitrarily. In
addition, we require a similar constraint on its derivative with
respect to the position of the center $\theta$. 
  \begin{condition}[Intermediate Condition]
  \label{def:intermediate}
  The correlation function $\rho_{\theta}$ satisfies the intermediate condition if
  \begin{align}\label{eq:inter1}
    |\rho_{\theta}^{(0,0)}(\eta)|&\leq \gamma_2<1\quad\text{and}\\
    |\rho_{\theta}^{(1,0)}(\eta)|&\leq \gamma_3\normTwo{\vphi^{\,(1)}(\theta)}^2 \label{eq:inter2}
  \end{align}
  hold for $\eta<N^-$ and $\eta>N^+$,
  where $N^\pm$ are defined as in the near condition, and $\gamma_2,\gamma_3$ are positive constants.
\end{condition}

\begin{figure}[tp]
  \centering
  \begin{tabular}{@{\vspace{0.75cm}}L@{\quad}CC}
    {\footnotesize Deconvolution} &
    \includegraphics{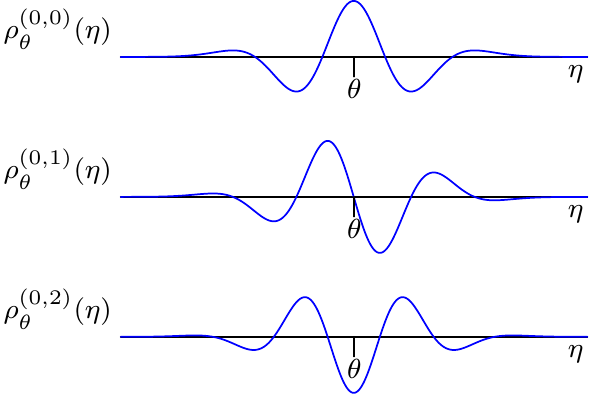}&\includegraphics{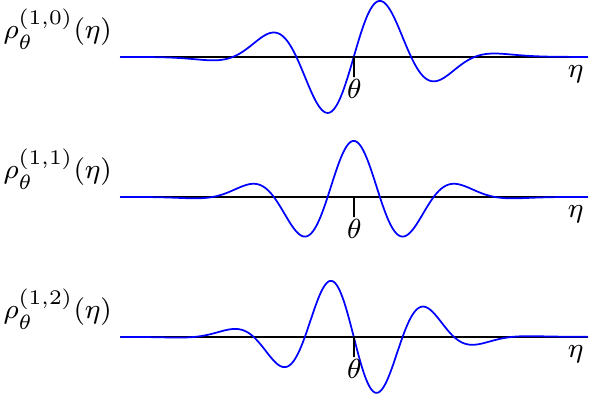}\\ {\footnotesize
      Super-resolution} &
    \includegraphics{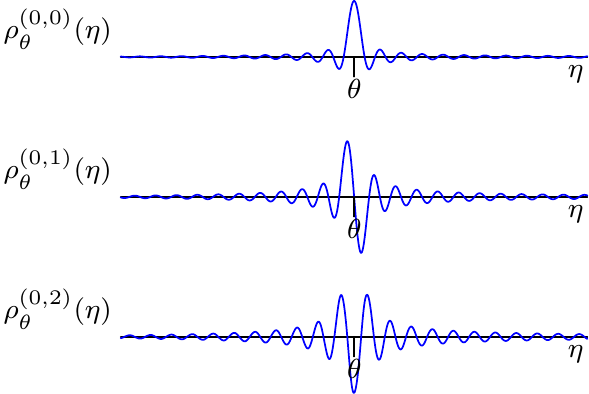}&\includegraphics{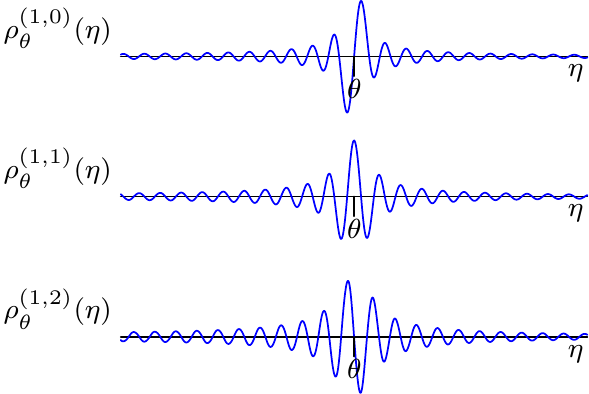}\\ {\footnotesize
      Heat-Source Localization} &
    \includegraphics{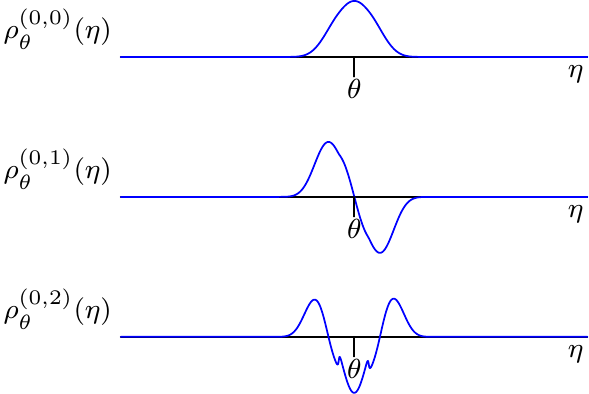}&\includegraphics{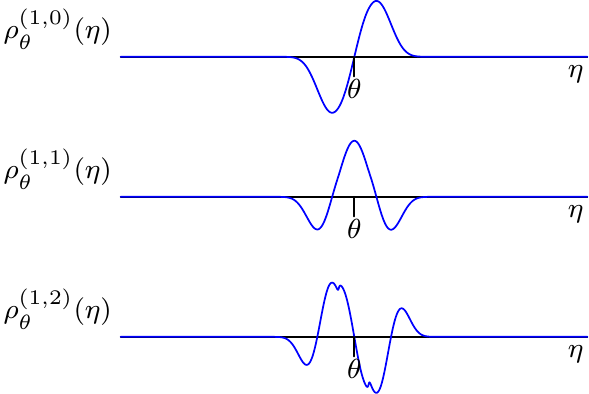}
  \end{tabular}
  \caption{Derivatives of the correlation functions arising in deconvolution, super-resolution, and heat-source localization.
    In all cases the correlation function is locally concave and all derivatives exhibit decaying tails. In the case of super-resolution the tail decay is not
    summable, but can be made summable by applying a window function to the measurements.}
  \label{fig:corrderivs}
\end{figure}
  
In the \emph{decay} region the correlation and its derivatives are bounded by a
  decaying function.
  \begin{condition}[Decay Condition]
  \label{def:decay}
   The correlation function $\rho_{\theta}$ satisfies the decay condition with decay constant $\sigma >0$ if
  \begin{align} \label{eq:decaycond1}
    |\rho_{\theta}^{(0,r)}(\eta)|&\leq C_{0,r}e^{-(|\eta-\theta|-D)/\sigma}\quad\text{and}\\
    |\rho_{\theta}^{(1,r)}(\eta)|&\leq C_{1,r}e^{-(|\eta-\theta|-D)/\sigma}\normTwo{\vphi^{\,(1)}(\theta)}^2\label{eq:decaycond2}
  \end{align}
  hold for $\eta<D^-$ and $\eta>D^+$, where $r=0,1,2$,
  $D^\pm:=\theta \pm D$, and $C_{q,r}$ are positive constants.
\end{condition}
The choice of exponential decay is for concreteness, and can be replaced
by a different summable decay bound\footnote{In the case of
super-resolution, the decay is not summable, but can be made summable
by applying a window to the data, which is standard practice in
spectral
super-resolution~\cite{windows_harris}.}. \Cref{fig:corrderivs} shows
the derivatives of the correlation functions for several SNL problems. 
\begin{figure}[tp]
  \centering
  \begin{tabular}{MM}
    \footnotesize{Dense Uniform}&
    \footnotesize{Irregular}\\
    \includegraphics{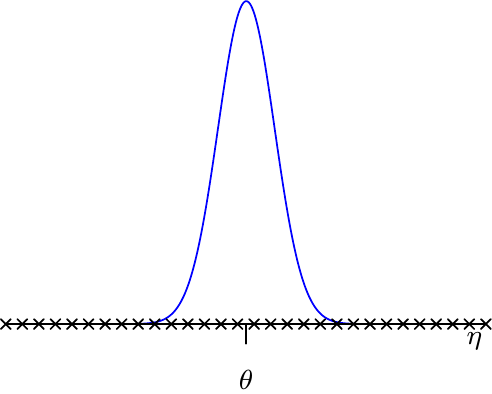}&
    \includegraphics{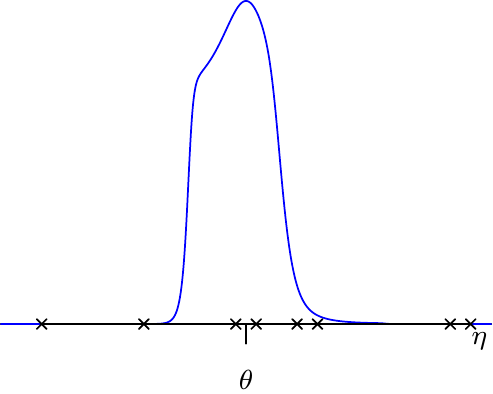}\\
    \footnotesize{Sparse Uniform}&
    \footnotesize{Right Side Only}\\
    \includegraphics{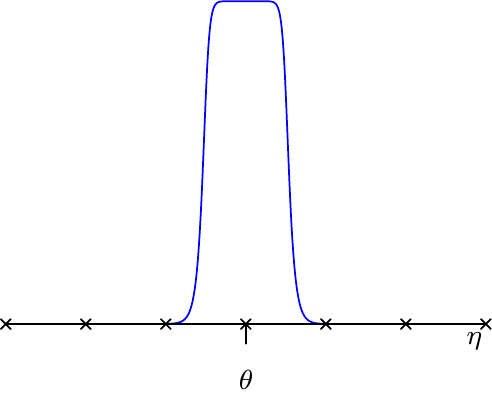}&
    \includegraphics{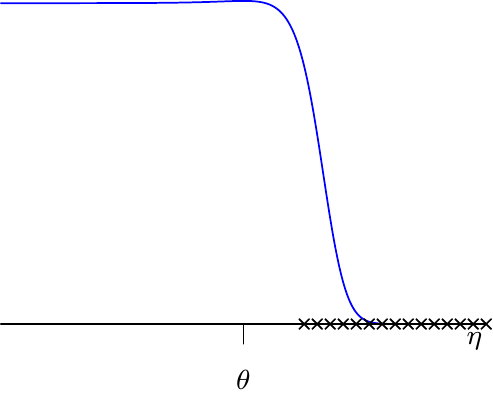}
  \end{tabular}
  \caption{
    \bdb{
    Plots of the correlation function $\rho_\theta(\eta)$ for a
    deconvolution problem with varying sampling patterns
    (indicated by the x markers on the horizontal axis).  Specifically,
    $\vphi(\theta)_i:=\exp(-(\theta-s_i)^2/2)$ are samples from a
    $\theta$-centered Gaussian kernel.
    In the
    dense uniform case, the correlation function satisfies
    \Cref{def:near,def:intermediate,def:decay}.  For irregular sampling pattern, the correlation function will in general satisfy the conditions if there are some samples close to $\theta$.
    For sparse uniform samples, the correlation function is nearly flat at $\theta$ since
    there are too few measurements to distinguish nearby parameters.
    If the samples are only located to the right of
    $\theta$, then all parameters on the left side have nearly
    identical measurements.
    }
  }
  \label{fig:sampchange}
\end{figure}

\bdb{
The sample locations $s_1,\ldots,s_n$ have a direct effect on the correlation function, and hence on whether an SNL problem satisfies
\Cref{def:near,def:intermediate,def:decay}. In \Cref{fig:sampchange}
we depict the correlation function
$\rho_\theta$ for different sampling patterns in a Gaussian deconvolution problem with
$\vphi(\theta)_i:=\exp(-(\theta-s_i)^2/2)$. For both regular and irregular sampling patterns, the correlation functions are well behaved as long as there are samples close to the true parameter. 
In contrast, when there are few samples, or the
samples are distant from the true parameter, the conditions in
\Cref{def:near,def:intermediate,def:decay} may be violated. We refer the interested reader to~\cite{bernstein2017deconvolution} for a more precise analysis of sampling patterns for deconvolution of point sources.}
\subsection{Exact Recovery for SNL Problems with Uniform Correlation Decay}
\label{sec:exactuniform}
In this section we focus on SNL problems where the correlation function $\rho_{\theta}$ of the measurement operator is approximately translation invariant, meaning that $\rho_{\theta}$ has similar properties for any value of $\theta$. Examples of such SNL problems include super-resolution, deconvolution, and heat-source localization if the conductivity is approximately uniform. We prove that TV-norm minimization recovers a superposition of Dirac measures exactly as long as the support satisfies a separation condition related to the decay properties of the correlation function and its derivatives. 

\begin{figure}[t]
  \centering \includegraphics{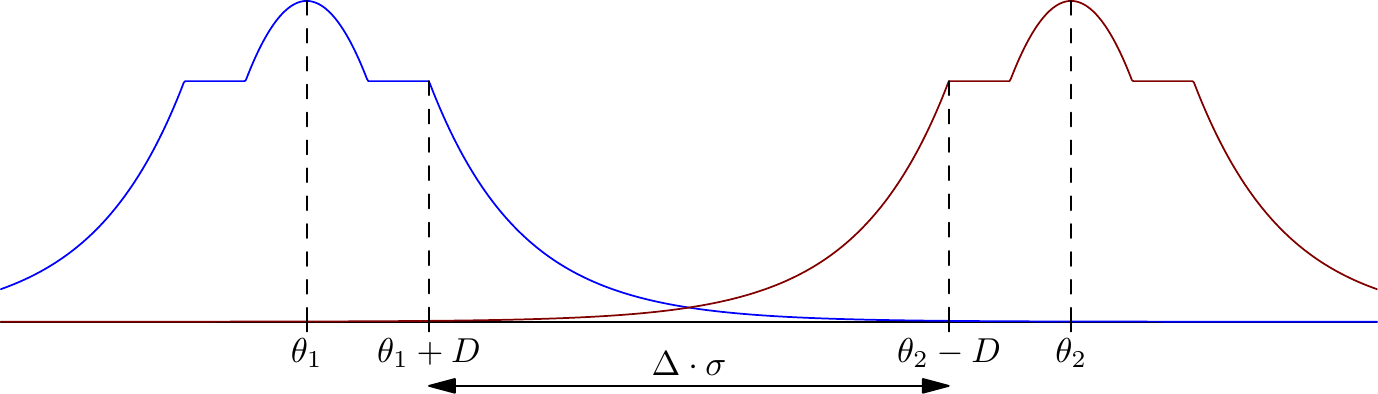}
  \caption{Illustration of the minimum-separation condition required
    by \Cref{thm:exact1d_simp}.
  }
  \label{fig:sep}
\end{figure}

\begin{theorem}
  \label{thm:exact1d_simp}
 Let $\Theta:=\{\theta_1,\ldots,\theta_k\}$ be the support of the
 measure $\mu$ defined in \Cref{eq:x_def}. Assume that the
 correlation functions $\rho_{\theta_i}$, $\theta_i \in \Theta$,
 satisfy \Cref{def:near,def:intermediate,def:decay} for fixed
 constants $N$, $D$, $\sigma$, $\gamma\in\RR^4$, and $C\in\RR^{2\times
   3}$. Then $\mu$  is 
  the unique solution to Problem~\eqref{pr:min_TV} as long as the
  minimum separation of $\Theta$ satisfies, 
  \begin{equation}
    \min_{i\neq j, \; \theta_i,\theta_j \in \Theta} \abs{\theta_i - \theta_j} > 2D + \Delta\sigma, \quad \Delta:= \max\brac{\log(1+2(C_{0,0}+C_{1,1})),\lambda_1,\lambda_2},
  \end{equation}
  where
  \begin{equation}
    \label{eq:DeltaCond2}
    \lambda_1 =
    2\log\brac{\frac{2(2(2C_{0,0}+C_{1,1}-C_{1,1}\gamma_2+C_{0,1}\gamma_3)+1-\gamma_2)}
      {-C_{0,0}+\sqrt{C_{0,0}^2+4(1-\gamma_2)(2(2C_{0,0}+C_{1,1}-C_{1,1}\gamma_2+C_{0,1}\gamma_3)
          +1-\gamma_2)}}},
  \end{equation}
  \begin{equation}
    \label{eq:DeltaCond3}
    \lambda_2 =
    2\log\brac{\frac{2(2((2C_{0,0}+C_{1,1})\gamma_0+C_{0,2}+C_{0,1}\gamma_1)+\gamma_0)}
      {-C_{0,2} + \sqrt{C_{0,2}^2 +
          4\gamma_0(2((2C_{0,0}+C_{1,1})\gamma_0+C_{0,2}+C_{0,1}\gamma_1)+\gamma_0)}}},
  \end{equation}
  and the constants $C_{q,r}$ are chosen so that
  \begin{equation}
\begin{aligned}
  C_{0,1}C_{1,0} &= C_{0,0}C_{1,1},\quad\text{and}\\ C_{0,1}C_{1,2}
  &= C_{1,1}C_{0,2}.
\end{aligned}
\label{eq:algebraic}
\end{equation}
\end{theorem}
Note that the condition in~\Cref{eq:algebraic} is only needed to simplify the statement and proof of our results. 

\Cref{thm:exact1d_simp} establishes that TV minimization recovers the true parameters of an SNL problem when the support separation
is larger than a constant that is proportional to the rate of decay of the correlation function and its derivatives. This separation is measured from the edges of the intermediate regions, as depicted in \Cref{fig:sep}. In stark contrast to compressed-sensing theory, the result holds for correlation functions that are arbitrarily close to one in the near regions, and may have arbitrary bounded fluctuations in the intermediate regions. 

\bdb{Our statement and proof of \Cref{thm:exact1d_simp} focuses on
  clarity over sharpness.  For example, in the Gaussian
  deconvolution problem depicted in \Cref{fig:sampchange} with dense
  uniform samples, the theorem requires a minimum separation of
  approximately $6.6$ to guarantee exact parameter recovery.
  To obtain this bound, we let $N=1$, $D=0$, $\sigma=1$,
  \begin{equation}
    \gamma = \MAT{0.185&0.983&0.788&0.868}^T,\quad
    C = \MAT{2.818&3.348&4.200\\6.786&8.060&10.113},
  \end{equation}
  and verify the conditions numerically.  This is
  in contrast with the minimum separation of approximately $3$ proven
  in \cite{bernstein2017deconvolution}.  Part of this discrepancy
  comes from applying exponential decay bounds
  to a Gaussian-shaped correlation function.
}

Our result requires conditions on the correlation functions centered at the true support $\Theta$, and also on their derivatives. The decay conditions on the derivatives constrain the correlation structure of the measurement operator when we perturb the position of the true parameters. For example, they implicitly bound pairwise correlations centered in a small neighborhood of the support. Exploring to what extent these conditions are necessary is an interesting question for future research.

\begin{figure}[t]
  \centering \includegraphics{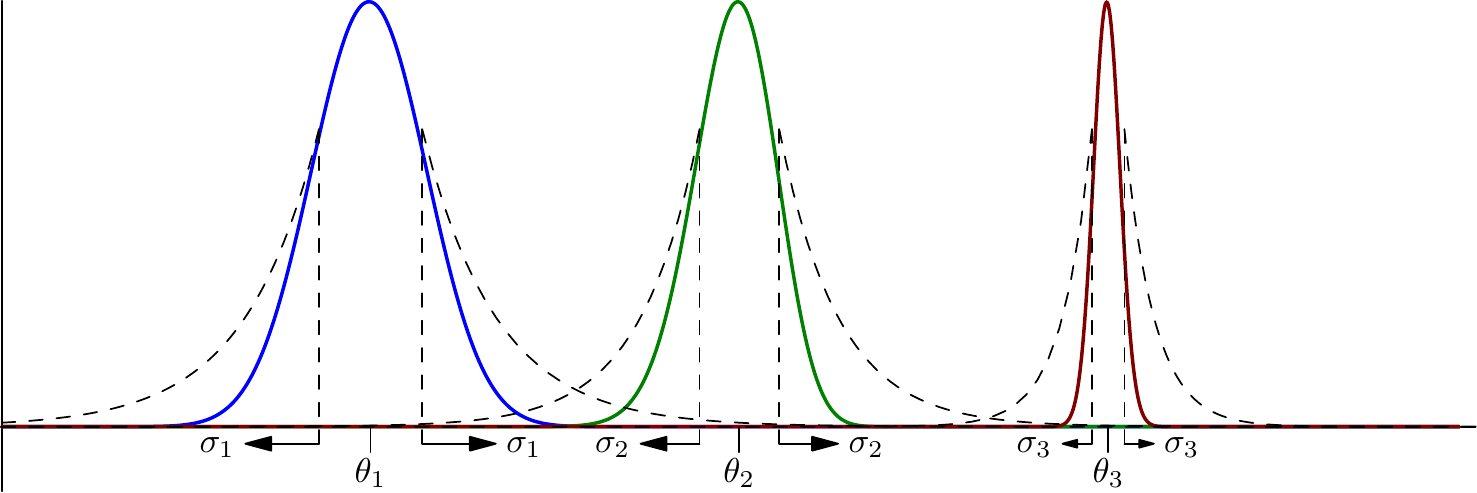}
  \caption{Support-centered correlations with varying decay
    parameters for the heat-source localization problem.
    The vertical dashed lines indicate the locations of
    $D_i^\pm$ for $i=1,2,3$, whereas the curved dashed lines indicate the exponential
    decay bounds.  As we move from left to right along the $\theta$-axis
    the thermal conductivity is decreasing, which causes the
    correlation functions to become narrower.  }
  \label{fig:width}
\end{figure}

\subsection{Exact Recovery for SNL Problems with Nonuniform Correlation Decay}
\label{sec:exactvarying}
The measurement operators associated to many SNL problems of practical interest have nonuniform correlations. \Cref{fig:sccorrs,fig:eegcorr} show that this is the case for heat-source localization with spatially-varying conductivity, and for estimation of brain-activity from EEG data. Our goal in this section is to establish recovery guarantees for such problems.

The conditions on the correlation structure of the measurement
operator required by \Cref{thm:exact1d_simp} only pertain to the
correlation functions centered at the true parameters. In order to
generalize the result, we allow the correlation function centered at
each parameter to satisfy the conditions in
Section~\ref{sec:correlation_decay} with different constants. This
makes it possible for the correlation to have near and intermediate
regions of varying widths around each element of the support, as well
as different decay constants in the far region. Our main result is
that TV-norm minimization achieves exact recovery for SNL problems
with nonuniform correlation structure, as long as the support
satisfies a minimum-separation condition dependent on the
corresponding \emph{support-centered} correlation functions. 

Let $\Theta$ be the support of our signal of interest, and assume
$\rho_{\theta_i}$ satisfies the decay conditions in
Section~\ref{sec:correlation_decay} with parameters $\sigma_i$, $D_i$,
and $N_i$, which are different for all $\theta_i \in \Theta$.
Extending the notation, we let $N_i^{\pm}:=\theta_i\pm N_i$ and
$D_i^{\pm}:=\theta_i\pm D_i$ denote the endpoints of the near and
decay regions, respectively.
Intuitively, when $\sigma_i$ and $D_i$ are small, the
corresponding correlation function $\rho_{\theta_i}$ is ``narrower''
and should require less separation than ``wider'' correlation
functions with large values of $\sigma_i$ and~$D_i$.  This is illustrated
in \Cref{fig:width}, where we depict $\rho_{\theta_i}$ for the
heat-source localization problem at three different values of~$i$.  The decay becomes more pronounced towards the right due to
the changing thermal conductivity of the underlying medium. For the problem to be well posed, one would expect $\theta_1$ to require more separation from other active sources than $\theta_2$, which in turn should require more than $\theta_3$. We confirm this intuition through numerical experiments in \Cref{sec:numexact}. To make it mathematically precise, we define the following generalized notion of support separation.

\begin{definition}[Generalized support Separation]
  \label{def:separation}
  Suppose for all $\theta_i\in\Theta$ that $\rho_{\theta_i}$ satisfies
  \Cref{def:decay} with parameters $D_i$ and $\sigma_i$.
  Define the normalized distance $d(\theta_i,\theta_j)$ for
  $i\neq j$ by  
  \begin{equation}
    d(\theta_i,\theta_j) = \frac{|\theta_i-\theta_j|-D_i-D_j}{\max(\sigma_i,\sigma_j)}.
  \end{equation}
  Assume that $\Theta$ is ordered so that
  $\theta_1<\theta_2<\cdots<\theta_{{k}}$.  $\Theta$
  has separation $\Delta>0$ if 
  $d(\theta_i,\theta_j) > |i-j|\Delta$
  for all $\theta_i,\theta_j\in\Theta$ with $i\neq j$.
\end{definition}
The normalized distance $d(\theta_i,\theta_j)$ is measured between the
edges of the decay regions of $\theta_i$ and~$\theta_j$, and
normalized by the level of decay.  This allows sharply decaying
correlation functions to be in close proximity with one another.  We
require $d(\theta_i,\theta_j) > |i-j|\Delta$ to prevent the parameters
from becoming too clustered.  If we only require the
weaker condition $d(\theta_i,\theta_j) > \Delta$, and if $\sigma_i$
grows very quickly with $i$, then we could have $d(\theta_1,\theta_j)\approx
\Delta$ for all $j>i$.  This causes too much overlap between the
correlation functions.

\begin{figure}[t]
  \centering \includegraphics{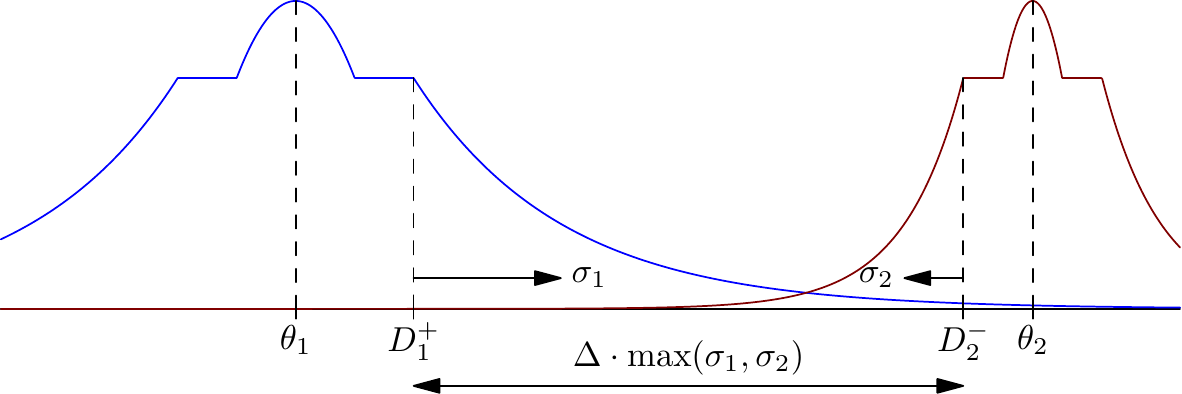}
  \caption{Illustration of the generalized support-separation condition in \Cref{def:separation}.}
  \label{fig:sepgen}
\end{figure}

\begin{figure}[t]
  \centering \includegraphics{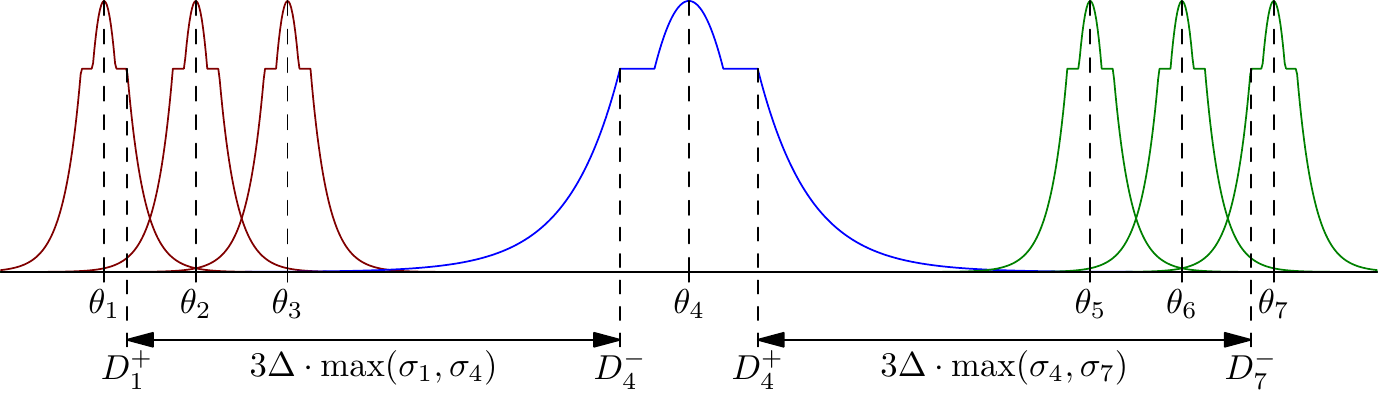}
  \caption{Non-trivial configuration allowed by \Cref{def:separation}.
    The central parameter $\theta_4$ has weaker decay, and thus
    requires more separation from the other parameters.}
  \label{fig:vacuous}
\end{figure} 

\Cref{fig:vacuous} gives an example of parameters and correlation functions that satisfy the conditions
of \Cref{def:separation}. The following theorem establishes exact-recovery guarantees under a condition on the generalized support separation.
  
\begin{theorem}
  \label{thm:exact1d}
  Suppose that for all $\theta_i\in\Theta$ that $\rho_{\theta_i}$ satisfies
  \Cref{def:near,def:intermediate,def:decay} and
  \eqref{eq:algebraic} with constants $N:=N_i$, $D:=D_i$, $\sigma:=\sigma_i$, $C$,
  and $\gamma$.  Note that $C$ and $\gamma$ are the same for each
  $\theta_i$.  Then the
  true measure $\mu$ defined in \eqref{eq:x_def} is the unique
  solution to Problem~\eqref{pr:min_TV} when $\Theta$ has separation
  $\Delta$ (as determined by \Cref{def:separation}) satisfying
  \bdb{
  \begin{equation}
    \Delta> \max\brac{\log(1+2(C_{0,0}+C_{1,1})),\lambda_1,\lambda_2},
    \label{eq:DeltaCond1}
  \end{equation}
  \eqref{eq:DeltaCond2}, \eqref{eq:DeltaCond3}, and \eqref{eq:algebraic}.}
\end{theorem}
The proof of \Cref{thm:exact1d}, which implies \Cref{thm:exact1d_simp}, is given in \Cref{sec:exact1dproof}. The theorem establishes that TV minimization recovers the true parameters of SNL problems with nonuniform correlation decays when the generalized support separation is larger than a constant. Equivalently, exact recovery is achieved as long as each true parameter $\theta_i$ is separated from the rest by a distance that is proportional to the rate of decay of the correlation function centered at $\theta_i$. The separation is measured from the edges of the intermediate regions, which can also vary in width as depicted in \Cref{fig:sep}. The result matches our intuition about SNL problems: the parameters can be recovered as long as they yield measurements that are not highly correlated. As mentioned previously, the theorem requires decay conditions on the derivatives of the correlation function, which constrain the correlation structure of the measurement operator. 

\subsection{Robustness to Noise}
\label{sec:robustness}
In practice, measurements are always corrupted by noisy
perturbations. Noise can be taken into account in our measurement
model \eqref{eq:snl_model} by incorporating an additive noise vector $z\in\RR^n$:
\begin{equation}
  \label{eq:snl_model_noise}
  y:= \sum_{i=1}^kc_i\vphi(\theta_i) +z.
\end{equation}
To adapt the TV-norm minimization problem \eqref{pr:min_TV} to such measurements, we relax the data consistency constraint from an equality to an inequality:
\begin{equation}
\label{pr:min_TV_noise}
  \begin{aligned}
    \underset{\tilde{\mu}}{\op{minimize}} \quad& \normTV{ \tilde{\mu} } \\ 
    \text{subject to} \quad&
    \normTwo{\int_{\R^p} \vec{\phi}(\theta) \tilde{\mu}
      \diffbrac{\theta} - y} \leq \xi,
  \end{aligned}
\end{equation}
where $\xi>0$ is a parameter that must be tuned according to the noise level. Previous works have established robustness guarantees for TV-norm minimization applied to specific SNL problems such as super-resolution~\cite{support_detection,superres_noisy} and deconvolution~\cite{bernstein2017deconvolution} at small noise levels. These proofs are based on dual certificates. Combining the arguments in~\cite{support_detection,bernstein2017deconvolution} with our dual-certificate construction in \Cref{sec:exact1dproof} yields robustness guarantees for general SNL problems in terms of support recovery at high signal-to-noise regimes. We omit the details, since the proof would essentially mimic the ones in~\cite{support_detection,bernstein2017deconvolution}.
\bdb{In \Cref{fig:NoiseExperiment} we show an application of
  \Cref{pr:min_TV_noise} to a noisy deconvolution problem taken from \cite{bernstein2017deconvolution}.}
\begin{figure}[tp]
\centering
\includegraphics{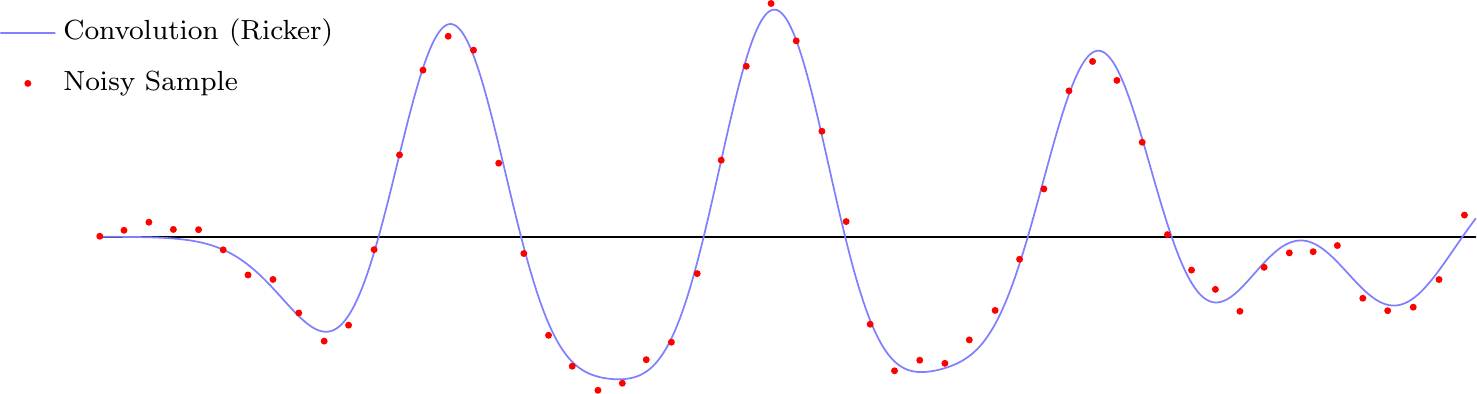}\vspace{.25cm}
\includegraphics{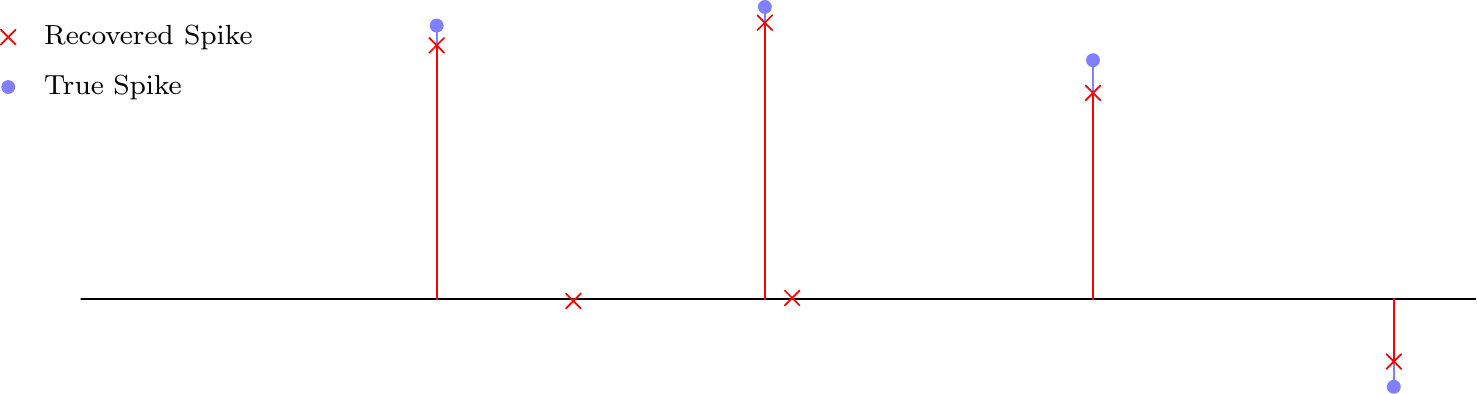}
\caption{Deconvolution from noisy samples (in red) with a
  signal-to-noise ratio of 20.7 dB for the Ricker wavelet.
  The noise is iid Gaussian. The recovered signals are obtained by solving
  \cref{pr:min_TV_noise} on a fine grid which contains the location
  of the original spikes.}
\label{fig:NoiseExperiment}
\end{figure}
\subsection{Discretization}
\label{sec:discretization}
The continuous optimization
problem~\eqref{pr:min_TV_noise} can be solved by applying $\ell_1$-norm minimization after discretizing the parameter space. This is a very popular approach in practice for a variety of SNL problems~\cite{Malioutov:2005jw,taylor1979deconvolution,zhao2011localization,silva2004evaluation,mcmrf}. If the true parameters lie on the discretization grid, then our exact-recovery results translate immediately. The following corollary is a discrete version of \Cref{thm:exact1d}.
\begin{corollary}
  \label{thm:exact1d_l1}
  Assume $\Theta$ lies on a known discretized grid
  $G:=\{\eta_1,\ldots,\eta_m\}$ so that $\Theta\subset G$.
  Furthermore, suppose the conditions of \Cref{thm:exact1d} hold so
  that $\mu$ as defined in \eqref{eq:x_def} is the unique solution to
  Problem~\eqref{pr:min_TV}.  Define the dictionary
  $\Phi\in\RR^{n\times m}$ by
  \begin{equation}
    \Phi := \MAT{\vphi(\eta_1) & \cdots & \vphi(\eta_m)}.
  \end{equation}
  Then the $\ell_1$-norm minimization problem
  \begin{equation}
    \label{pr:min_l1_corr}
    \begin{aligned}
      \underset{\tilde{x}}{\op{minimize}} \quad& \normOne{ \tilde{x} }
      \\\op{subject~to} \quad& \Phi\tilde{x}=y
    \end{aligned}
  \end{equation}
  has a unique solution $x\in\RR^m$ satisfying
  \begin{equation}
    x_j :=
    \begin{cases}
      c_i & \text{if $\eta_j=\theta_i$,}\\ 0 & \text{otherwise,}
    \end{cases}
  \end{equation}
  for $j=1,\ldots,m$.
\end{corollary}
\begin{proof}
  If the support of $\tilde{\mu}$ in Problem~\eqref{pr:min_TV} is
  restricted to lie on $G$, then Problems~\eqref{pr:min_TV} and
  \eqref{pr:min_l1_corr} are the same.  Thus $\normTV{\mu}$ must be
  smaller than $\normOne{\tilde{x}}$ for any $\tilde{x}$ such that $\Phi\tilde{x}=y$.  By
  assumption $\mu$ is supported on $G$, so the result follows.
\end{proof}

Of course, the true parameters may not lie on the grid used to solve the $\ell_1$-norm minimization problem. The proof techniques used to derive robustness guarantees for super-resolution and deconvolution in~\cite{support_detection,bernstein2017deconvolution} can be leveraged to provide some control over the discretization error. Performing a more accurate analysis of discretization error for SNL problems is an interesting direction for future research.

\subsection{Related Work}
\label{sec:rw}
\subsubsection{Sparse Recovery via Convex Programming from Deterministic Measurements}
\label{sec:rwinco}

In \cite{Donoho2197}, Donoho and Elad develop a theory of sparse recovery from generic measurements based on the \textit{spark} and \textit{coherence} of the measurement matrix $\Phi$.
The spark is defined to be the smallest positive value $s$ such that $\Phi$ has $s$ linearly dependent columns.  
The coherence,
which we denote by $M(\Phi)$, is the maximum absolute correlation between any
two columns of $\Phi$. 
The authors show that exact recovery is achieved by
$\ell_1$-minimization when the number of true parameters is less than
$(1+1/M(\Phi))/2$.  
As discussed in
\Cref{sec:cs}, these arguments are inapplicable to the
finely-discretized parameter spaces occurring in SNL problems since neighboring
columns of $\Phi$ have correlations approaching one.
In \cite{dossal2010numerical}, the authors provide a support-dependent condition for exact
recovery of discrete vectors.  Using our notation, they require that $\beta(\Theta)/(1-\alpha(\Theta))<1$
where
\begin{equation}
  \alpha(\Theta) := \max_{\theta_i\in \Theta}
  \sum_{\theta_k\in\Theta,k\neq i}|\vphi(\theta_i)^T\vphi(\theta_k)|
  \quad\text{and}\quad
  \beta(\Theta) := \max_{\eta_j\notin \Theta}\sum_{\theta_k\in \Theta}|\vphi(\eta_j)^T\vphi(\theta_j)|.
\end{equation}
Here $\vphi(\eta_j)$ for $\eta_j\notin\Theta$ ranges over the columns of $\Phi$ that do not
correspond to the true parameters.  This condition is also inapplicable for matrices arising in our problems of interest because $\beta(\Theta)$ approaches one (or
larger) for finely-discretized parameter spaces. Sharper exact-recovery guarantees in subsequent works~\cite{chandrasekaran2012convex,donoho2006compressed,needell2009cosamp,candes2005decoding,bickel2009simultaneous,candes2008restricted} require randomized measurements, and therefore do not hold for deterministic SNL problems as explained in \Cref{sec:cs}.

\subsubsection{Convex Programming Applied to Specific SNL Problems}
\label{sec:rwsnltheory}
In \cite{superres,superres_new}, the authors establish exact recovery guarantees for super-resolution via convex optimization by leveraging parameter separation (see \Cref{sec:separation}). Subsequent works build upon these results to study the robustness of this methodology to noise~\cite{tang2014robust,support_detection,superres_noisy,peyreduval}, missing data~\cite{tang2012offgrid}, and outliers~\cite{fernandez2017demixing}.
A similar analysis is carried out in
\cite{bernstein2017deconvolution} for deconvolution. The authors
establish a sampling theorem for Gaussian and Ricker-wavelet
convolution kernels, which characterizes what sampling patterns yield
exact recovery under a minimum-separation condition. Other works have
analyzed the Gaussian deconvolution problem under nonnegativity
constraints~\cite{schiebinger2015superresolution,eftekhari2018sparse},
and also for randomized measurements~\cite{poon2018dual}. All of these
works exploit the properties of specific measurement operators. In
contrast, the present paper establishes a general theory that only
relies on the correlation structure of the measurements. The works
that are closer to this spirit
are~\cite{bendory2016robust,tang2013atomic}, which analyze
deconvolution via convex programming for generic convolution
kernels. The results in \cite{bendory2016robust} require quadratically
decaying bounds on the first three derivatives of the convolution
kernel. In \cite{tang2013atomic}, the authors prove exact recovery
assuming bounds on the first four derivatives of the autocorrelation
function of the convolution kernel. In contrast to these works, our
results allow for discrete irregular sampling and for measurement
operators that are not convolutional, which is necessary to analyze applications such as heat-source localization or estimation of brain
activity. 

\bdb{
  Algorithms for solving the convex programs arising from SNL problems
  divide into two categories: grid-based and grid-free (or
  off-the-grid). In grid-based
  methods the parameter space is discretized, thus yielding a
  finite-dimensional $\ell_1$ minimization problem that can be solved
  using standard methods (as we have done in \Cref{sec:numerical}).
  Grid-free methods attempt to directly solve the infinite-dimensional,
  but with weaker guarantees (see \cite{hettich1993semi} and
  \cite{lopez2007semi} for review articles).  Recently, several authors
  \cite{denoyelle2019sliding,eftekhari2019sparse,boyd2017alternating} have
  proposed grid-free algorithms based on the conditional gradient method
  \cite{frank1956algorithm}.
}
\subsubsection{Other Methodologies}
SNL parameter recovery can be formulated as a nonlinear least squares
problem \cite{golub2003separable}. For a fixed value of the parameters $\theta_1,\ldots,\theta_k$, the optimal coefficients $c_1,\ldots,c_k$ in \eqref{eq:snl_cts} have a closed form solution. This makes it possible to minimize the nonlinear cost function with respect to $\theta_1,\ldots,\theta_k$ directly, a technique known as variable projection~\cite{golub1973differentiation}. As shown in \Cref{fig:varproj}, a downside to this approach is that it may converge to suboptimal local minima, even in the absence of noise.


Prony's method~\cite{deProny:tg,Stoica:2005wf} and the
finite-rate-of-innovation (FRI)
framework~\cite{vetterli2002sampling,dragotti2007sampling}
can be applied to tackle
SNL problems, as long one can recast them as spectral super-resolution
problems.  This provides a recovery method that avoids discretizing
the parameter space.  
The FRI framework has also been applied to arbitrary
non-bandlimited convolution kernels~\cite{uriguen2013fri} and
nonuniform sampling patterns~\cite{pan2016towards}, but
without exact-recovery guarantees.
These techniques have been recently extended by Dragotti and Murray-Bruce
\cite{murray2017universal} to physics-driven SNL problems.
By approximating complex exponentials with weighted sums of Green's functions,
they are able to recast parameter recovery as a related spectral
super-resolution problem that approximates the true SNL problem.

\section{Proof of Exact-Recovery Results}
\label{sec:exact1dproof}
\label{sec:RecoveryCert}
\subsection{Dual Certificates}
\label{sec:DualCert}
To prove \Cref{thm:exact1d} we construct a
\textit{certificate} that guarantees exact recovery.
\begin{proposition}[Proof in \Cref{sec:dualcertproof}]
  \label{thm:dualcert}
  Let $\Theta=\{\theta_1,\ldots,\theta_{k}\}\subset\RR$
  denote the support of the signed
  atomic measure $\mu$.  Assume that for any 
  sign pattern $\xi\in\{\pm1\}^{k}$ there is a $\tilde{q}\in\RR^p$
  such that $\widetilde{Q}(\theta):=\tilde{q}^T\vphi(\theta)$
  satisfies
  \begin{align}
    \widetilde{Q}(\theta_i)=\xi_i,&\quad\quad\forall\theta_i\in \Theta,\label{eq:certinterp}\\
    |\widetilde{Q}(\theta)|<1,&\quad\quad\forall \theta\in\Theta^c.\label{eq:certbound}
  \end{align}
  Then $\mu$ is the
  unique solution to problem~\eqref{pr:min_TV}.
\end{proposition}
To prove exact recovery of a signal we need to show that
it is possible to interpolate the sign pattern of its amplitudes, which we denote by $\xi$, on its support
$\Theta$ using an interpolating function $\widetilde{Q}(\theta)$ that is expressible as a
linear combination of the coordinates of $\vphi(\theta)$.
The coefficient vector of this linear combination, denoted by
$\tilde{q}$,  is known as a \textit{dual certificate} in
the literature because it certifies recovery and is a solution to the
Lagrange dual of problem~\eqref{pr:min_TV}:
\begin{equation}
  \label{pr:TVNormDual}
\begin{array}{ll}
  \underset{ \nu }{\text{maximize}} & \nu^Ty\\
  \text{subject to} & \sup_\theta\left|\nu^T\vphi(\theta)\right|\leq1.
\end{array}
\end{equation}
Dual certificates have been widely used to derive guarantees for inverse
problems involving random measurements, including compressed
sensing~\cite{Candes:2005cs,candes2011probabilistic}, matrix
completion~\cite{candes2012exact} and phase
retrieval~\cite{candes2015phase}. In such cases, the certificates are usually 
constructed by leveraging concentration bounds and other tools from probability
theory~\cite{vershynin2010introduction}. In contrast, our setting is
completely deterministic.  More recently, dual certificates have been
proposed for specific deterministic problems such as super-resolution~\cite{superres} and
deconvolution~\cite{bernstein2017deconvolution}. Our goal here is to provide a construction that is valid for generic SNL models with correlation decay.

\subsection{Correlation-Based Dual Certificates}
\label{sec:CorrelationCert}
Our main technical contribution is a certificate that only depends on the correlation function of the measurement operator. In contrast, certificate constructions in previous works on SNL problems~\cite{superres,bernstein2017deconvolution,bendory2016robust}
typically rely on problem-specific structure, with the exception of~\cite{tang2013atomic} which proposes a certificate for time-invariant problems based on autocorrelation functions.
 
In our SNL problems of interest the function $\vphi(\theta)$ mapping the parameters of interest to the data is assumed to be continuous and smooth. As a result, \Cref{eq:certinterp,eq:certbound} imply that any valid interpolating function $\widetilde{Q}$ reaches a local extremum at each $\theta_i\in\Theta$. Equivalently, $\widetilde{Q}$ satisfies the following $2k$ \textit{interpolation equations} for all $\theta_i\in\Theta$:
\begin{align}
  \widetilde{Q}(\theta_i) & = \xi_i,\label{eq:interp1}\\
  \widetilde{Q}^{(1)}(\theta_i) &= 0.\label{eq:interp2}
\end{align}
Inspired by this observation, we define 
\begin{equation}
  q := \sum_{i=1}^{k}\alpha_i\vphi(\theta_i) +
  \beta_i\frac{\vphi^{(1)}(\theta_i)}{\normTwo{\vphi^{(1)}(\theta_i)}^2}
  \label{eq:qcert}
\end{equation}
where $\alpha_i,\beta_i\in\RR$, $i=1,\ldots,k$, are chosen so that $Q(\theta):=q^T\vphi(\theta)$ satisfies the interpolation equations. Crucially, this choice of coefficients yields an interpolation function $Q$ that is a linear combination of correlation functions centered at $\Theta$,
\begin{align}
  Q(\theta) & = \sum_{i=1}^k \alpha_i\vphi(\theta_i)^T\vphi(\theta) +
  \beta_i\frac{\vphi^{(1)}(\theta_i)^T\vphi(\theta)}{\normTwo{\vphi^{(1)}(\theta_i)}^2}\\
  & = \sum_{i=1}^k \alpha_i\rho_{\theta_i}(\theta) + 
  \beta_i\frac{\rho^{(1,0)}_{\theta_i}(\theta)}{\rho_{\theta_i}^{(1,1)}(\theta_i)}.\label{eq:Qcorr}
\end{align}

In essence, we interpolate the sign pattern of the signal on its support using support-centered correlations. Since $\rho_{\theta_i}(\theta_i)=1$ and
$\rho_{\theta_i}^{(1,0)}(\theta_i)=0$, \Cref{eq:interp1,eq:Qcorr} imply
$\alpha_i\approx\xi_i$ when $\rho(\theta,\eta)$ and
$\rho^{(1,0)}(\theta,\eta)$ decay as $|\theta-\eta|$ grows large and
 the support is sufficiently separated.  The term that depends on $\beta$ can be interpreted as a correction to the derivatives of $Q$ so that \eqref{eq:interp2} is satisfied.
The normalizing factor used in \eqref{eq:Qcorr} makes this explicit:
\begin{align}
\left.\frac{\partial}{\partial\theta}\right|_{\theta=\theta_i}
\beta_i\frac{\rho^{(1,0)}_{\theta_i}(\theta)}{\rho_{\theta_i}^{(1,1)}(\theta_i)}
= \beta_i
\end{align}
for $i=1,\ldots,k$. \Cref{fig:dualcert} illustrates the construction for the heat-source localization problem.
The construction is inspired by previous interpolation-based certificates tailored to super-resolution \cite{superres} and
deconvolution \cite{bernstein2017deconvolution}.

\begin{figure}[t]
  \centering
  \includegraphics{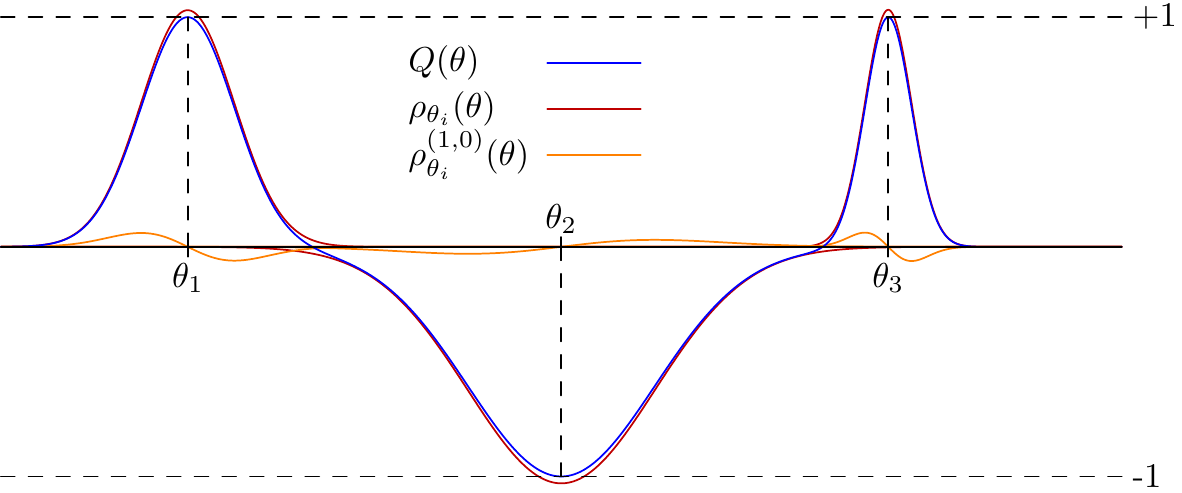}
  \caption{The image shows the interpolating function $Q(\theta)$ defined in \Cref{eq:Qcorr} for the heat-source localization problem. The function is a linear combination
    of $\rho_{\theta_i}(\theta)$ (red curves) and
    $\rho_{\theta_i}^{(1,0)}(\theta)$ (orange curves) for $\theta_i \in \Theta$.}
  \label{fig:dualcert}
\end{figure}

In the remainder of this section we show that our proposed construction yields a valid certificate if the conditions of
\Cref{thm:exact1d} hold. In \Cref{sec:invertproof} we
prove \Cref{thm:invert}, which establishes that the interpolation
equations have a unique solution and therefore $Q$ satisfies \eqref{eq:certinterp}.
\begin{lemma}
  \label{thm:invert}
  Under the assumptions of \Cref{thm:exact1d} there exist
  $\alpha,\beta\in\RR^{k}$ which uniquely solve \Cref{eq:interp1,eq:interp2}.
\end{lemma}
In \Cref{sec:qboundproof} we prove \Cref{thm:qbound}, which establishes that
$Q$ satisfies \eqref{eq:certbound}.
This completes the proof of \Cref{thm:exact1d}.
\begin{lemma}
  \label{thm:qbound}
  Let $\alpha,\beta\in\RR^{k}$ be the coefficients obtained in \Cref{thm:invert}.
  Under the assumptions of \Cref{thm:exact1d} , the interpolating function $Q$ defined in \eqref{eq:Qcorr} satisfies
  $|Q(\theta)|<1$ for $\theta\notin\Theta$.  
\end{lemma}

\subsection{Proof of \Cref{thm:invert}}
\label{sec:invertproof}
To simplify notation we define the $i$th normalized
correlation and its derivatives by 
\begin{equation}
  \hrho_{\theta_i}^{(q,r)}(\theta) := \frac{\rho_{\theta_i}^{(q,r)}(\theta)}{\normTwo{\vphi^{(q)}(\theta_i)}^2},
\end{equation}
for $q=0,1$ and $r=0,1,2$, where \bdb{$\rho_{\theta_i}^{(q,r)}$} is defined in \Cref{eq:def_rho}.  Using this notation we have
\begin{align}
  Q(\theta)
  & = \sum_{i=1}^k \alpha_i\rho_{\theta_i}(\theta) + 
  \beta_i\frac{\rho^{(1,0)}_{\theta_i}(\theta)}{\rho_{\theta_i}^{(1,1)}(\theta_i)}\\
  & = \sum_{i=1}^k \alpha_i\hrho_{\theta_i}(\theta)+
  \beta_i\hrho_{\theta_i}^{(1,0)}(\theta).
\end{align}

To prove \Cref{thm:invert} we establish the following stronger
result, which also gives useful bounds on $\alpha$ and $\beta$.
Throughout we assume that $\rho$ satisfies
\Cref{def:near,def:intermediate,def:decay} and \eqref{eq:algebraic}
with parameters $\gamma$, $C$, $N$, $D$, and $\sigma$,
and that $\Theta$ has separation $\Delta$.
\begin{lemma}
  \label{thm:schurCDelta}
  Let $s:=2e^{-\Delta}/(1-e^{-\Delta})$.
  If $\Delta>\log(1+2(C_{0,0}+C_{1,1}))$ then there are
  $\alpha,\beta\in\RR^{k}$ which uniquely solve equations \eqref{eq:interp1} and
  \eqref{eq:interp2}. Furthermore, we have
  \begin{equation*}
    \normInf{\alpha}\leq \frac{1-C_{1,1}s}{1-(C_{0,0}+C_{1,1})s},\quad
    \normInf{\beta}\leq \frac{C_{0,1}s}{1-(C_{0,0}+C_{1,1})s},\quad\text{and}\quad
    1-\normInf{\alpha-\xi}\geq \frac{1-(2C_{0,0}+C_{1,1})s}{1-(C_{0,0}+C_{1,1})s}.
  \end{equation*}
\end{lemma}
To prove \Cref{thm:schurCDelta} we begin by rewriting Equations
\eqref{eq:interp1} and \eqref{eq:interp2} in block matrix-vector form for the $Q$ function defined in \Cref{eq:Qcorr}:
\begin{equation}
  \MAT{\II+P^{(0,0)}&P^{(1,0)}\\P^{(0,1)}&\II+P^{(1,1)}}\MAT{\alpha\\\beta}=\MAT{\xi\\0},\label{eq:interpmat}
\end{equation}
where $\II\in\RR^{{k}\times{k}}$ is the identity matrix, and
$P^{(q,r)}\in\RR^{{k}\times{k}}$ satisfies
\begin{equation}
  P^{(q,r)}_{ij} = \hat{\rho}_{\theta_j}^{(q,r)}(\theta_i)
\end{equation}
for $i\neq j$ and $P^{(q,r)}_{ii}=0$.  To see why $P_{ii}^{(1,0)}=0$
agrees with equations \eqref{eq:interp1} and \eqref{eq:interp2}, note that
\begin{equation}
  \hrho_{\theta_i}^{(1,0)}(\theta_i) =
  \frac{\vphi^{(1)}(\theta_i)^T\vphi(\theta_i)}{\normTwo{\vphi^{(1)}(\theta_i)}^2}
  = 0.
\end{equation}
The atoms are normalized-- 
$\normTwo{\vphi(\theta)}=1$ for all $\theta$-- which implies
$$0=\frac{d}{d\theta}\normTwo{\vphi(\theta)}^2 = 2\vphi^{(1)}(\theta)^T\vphi(\theta).$$
For the same reason it follows that $\hrho_{\theta_i}^{(0,1)}(\theta_i)=0$.

Our plan is to
bound the norm of each $P^{(q,r)}$.  If these norms are sufficiently
small then the matrix in \Cref{eq:interpmat} is nearly the identity, and the desired
result follows from a linear-algebraic argument.  Define $\epsilon_i^{(q,r)}(\theta)$ by
\begin{equation}
\epsilon_i^{(q,r)}(\theta) := \sum_{j\neq
  i}|\hat{\rho}_{\theta_j}^{(q,r)}(\theta)|,
\end{equation}
where $q=0,1$ and $r=0,1,2$.  
Here we think of $\theta$ as a point close to
$\theta_i$, so $\epsilon_i^{(q,r)}(\theta)$ captures the
cumulative correlation from the other, more distant elements of $\Theta$.  We expect
$\epsilon_i^{(q,r)}(\theta)$ to be small when $\rho$ has sufficient decay
and $\Theta$ is well separated.  For a matrix $A$ let $\normInf{A}$
denote the infinity-norm defined by
\begin{equation}
  \normInf{A} := \sup_{\normInf{x}=1}\normInf{Ax}.
\end{equation}
$\normInf{A}$ equals the maximum sum of absolute values in any row of $A$. We have
\begin{equation}
  \|P^{(q,r)}\|_\infty =\epsilon^{(q,r)} :=
  \max_{i\in\{1,\ldots,{k}\}}\epsilon_i^{(q,r)}(\theta_i).
\end{equation}
The following lemma shows that equation
\eqref{eq:interpmat} is invertible when $\epsilon^{(q,r)}$ is
sufficiently bounded.
\begin{lemma}[Proof in \Cref{sec:schurboundproof}]
  \label{thm:schurbound}
  Suppose 
  \begin{align}
    \epsilon^{(1,1)}&<1\quad\text{and}\label{eq:schurcond1}\\
    c:=\epsilon^{(0,0)}+\frac{\epsilon^{(1,0)}\epsilon^{(0,1)}}{1-\epsilon^{(1,1)}}&<1.
    \label{eq:schurcond2}
  \end{align}
  Then the matrix in \eqref{eq:interpmat} is invertible and
  \begin{align}
    \|\alpha\|_\infty
    & \leq
    \frac{1}{1-c}\\
    \|\beta\|_\infty
    & \leq
    \frac{\epsilon^{(0,1)}/(1-\epsilon^{(1,1)})}{1-c}\\
    \|\alpha-\xi\|_\infty &\leq \frac{c}{1-c}.
  \end{align}
\end{lemma}
To apply \Cref{thm:schurbound} we must first bound $\epsilon^{(q,r)}$
for $q,r\in\{0,1\}$.  By the decay and separation conditions
(\Cref{def:decay} and \Cref{def:separation}), $\theta_j$ lies in the exponentially
decaying tail of $\rho_{\theta_i}$, for $i\neq j$.  This gives
\begin{align}
  \epsilon_i^{(q,r)}(\theta_i)
  & = \sum_{j\neq i}|\hat{\rho}_{\theta_j}^{(q,r)}(\theta_i)|\\
  & \leq
  C_{q,r}\sum_{\theta_j>\theta_i}\exp\brac{-d(\theta_i,\theta_j)}
    + C_{q,r}\sum_{\theta_j<\theta_i}\exp\brac{-d(\theta_i,\theta_j)}\\
  & \leq 2C_{q,r}\sum_{k=1}^\infty \exp\brac{-k\Delta}\\
    & = \frac{2C_{q,r}e^{-\Delta}}{1-e^{-\Delta}} = C_{q,r}s,
\end{align}
where $s:=2e^{-\Delta}/(1-e^{-\Delta})$ and the distance $d$
is defined in \Cref{def:separation}.
As this bound is independent of $i$, we have
\begin{equation}
  \label{eq:epsbound}
  \epsilon^{(q,r)} \leq C_{q,r}s
\end{equation}
for $q,r\in\{0,1\}$.  In terms of the
conditions of \Cref{thm:schurbound}, we obtain
\begin{equation}
c = \epsilon^{(0,0)}+\frac{\epsilon^{(1,0)}\epsilon^{(0,1)}}{1-\epsilon^{(1,1)}}
\leq C_{0,0}s + \frac{C_{1,0}C_{0,1}s^2}{1-C_{1,1}s}
= \frac{C_{0,0}s+(C_{1,0}C_{0,1}-C_{0,0}C_{1,1})s^2}{1-C_{1,1}s}
= \frac{C_{0,0}s}{1-C_{1,1}s},
\label{eq:cbound}
\end{equation}
assuming \Cref{eq:algebraic} ($C_{1,0}C_{0,1}=C_{0,0}C_{1,1}$) holds.
Thus we have $c<1$, as required by \eqref{eq:schurcond2}, when $s<\frac{1}{C_{0,0}+C_{1,1}}$.
Using this, we can find conditions on $\Delta$ so that the
hypotheses of \Cref{thm:schurbound} hold.
If $\Delta>\log(1+\kappa)$
for some $\kappa> 0$ then, using that $f(x)=e^{-x}/(1-e^{-x})$ is decreasing
for $x>0$, we have
\begin{equation}
  \label{eq:delta}
  \frac{s}{2} = \frac{e^{-\Delta}}{1-e^{-\Delta}} <
  \frac{1/(1+\kappa)}{1-1/(1+\kappa)} =
  \frac{1}{\kappa}.
\end{equation}
Letting $\kappa = 2C_{1,1}$ this shows that $\Delta>\log(1+2C_{1,1})$ implies
\begin{equation}
  \epsilon^{(1,1)} \leq C_{1,1}s = 2C_{1,1}\frac{s}{2} < 2C_{1,1}\frac{1}{2C_{1,1}}=1,
\end{equation}
giving \eqref{eq:schurcond1}.
For condition \eqref{eq:schurcond2} we let $\kappa=2(C_{0,0}+C_{1,1})$.
Then \eqref{eq:delta} shows that
$\Delta>\log(1+2(C_{0,0}+C_{1,1}))$ implies
\begin{equation}
  s = 2\frac{s}{2} < 2\frac{1}{2(C_{0,0}+C_{1,1})} = \frac{1}{C_{0,0}+C_{1,1}}
\end{equation}
as required.  Thus when $\Delta>\log(1+2(C_{0,0}+C_{1,1}))$ both conditions of
\Cref{thm:schurbound} hold.  By plugging \eqref{eq:epsbound} and
\eqref{eq:cbound} into the
bounds of \Cref{thm:schurbound} we obtain \Cref{thm:schurCDelta}.
This completes the proof of \Cref{thm:invert}.

\subsection{Proof of \Cref{thm:qbound}}
\label{sec:qboundproof}
\Cref{thm:invert} implies that we can solve \eqref{eq:interp1} and
\eqref{eq:interp2} for $\alpha$ and $\beta$ and then obtain $Q$ via \eqref{eq:Qcorr}.
To prove \Cref{thm:qbound} we must show that 
$|Q(\theta)|<1$ for $\theta\in\Theta^c$.  This is accomplished in two
steps.  We first show
that $|Q(\theta)|<1$ for $\theta$ that are not in the near region of any
$\rho_{\theta_i}$, $i=1,\ldots,k$.
Secondly, we show that for any $\theta$ in the near region of some $\rho_{\theta_i}$
we have $Q^{(2)}(\theta)<0$, where we assume $\xi_i=1$ without loss of
generality.  This proves that $Q$ has a local maximum at $\theta_i$,
and is smaller than one nearby. 

To bound $Q$ and $Q^{(2)}$ we apply the triangle inequality to \eqref{eq:Qcorr}, obtaining the following lemma.
\begin{lemma}
  \label{thm:dualcertbound}
  Fix $\theta_i\in\Theta$ and let $Q$ be defined as in
  \eqref{eq:Qcorr}. For all $\theta\in\RR$
  \begin{equation}
    \label{eq:qbound}
    |Q(\theta)|
    \leq
    \|\alpha\|_\infty|\hat{\rho}_{\theta_i}(\theta)|+\|\beta\|_\infty|\hat{\rho}_{\theta_i}^{(1,0)}(\theta)|
    + \|\alpha\|_\infty
    \epsilon_i^{(0,0)}(\theta)+\|\beta\|_\infty\epsilon_i^{(1,0)}(\theta).
  \end{equation}
  If $\hat{\rho}_{\theta_i}^{(0,2)}(\theta)\leq0$ and $\xi\in\{-1,+1\}^k$ is the
  sign pattern interpolated by $Q$ then
  \begin{equation}
    \label{eq:ddqbound}
    Q^{(2)}(\theta)\leq
    (1-\normInf{\alpha-\xi})\hat{\rho}_{\theta_i}^{(0,2)}(\theta)+\|\beta\|_\infty|\hat{\rho}_{\theta_i}^{(1,2)}(\theta)|
    + \|\alpha\|_\infty
    \epsilon_i^{(0,2)}(\theta)+\|\beta\|_\infty\epsilon_i^{(1,2)}(\theta).
  \end{equation}
\end{lemma}
In the next lemma we show that the separation conditions in
\Cref{def:separation} ensure that the support does not cluster together.
We assume that $\theta_1<\theta_2<\cdots<\theta_{|\Theta|}$.
\begin{lemma}[Proof in \Cref{sec:clusterproof}]
  \label{thm:cluster}
  Fix $\theta_i\in\Theta$ and let $\theta>\theta_i$.  If
  $i \leq k -1$, assume that 
  \begin{equation}
    \label{eq:closer}
    \frac{\theta - D_i^+}{\sigma_i} \leq \frac{D_{i+1}^--\theta}{\sigma_{i+1}}.
  \end{equation}
  Then
  \begin{alignat}{2}
    \frac{D_{j}^--\theta}{\sigma_{j}} &\geq \frac{\Delta}{2} &&\quad\text{if
      $j=i+1$,}\label{eq:cluster1}\\
    \frac{D_{j}^--\theta}{\sigma_{j}} &\geq \Delta(j-(i+1)) &&\quad\text{if
      $j>i+1$,}\label{eq:cluster2}\\
    \frac{\theta-D_j^+}{\sigma_{j}} &\geq \Delta(i-j) &&\quad\text{if
      $j<i$,}\label{eq:cluster3}
  \end{alignat}
  where $\Delta$ is defined in \Cref{def:separation}, as long as 
  \Cref{def:near,def:intermediate,def:decay} and \eqref{eq:algebraic}
  hold.
\end{lemma}
Inequality \eqref{eq:closer} implies that $\theta_i$ is the closest element of
$\Theta$ to $\theta$, if we use the generalized distance normalized by $\sigma$.\\
\subsection*{Bounding $|Q(\theta)|$ Outside the Near Region}
Our goal is to bound $|Q(\theta)|$ for $\theta\in\RR$ that are not in
the near region of any $\rho_{\theta_i}$.  We can assume,
flipping the axis if necessary, 
that the conditions of \Cref{thm:cluster} hold for the $\theta_i$
closest to $\theta$ and that $\theta\geq
N_i^+$  (recall that $N_i^+=\theta_i+N_i$ is the
boundary of the near region of $\rho_{\theta_i}$).  
Then $\theta$ lies in the exponentially decaying tail of $\rho_{\theta_j}$ for
$j\neq i$, so that
\begin{align}
  \epsilon_i^{(q,r)}(\theta)
  &= \sum_{j\neq i}\hat{\rho}_{\theta_j}^{(q,r)}(\theta)\\
  & \leq C_{q,r}\exp(-(D_{i+1}^--\theta)/\sigma_{i+1})\nonumber\\
  &\quad+
  \sum_{j<i}C_{q,r}\exp(-(\theta - D_j^+)/\sigma_j)+
  \sum_{j>i+1}C_{q,r}\exp(-(D_j^--\theta)/\sigma_j)\\
  & \leq C_{q,r}e^{-\Delta/2}+2C_{q,r}\sum_{k=1}^\infty e^{-k\Delta}\\
  & = C_{q,r}(x+s), \label{eq:eps_bound}
\end{align}
where $x=e^{-\Delta/2}$ and $s$ is defined in \Cref{thm:schurCDelta}.
Plugging this into \eqref{eq:qbound} and combining with
\Cref{thm:schurCDelta} yields
\begin{align}
  |Q(\theta)|
  & \leq
  \frac{1-C_{1,1}s}{1-(C_{0,0}+C_{1,1})s}
  \brac{|\hat{\rho}_{\theta_i}(\theta)|+C_{0,0}(x+s)}\notag\\
  &\qquad+\frac{C_{0,1}s}{1-(C_{0,0}+C_{1,1})s}
  \brac{|\hat{\rho}_{\theta_i}^{(1,0)}(\theta)|+C_{1,0}(x+s)}\\
  & \leq
  \frac{1-C_{1,1}s}{1-(C_{0,0}+C_{1,1})s}
  \brac{\gamma_2+C_{0,0}(x+s)}
  +\frac{C_{0,1}s}{1-(C_{0,0}+C_{1,1})s}
  \brac{\gamma_3+C_{1,0}(x+s)}.
\end{align}
Solving for $|Q(\theta)|<1$ we obtain
\begin{equation}
  \brac{1-C_{1,1}s}\brac{\gamma_2+C_{0,0}(x+s)}
  +\brac{C_{0,1}s}\brac{\gamma_3+C_{1,0}(x+s)} <
  1-(C_{0,0}+C_{1,1})s.
\end{equation}
Isolating $s$ and $x$,
\begin{equation}
  s(2C_{0,0}+C_{1,1}-C_{1,1}\gamma_2+C_{0,1}\gamma_3)+xC_{0,0}<1-\gamma_2,
\end{equation}
where we apply \Cref{eq:algebraic} ($C_{1,0}C_{0,1}=C_{0,0}C_{1,1}$) to cancel terms.
Since $s=\frac{2x^2}{1-x^2}$ we obtain the inequality
\begin{equation}
  \frac{2x^2}{1-x^2}a + xb < c
\end{equation}
where $a:=2C_{0,0}+C_{1,1}-C_{1,1}\gamma_2+C_{0,1}\gamma_3>0$, $b:=C_{0,0}$, and
$c:=1-\gamma_2>0$.  The following lemma shows that this inequality is
satisfied by our assumptions, completing this part of the proof.
\begin{lemma}[Proof in \Cref{sec:quadineqproof}]
  \label{lem:quadineq}
  Let $a,b,c,\Delta>0$, and $x:=e^{-\Delta/2}$.  Then the inequality
  \begin{equation}
    \label{eq:specialineq}
    \frac{2x^2}{1-x^2}a + xb < c
  \end{equation}
  is satisfied when
  \begin{equation}
    \label{eq:solution}
    \Delta > 2\log\brac{\frac{2(2a+c)}{-b + \sqrt{b^2+4(2a+c)c}}}.
  \end{equation}
\end{lemma}
\subsection*{Bounding $Q^{(2)}(\theta)$ in the Near Region}
For the final part of the proof we must prove $Q^{(2)}(\theta)<0$ for
$\theta$ in the near region of some $\theta_i$ with $\xi_i=1$.  Since the
interpolation equations \eqref{eq:interp1} and \eqref{eq:interp2}
guarantee that $Q(\theta_i)=1$, strict concavity implies that $Q$ has a
unique maximum on $[N_i^-,N_i^+]$ thus establishing that $Q(\theta)<1$
for $\theta\in[N_i^-,N_i^+]\setminus\{\theta_i\}$. We cannot have
$Q(\theta)\leq-1$ on the near region since we already showed that $|Q(\theta)|<1$ outside the near region.
We can assume,
without loss of generality, 
that the conditions of \Cref{thm:cluster} hold for some $i$ and that $\theta\leq
N_i^+$.  By the same argument used to obtain \Cref{eq:eps_bound}, we have
\begin{equation}
  \epsilon_i^{(q,r)}(\theta) \leq C_{q,r}(x+s).
\end{equation}
Plugging this into \eqref{eq:qbound} and applying
\Cref{thm:schurCDelta} we obtain
\begin{align}
  Q^{(2)}(\theta)
  & \leq \frac{1-(2C_{0,0}+C_{1,1})s}{1-(C_{0,0}+C_{1,1})s}\hrho_{\theta_i}^{(0,2)}(\theta)
  +\frac{1-C_{1,1}s}{1-(C_{0,0}+C_{1,1})s}\brac{C_{0,2}(x+s)} \nonumber\\
  &+ \frac{C_{0,1}s}{1-(C_{0,0}+C_{1,1})s}
  \brac{|\hat{\rho}_{\theta_i}^{(1,2)}(\theta)|+C_{1,2}(x+s)}\\
  & \leq \frac{1-(2C_{0,0}+C_{1,1})s}{1-(C_{0,0}+C_{1,1})s}\brac{-\gamma_0}
  +\frac{1-C_{1,1}s}{1-(C_{0,0}+C_{1,1})s}\brac{C_{0,2}(x+s)} \nonumber\\
  &+ \frac{C_{0,1}s}{1-(C_{0,0}+C_{1,1})s}
  \brac{\gamma_1+C_{1,2}(x+s)},
\end{align}
where $\hrho_{\theta_i}^{(0,2)}(\theta)\leq 0$ by \Cref{def:near}.  Solving
for $Q^{(2)}(\theta)<0$ yields
\begin{equation}
  -(1-(2C_{0,0}+C_{1,1})s)\gamma_0 + (1-C_{1,1}s)C_{0,2}(x+s) +
  C_{0,1}s(\gamma_1+C_{1,2}(x+s)) < 0.
\end{equation}
Grouping into terms involving $s$ and $x$ we obtain
\begin{equation}
  s((2C_{0,0}+C_{1,1})\gamma_0+C_{0,2}+C_{0,1}\gamma_1) + C_{0,2}x < \gamma_0,
\end{equation}
where we apply \Cref{eq:algebraic} ($C_{0,1}C_{1,2}=C_{0,2}C_{1,1}$) to cancel terms.
Letting $a:=(2C_{0,0}+C_{1,1})\gamma_0+C_{0,2}+C_{0,1}\gamma_1>0$, $b:=C_{0,2}$, and
$c:=\gamma_0>0$ we obtain an inequality of the form
\eqref{eq:specialineq}.  Applying \Cref{lem:quadineq} completes the proof.

\begin{figure}[tp]
   \includegraphics[scale=0.43]{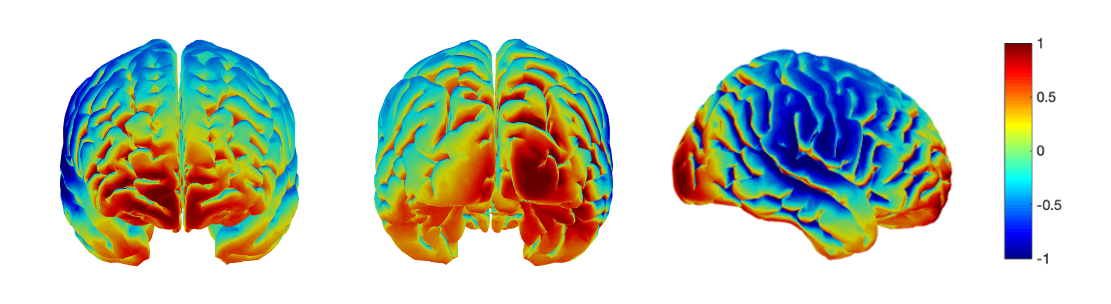}
  \caption{Interpolation function $Q$ associated to a valid dual certificate for the electroencephalography brain-activity localization problem. The dual certificate is associated to a signal consisting of the three sources of brain activity shown in Figure~\ref{fig:eegatoms}. $Q$ interpolates the sign pattern of the signal, which equals +1 on $\theta_1$ and $\theta_3$, and -1 on $\theta_2$. The certificate is built using the multidimensional extension of our proposed interpolation technique described in Section~\ref{sec:2d}.}
  \label{fig:eegcert}
\end{figure}

\subsection{Extensions to Higher Dimensions}
\label{sec:2d}
In this section we briefly describe an extension of our dual certificate construction to settings where the parameter space has dimension $p>1$. We leave a more detailed analysis to future work.
In higher dimensions, one can build the interpolating function $Q$ by setting
\begin{equation}
  Q(\eta) := \sum_{i=1}^k \left[\alpha_i\rho_{\theta_i}(\eta)
    + \sum_{l=1}^p \beta_{i,l}
    \frac{\partial_{1,l}\,\rho_{\theta_i}(\eta)}{
      \normTwo{\partial_{l}\vphi(\theta_i)}^2}\right],
\end{equation}
where $\alpha\in\RR^k$, $\beta\in\RR^{k\times p}$, $\partial_l$
denotes the partial derivative with respect to the $l$th coordinate,
and
$\partial_{1,l}\rho_\theta(\eta)$ denotes the partial derivative of $\rho_\theta(\eta)$ with
respect to the $l$th coordinate of $\theta$. The coefficients are chosen so that $Q$ satisfies the analog of the interpolation \Cref{eq:interp1,eq:interp2}, 
\begin{align}
  Q(\theta_i) & = \xi_i,\label{eq:interp1_d}\\
  \nabla Q(\theta_i) & = 0.\label{eq:interp2_d}
\end{align}
In \Cref{fig:eegcert} we show an example of an interpolating function $Q$ for the electroencephalography brain-activity localization problem. The interpolating function is associated to the signal with three active sources of brain activity from Figure~\ref{fig:eegatoms}. To control $Q$ for problems with correlation decay, one can extend the correlation-decay conditions in \Cref{sec:correlation_decay} to higher dimensions by requiring analogous bounds on the first and second partial derivatives.  For example, the local quadratic bound in \Cref{eq:near1} becomes a bound on the eigenvalues of the Hessian of $\rho_\theta$. Similar conditions have been used in \cite{superres,bendory2016robust}
to obtain recovery guarantees for the super-resolution and deconvolution problems in two dimensions.

\section{Numerical Experiments}
\label{sec:numerical}
\subsection{Heat-source localization}
\label{sec:numexact}
Our theoretical results establish that convex programming yields exact recovery in parameter-estimation problems that have correlation decay. In this section we
investigate this phenomenon numerically for a heat-source localization
problem in one dimension. The heat sources are modeled as a collection
of point sources, 
\begin{equation}
  \mu := \sum_{\theta_i\in\Theta} c_i\delta_{\theta_i},
\end{equation}
where $\Theta$ is a finite number of points in the unit interval and $c_i\in\RR$ for $i=1,\ldots,k$. 

The heat $u(\theta,t)$ at
position $\theta$ and time $t$ is assumed to evolve following the heat equation with Neumann boundary conditions (see e.g. \cite{li2014heat}),
\begin{align}
  \frac{\partial}{\partial t}u(\theta,t) &= \frac{\partial}{\partial
    \theta}\left(c(\theta)\frac{\partial}{\partial \theta}u(\theta,t)\right),\label{eq:heat1}\\
  \frac{\partial}{\partial \theta}u(-0.5,t) &= \frac{\partial}{\partial
    \theta}u(0.5,t) = 0\label{eq:heat2}
\end{align}
on the unit interval, where $c(\theta)$ represents the conductivity of the medium at $\theta$ (see \Cref{fig:num_conductivity}). The data are heat measurements $u(j,T)$,
where $j$ is sampled on a regular grid of 100 points at a fixed time $T:=10^{-4}$. Our goal is to recover the initial heat distribution at $t=0$ represented by the heat sources $\mu$. This is an SNL problem where $\vphi(\theta)$ corresponds to the measurements caused by a source located at $\theta$. 

\begin{figure}[t]
  \centering
  \includegraphics[width=0.4\linewidth]{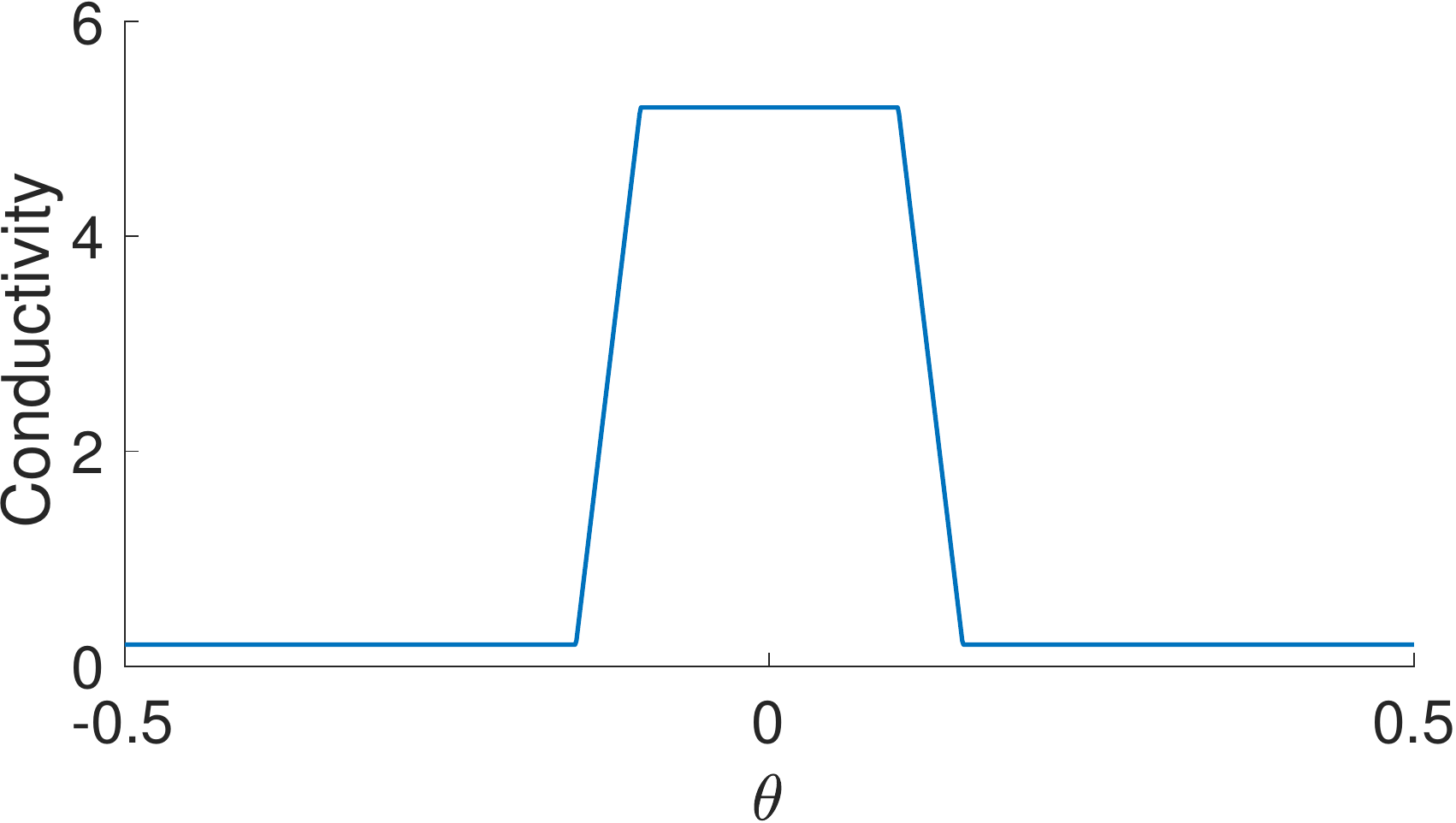}
  \caption{The conductivity function $c(\theta)$ used in the experiment described in Section~\ref{sec:numexact}. Heat diffuses faster in the central
    region, where the conductivity is higher.}
  \label{fig:num_conductivity}
\end{figure}

If the conductivity is constant, then $\vphi(\theta)$ has a
Gaussian-like shape with fixed width sampled on the measurement grid,
and the correlation function $\rho_{\theta}$ is also shaped like a
Gaussian with fixed width. In this example, the conductivity varies
(see \Cref{fig:num_conductivity}), which results in correlations that
are still mostly Gaussian-like but have very different widths (see
\Cref{fig:num_corrs}). The measurement operator therefore has
nonuniform correlation decay properties. According to the theory
presented in \Cref{sec:exactvarying}, we expect convex programming to
recover the heat-source location as long as they are separated by a
minimum separation that takes into account the structure of the
correlation decay. To verify whether this is the case we consider two
different separation measures. The first just measures the minimum
separation between the sources, 
\begin{equation}
  \Delta_{\op{sep}} := \min_{i\neq j}|\theta_i-\theta_j|. \label{eq:stsep}
\end{equation}
The second takes into account the correlation function of the SNL problem,
\begin{equation}
  \Delta_{\op{corr}} := \min_{i\neq
    j}\frac{|\theta_i-\theta_j|}{\max(\sigma_i,\sigma_j)} \label{eq:sigsep}.
\end{equation}
where $\sigma_i$, $1\leq i \leq k$, is the standard deviation of the best Gaussian upper bound on the correlation function centered at $\theta_i$,
\begin{equation}
  \sigma_i := \inf\{s>0 \mid \rho_{\theta_i}(\theta) <
  e^{-(\theta_i-\theta)^2/(2s^2)}~\text{for all $\theta$.}\}\label{eq:sigg}
\end{equation}

\begin{figure}[t]
 \hspace{-0.5cm}
 {\footnotesize
  \begin{tabular}{>{\centering\arraybackslash}m{0.08\linewidth} >{\centering\arraybackslash}m{0.2\linewidth} >{\centering\arraybackslash}m{0.2\linewidth} >{\centering\arraybackslash}m{0.2\linewidth} >{\centering\arraybackslash}m{0.2\linewidth}}
    & Dilation Factor\\
    &$1$ & $0.68$ & $0.34$ & $0.01$\\
    Uniform &
    \includegraphics[width=\linewidth]{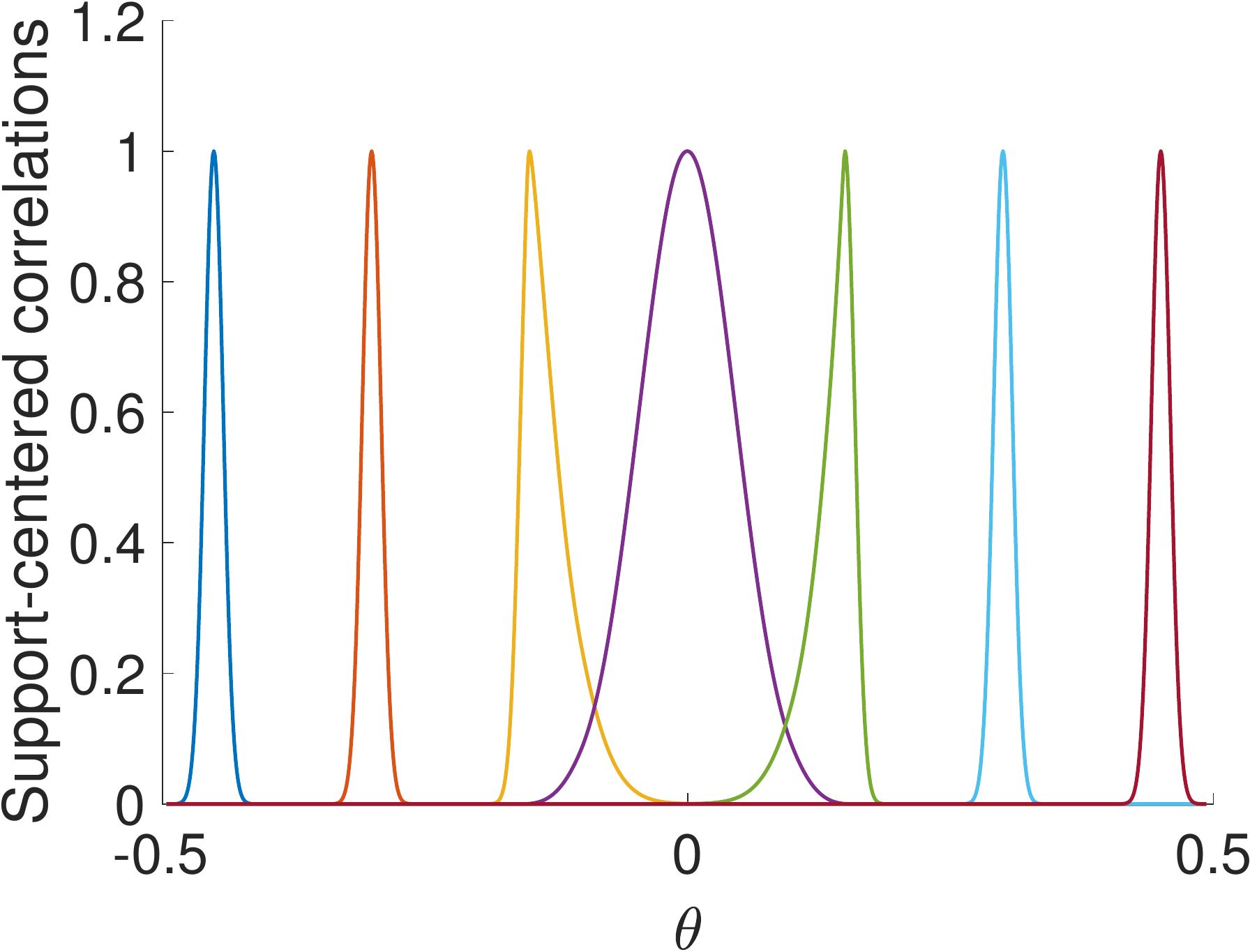}&
    \includegraphics[width=\linewidth]{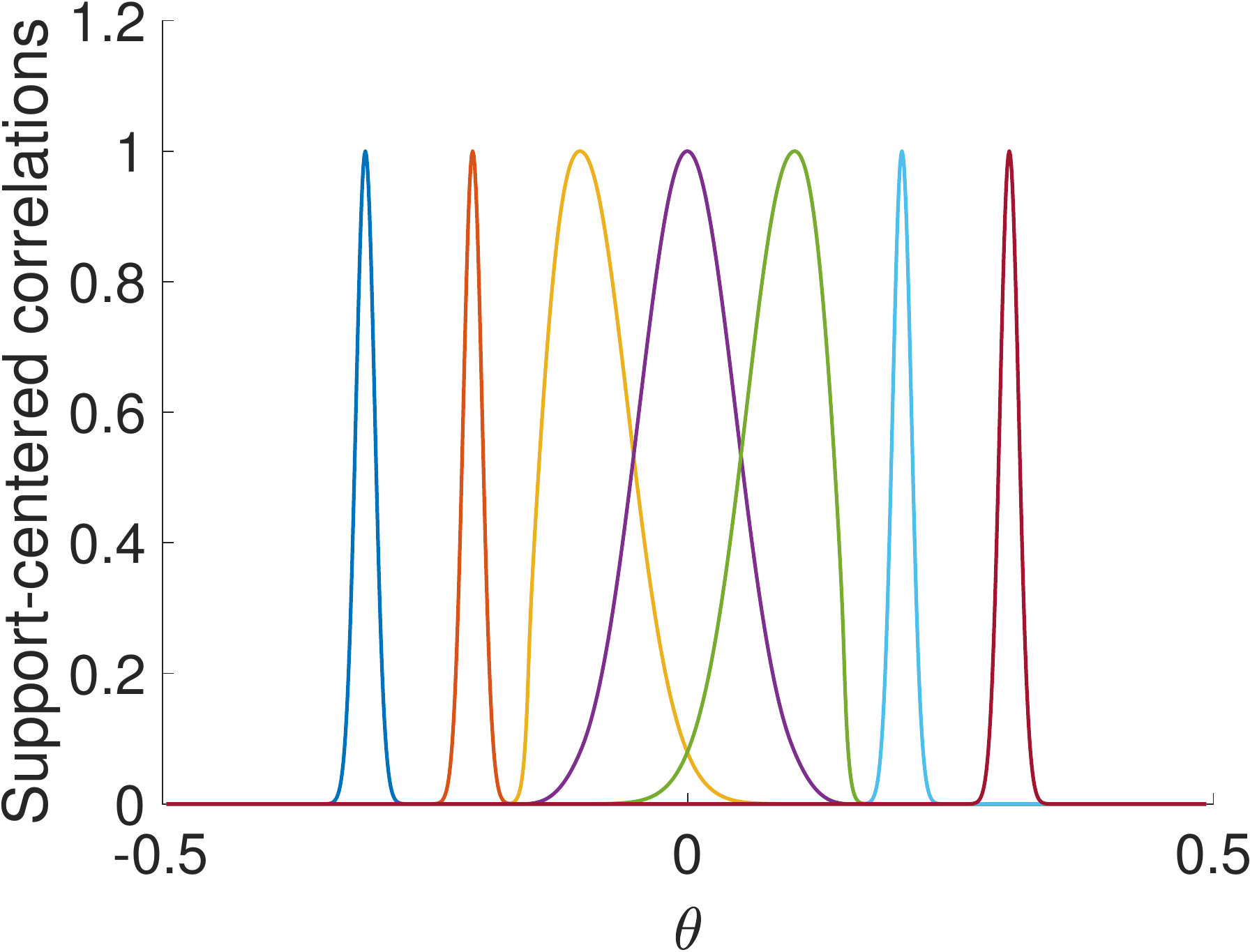}&
    \includegraphics[width=\linewidth]{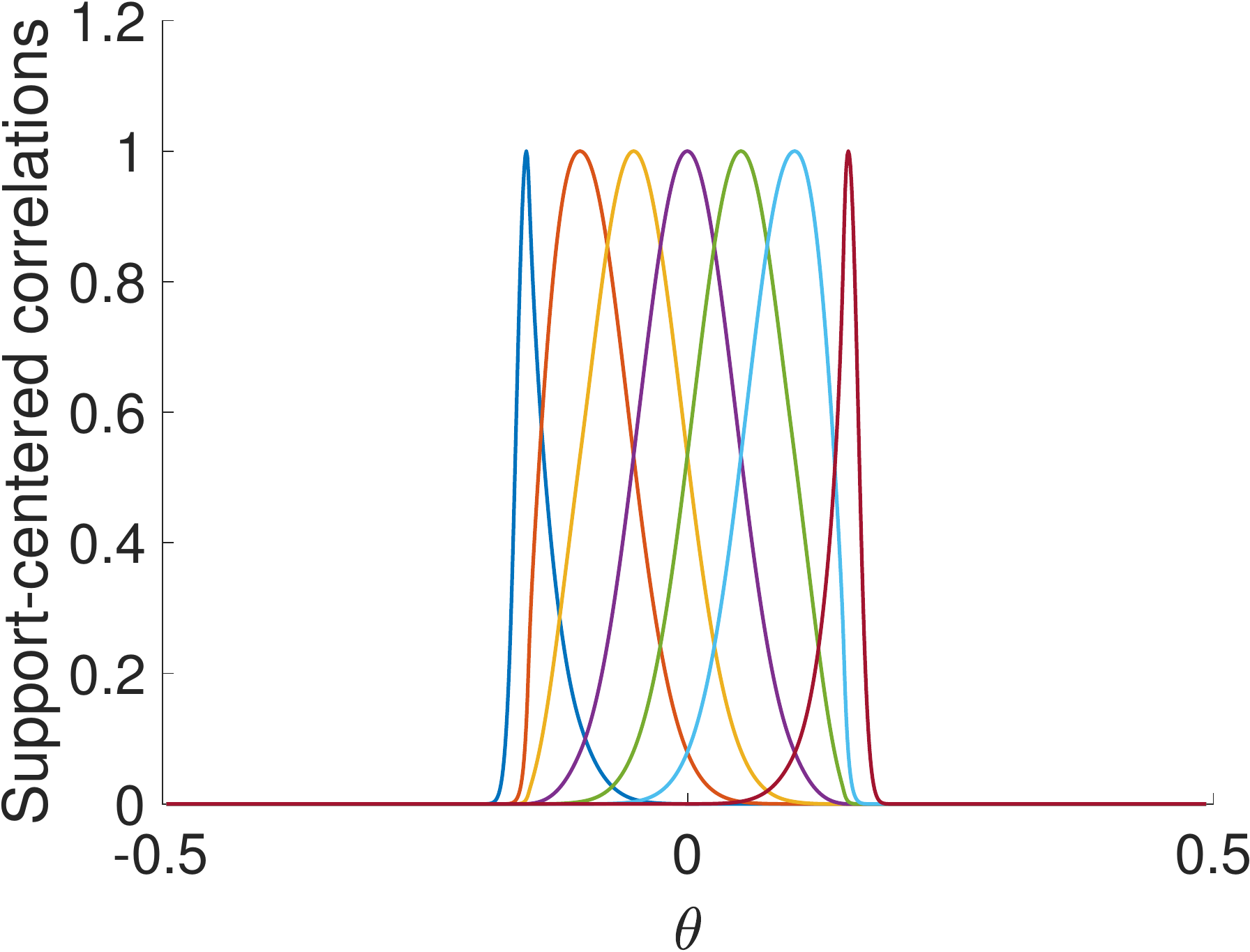}&
    \includegraphics[width=\linewidth]{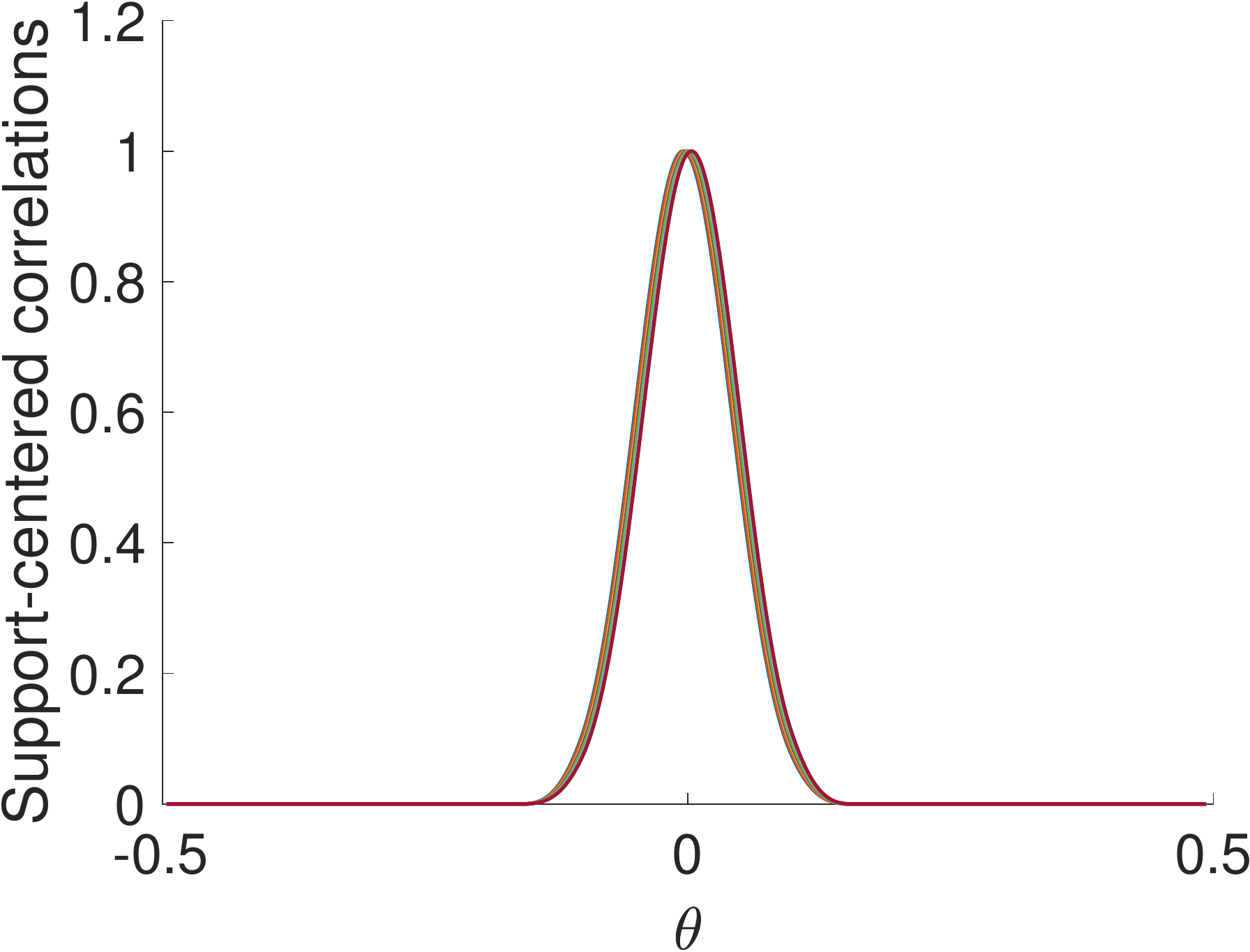}\\
    Clustered &
    \includegraphics[width=\linewidth]{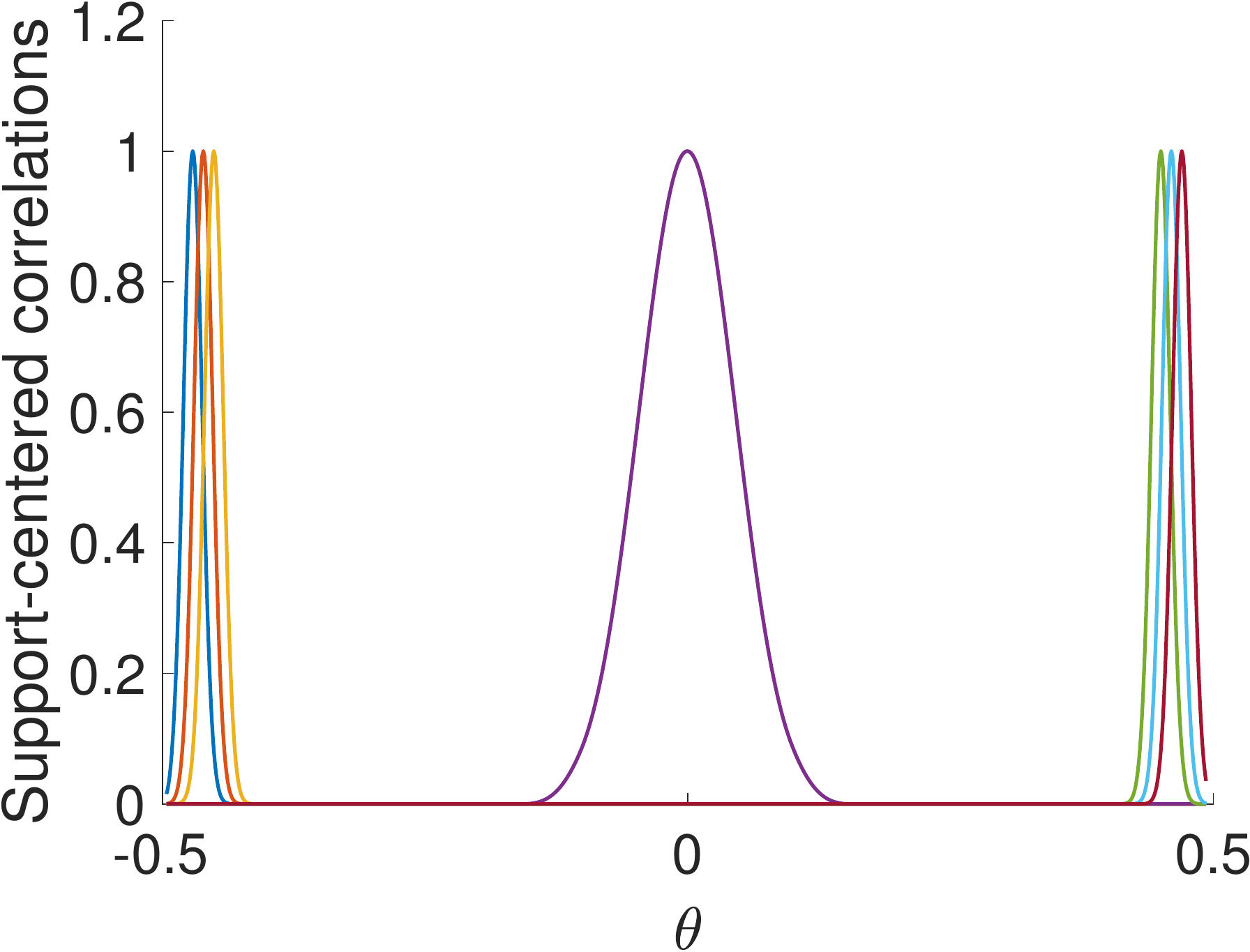}&
    \includegraphics[width=\linewidth]{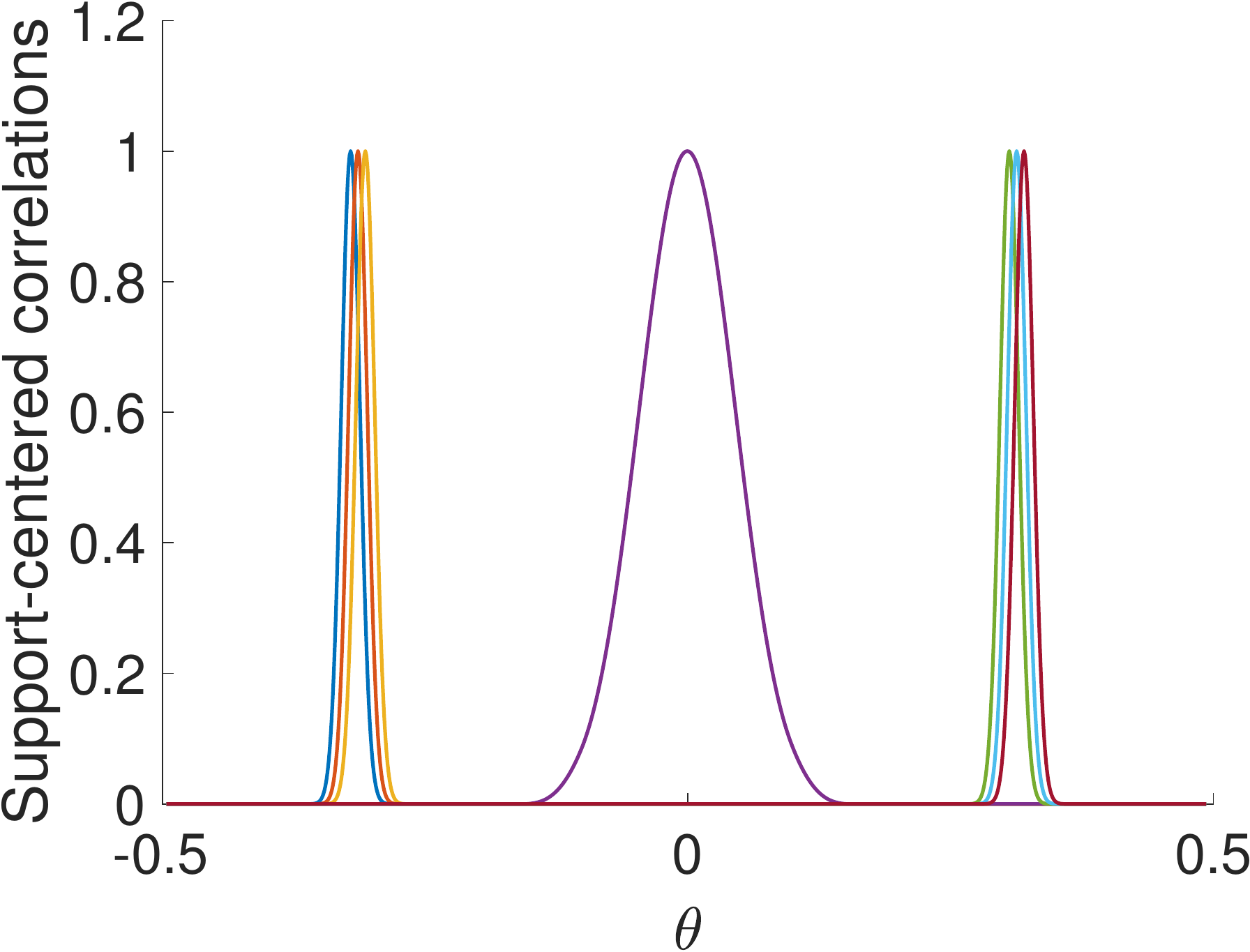}&
    \includegraphics[width=\linewidth]{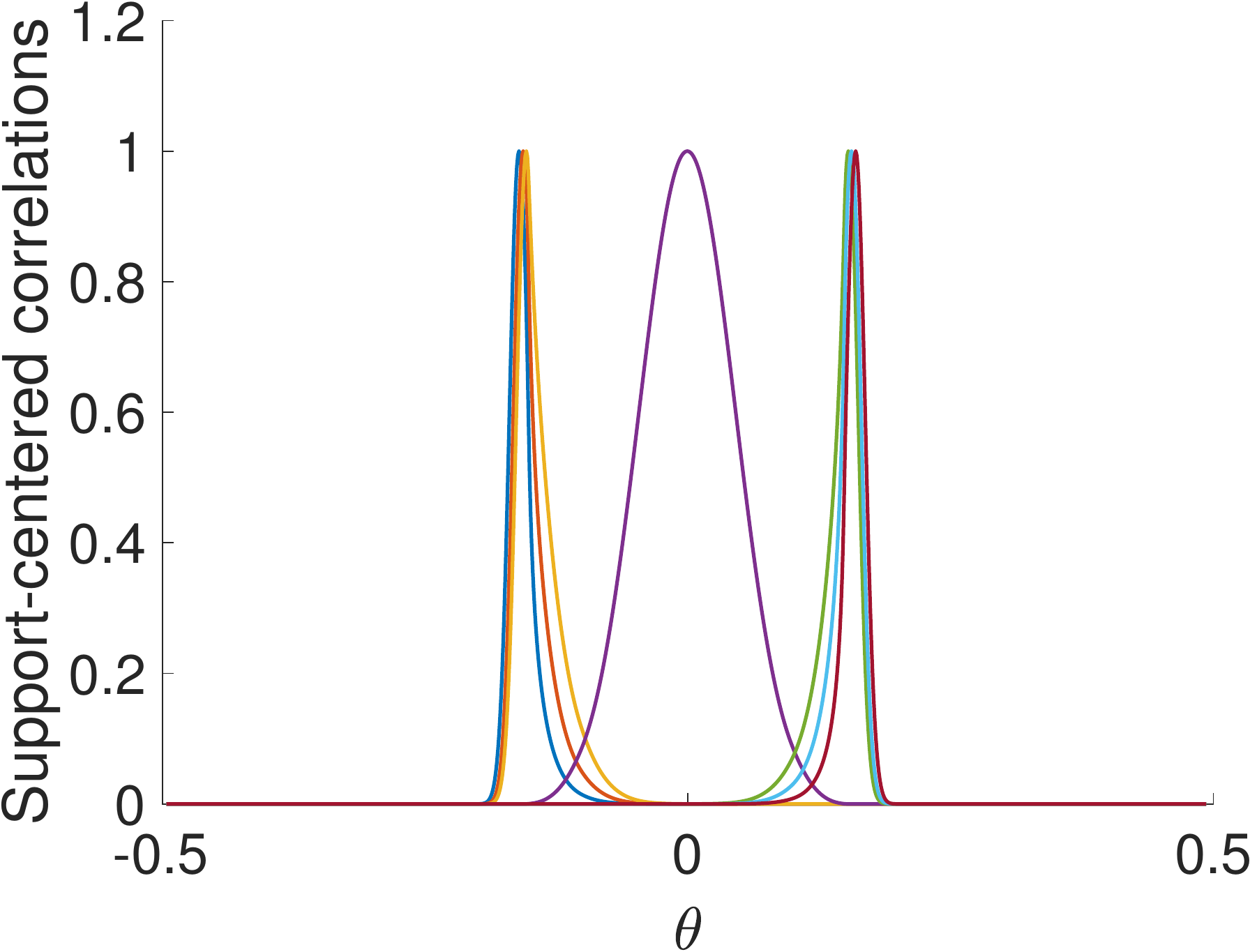}&
    \includegraphics[width=\linewidth]{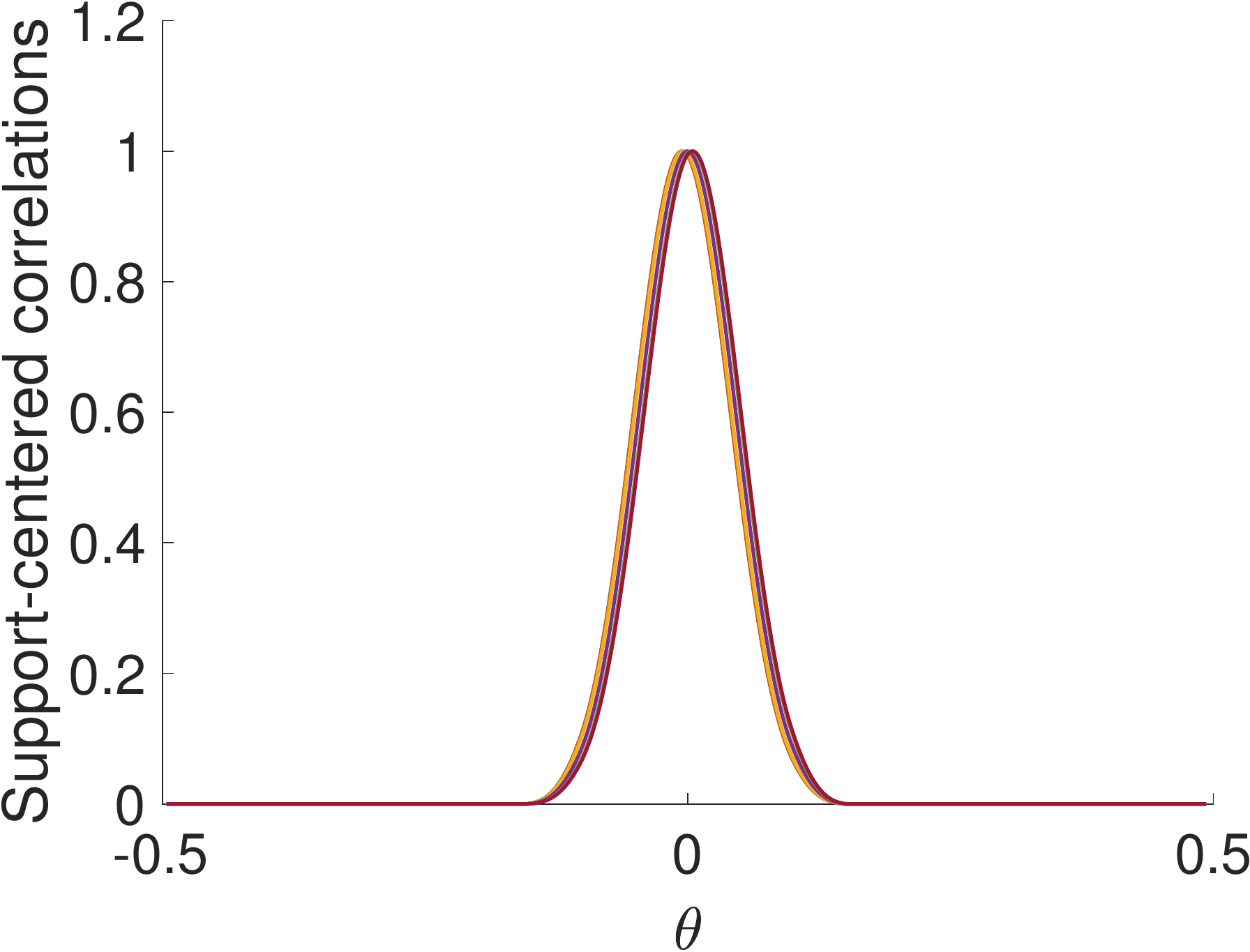}
  \end{tabular}
  }
  \caption{Examples of the correlation functions
    $\rho_{\theta_1},\ldots,\rho_{\theta_7}$ for heat sources with different separations and configurations.  The
    sources used in the top row of figures are uniformly spaced on the unit interval. The heat sources in the bottom row have one central source, and two clusters of three sources located at the ends. From left to right the supports are dilated by different factors. This reduces the separation between the heat sources but maintains their relative position.}
  \label{fig:num_corrs}
\end{figure}

The question is which of these separations characterizes the
performance of the convex-programming approach more accurately. We
consider two different heat-source configurations: one where the
sources are uniformly spaced and another where they are clustered, as
depicted in \Cref{fig:num_corrs}.
\bdb{In the clustered configuration, groups of spikes are placed in
  regions of low conductivity where each $\rho_{\theta_i}$ exhibits
  sharp decay. This
  creates a noticeable discrepancy between the two measures of
  separation.  We expect exact recovery to occur for much
  smaller values of $\Delta_{\op{sep}}$ in the clustered case than it
  does in the uniform case since it doesn't account for the
  correlation decay.  In contrast, we expect $\Delta_{\op{corr}}$ to
  have similar thresholds for exact recovery in both cases, as it
  captures the inherent difficulty in the problem.}
To recover these sources we solve
the discretized $\ell_1$-norm minimization problem described in
\Cref{sec:discretization}  using CVX~\cite{cvx}. The measurement
matrix is computed by solving the differential equation on a grid of
$10^3$ points in the parameter space.

\begin{figure}[t]
  \centering
  \begin{tabular}{>{\centering\arraybackslash}m{0.07\linewidth} >{\centering\arraybackslash}m{0.4\linewidth} >{\centering\arraybackslash}m{0.4\linewidth}}
    & $\Delta_{\op{sep}}$ & $\Delta_{\op{corr}}$\\
    Uniform
    & \includegraphics[width=\linewidth]{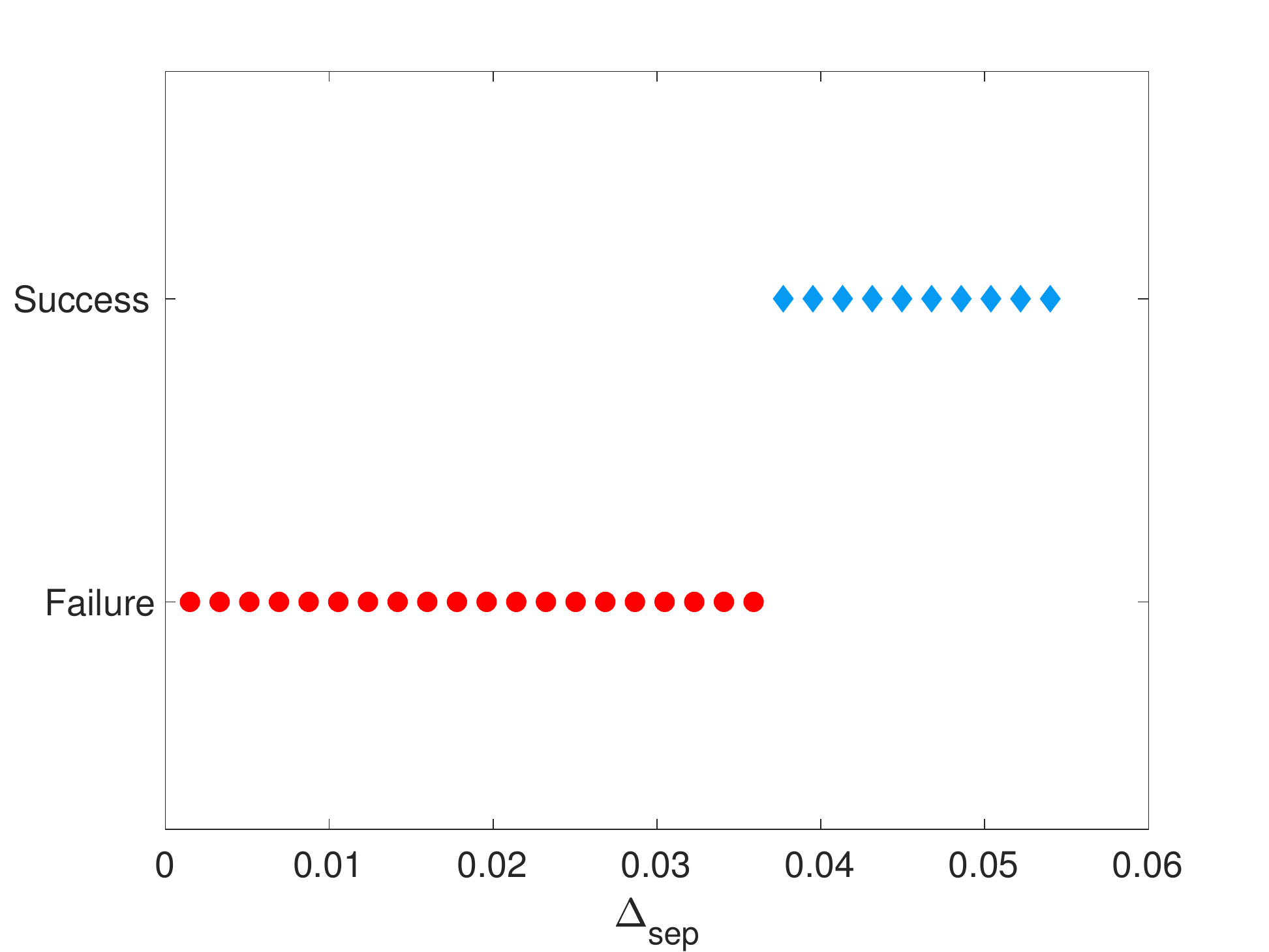}&
    \includegraphics[width=\linewidth]{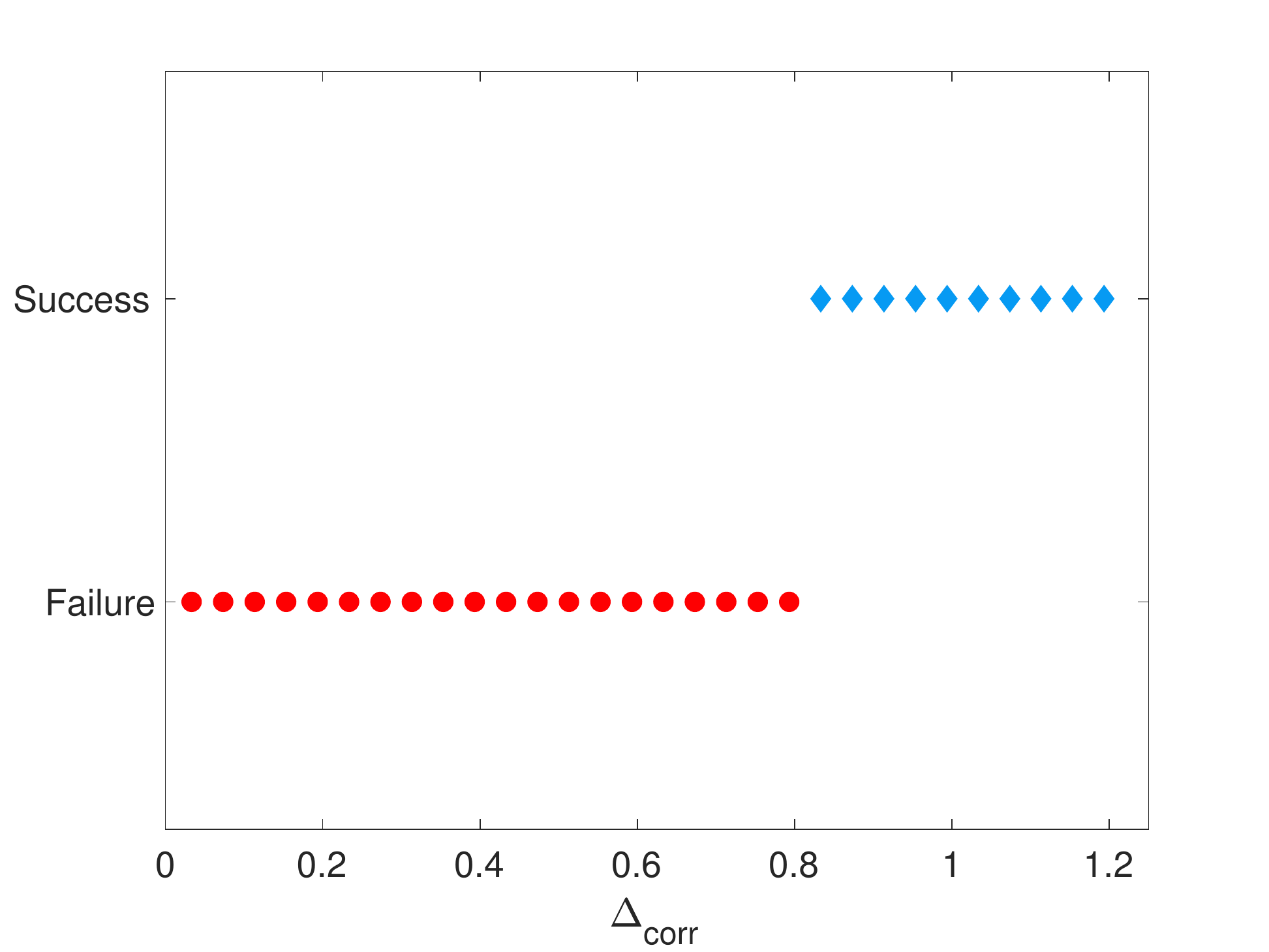}\\
    Clustered
    & \includegraphics[width=\linewidth]{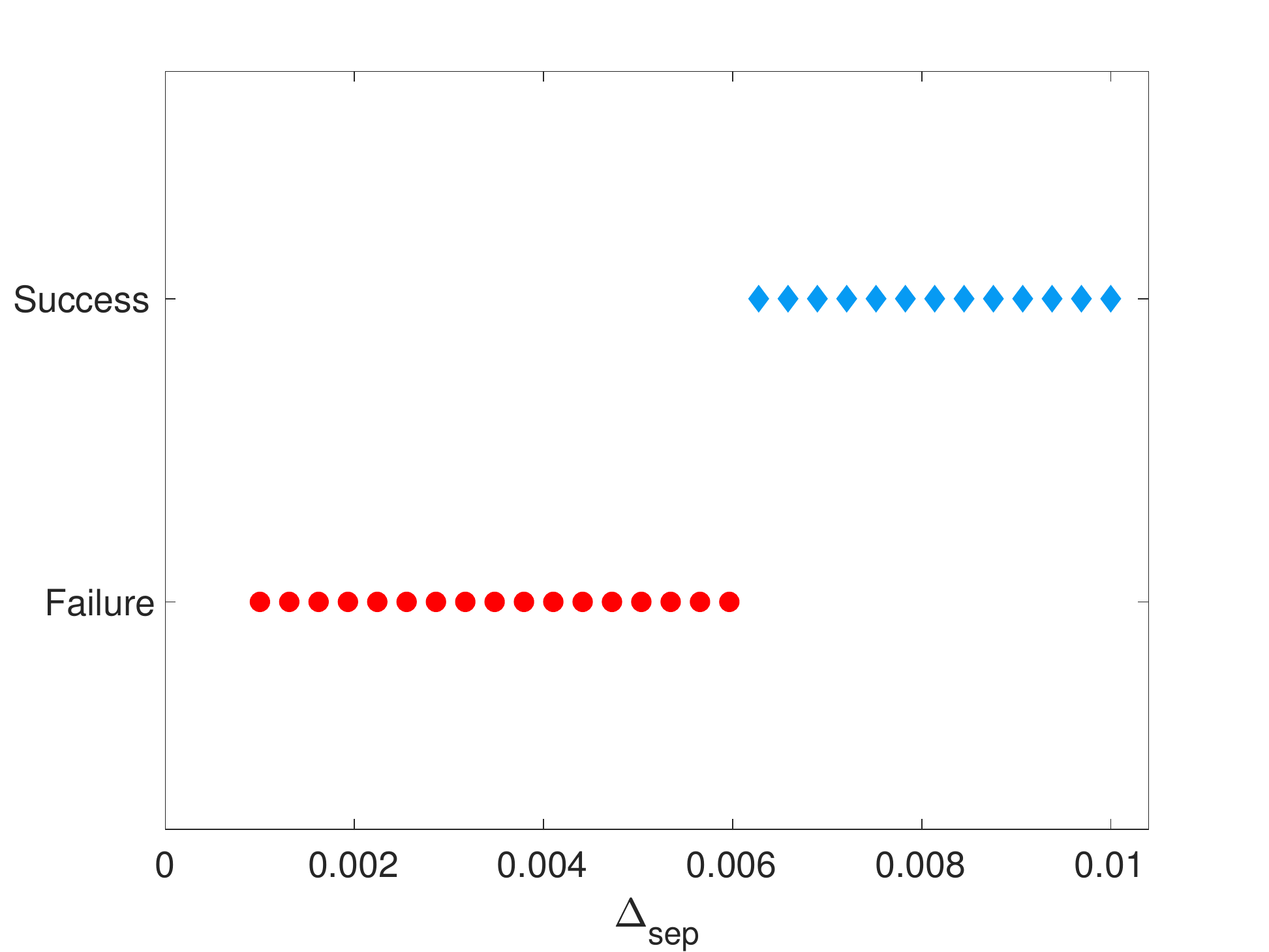} &
    \includegraphics[width=\linewidth]{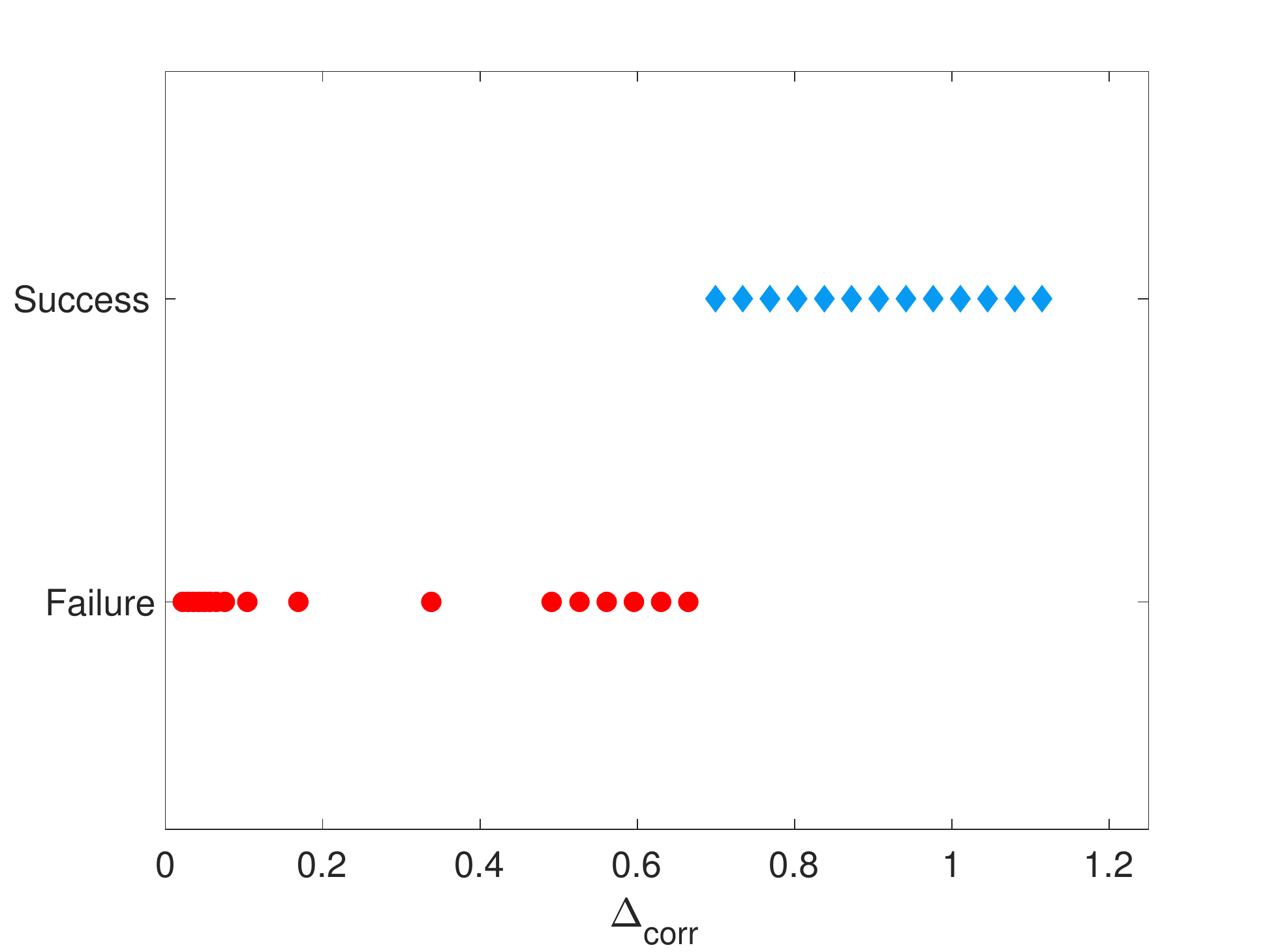}
  \end{tabular}
  \caption{
 The figure shows the performance of heat-source localization based on
 convex programming for the experiment described in
 \Cref{sec:numexact}. The upper row shows the result for sources with
 uniformly spaced supports, whereas the bottom row shows results for
 sources with clustered supports (see \Cref{fig:num_corrs}). The left
 column plots the results with respect to the minimum separation
 $\Delta_{\op{sep}}$. The right column plots the results in terms of
 the correlation-aware separation $\Delta_{\op{corr}}$. Recovery
 succeeds if the \bdb{relative} $\ell_2$-norm recovery error is smaller than $3 \cdot
 10^{-5}$. All failures have an \bdb{absolute} error above $9 \cdot 10^{-3}$, and
 most have errors above 1. For comparison the $\ell_2$ norm of the
 true signal equals $\sqrt{7}$ ($\approx 2.65$) in all cases.} 
  \label{fig:cvxnonuniform}
\end{figure}

\Cref{fig:cvxnonuniform} shows the results. For both types of support, we observe that exact recovery occurs as long as the minimum separation $\Delta_{\op{sep}}$ between sources is larger than a certain value. However, this value is very different for the two types of support. It is much smaller for the clustered support. Intuitively, this is due to the fact that the clustered sources are mostly located in the region where the conductivity is lower, and consequently the correlation decay is much faster. Our theory suggests that exact recovery will occur at smaller separations for such configurations, which is what we observe. Quantifying separation using $\Delta_{\op{corr}}$ accounts for the variation in correlation decay, resulting in a similar phase transition for both types of support. This is consistent with the theory for SNL problems with nonuniform correlation decay developed in \Cref{sec:exactvarying}.

\subsection{Estimation of Brain Activity via Electroencephalography}
\label{sec:eeg}

In this section we consider the SNL problem of brain-activity localization from electroencephalography (EEG) data. Areas of focalized brain activity are usually known as \emph{sources}. In EEG, as well as in magnetoencephalography, sources are usually modeled using electric dipoles, represented mathematically as point sources or Dirac measures at the corresponding locations~\cite{baillet2001electromagnetic}. EEG measurements are samples of the electric potential on the surface of the head obtained using an array of electrodes. The mapping from the source positions and amplitudes to the EEG data is governed by Poisson's equation. For a fixed model of the head geometry-- obtained for example from 3D images acquired using magnetic-resonance imaging-- the mapping can be computed by solving the differential equation numerically. By linearity, this corresponds to an SNL model where $\theta$ represents the position of a dipole and $\phi_t\brac{\theta}$ is the value of the corresponding electrical potential at a location $t$ of the scalp.

To simulate realistic EEG data with associated ground-truth source positions, we use Brainstorm~\cite{tadel2011brainstorm}, an open-source software for the analysis of EEG and other electrophysiological recordings. We use template ICBM152, which is a nonlinear average of 152 3D magnetic-resonance image scans~\cite{fonov2011unbiased}, to model the geometry of the brain and head. The sources are modeled as electrical dipoles situated on the cortical surface, discretized as a tesselated grid with 15,000 points. The orientation of the dipoles is assumed to be perpendicular to the cortex. The data-acquisition sensor array has 256 channels (HydroCel GSN, EGI). The data corresponding to each point on the grid are simulated numerically using the open-source software OpenMEEG~\cite{gramfort2010openmeeg}. The computation assumes realistic electrical properties for the brain, skull and scalp. As a result, we obtain a $256 \times 15,000$ measurement operator that can be used to generate data corresponding to different combinations of sources. Figure~\ref{fig:eegatoms} shows three examples. 

The correlation structure of the EEG linear measurement operator is depicted in Figure~\ref{fig:eegcorr}. The complex geometric structure of the folds in cortex induce an intricate correlation structure close to the source location. However, the correlation presents a clear correlation decay; points that are far enough from the source have low correlation values. Our theoretical results consequently suggest that $\ell_1$-norm minimization should succeed in recovering superpositions of sources that are sufficiently separated. This is supported by the fact that for such superpositions, we can build valid dual certificates as illustrated in Figure~\ref{fig:eegcert}.

\begin{figure}[t]
 \begin{tabular}{ >{\centering\arraybackslash}m{0.48\linewidth} >{\centering\arraybackslash}m{0.48\linewidth}}
 Best case & Worst case \\
\includegraphics[width=\linewidth]{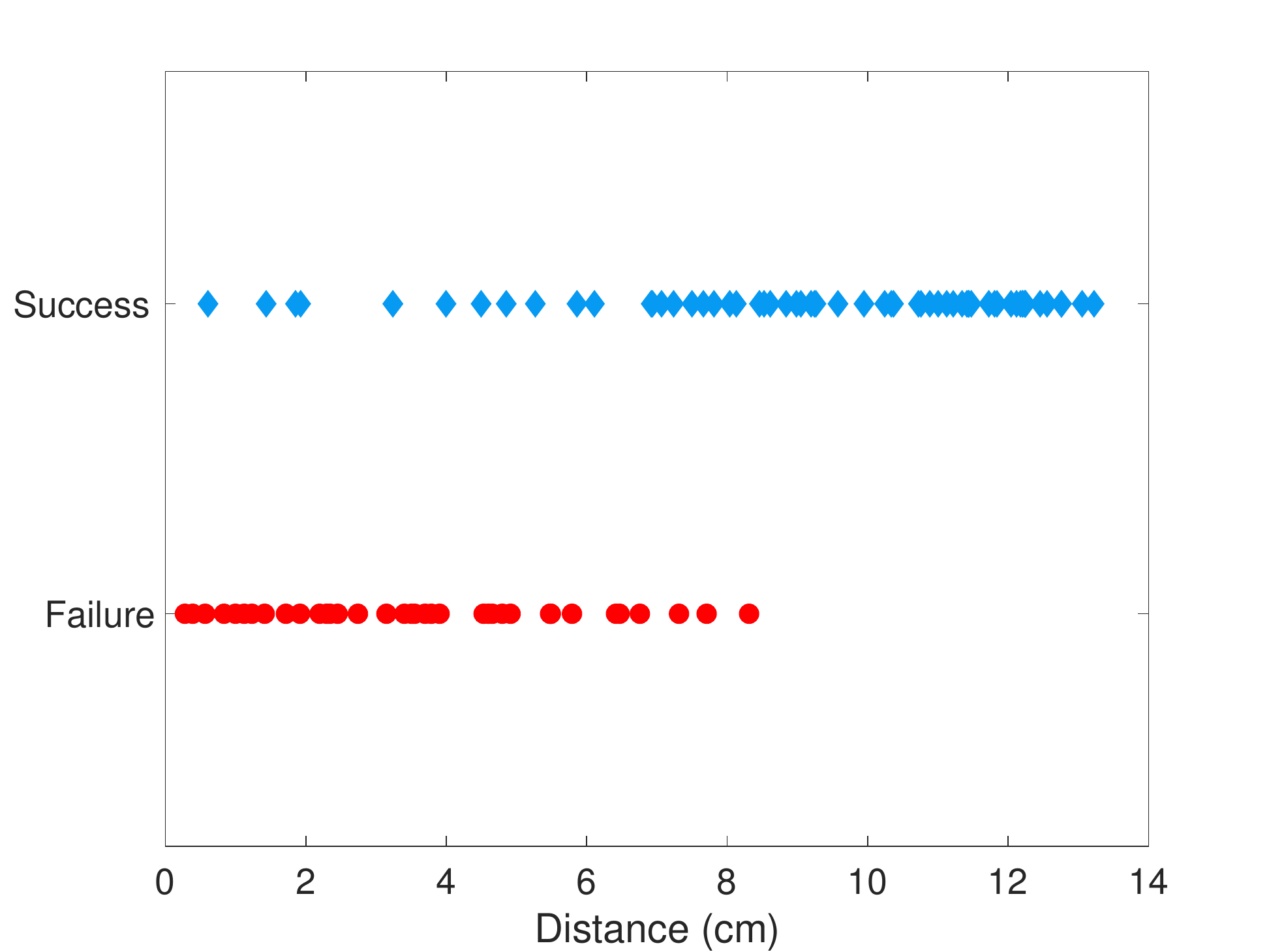} &
\includegraphics[width=\linewidth]{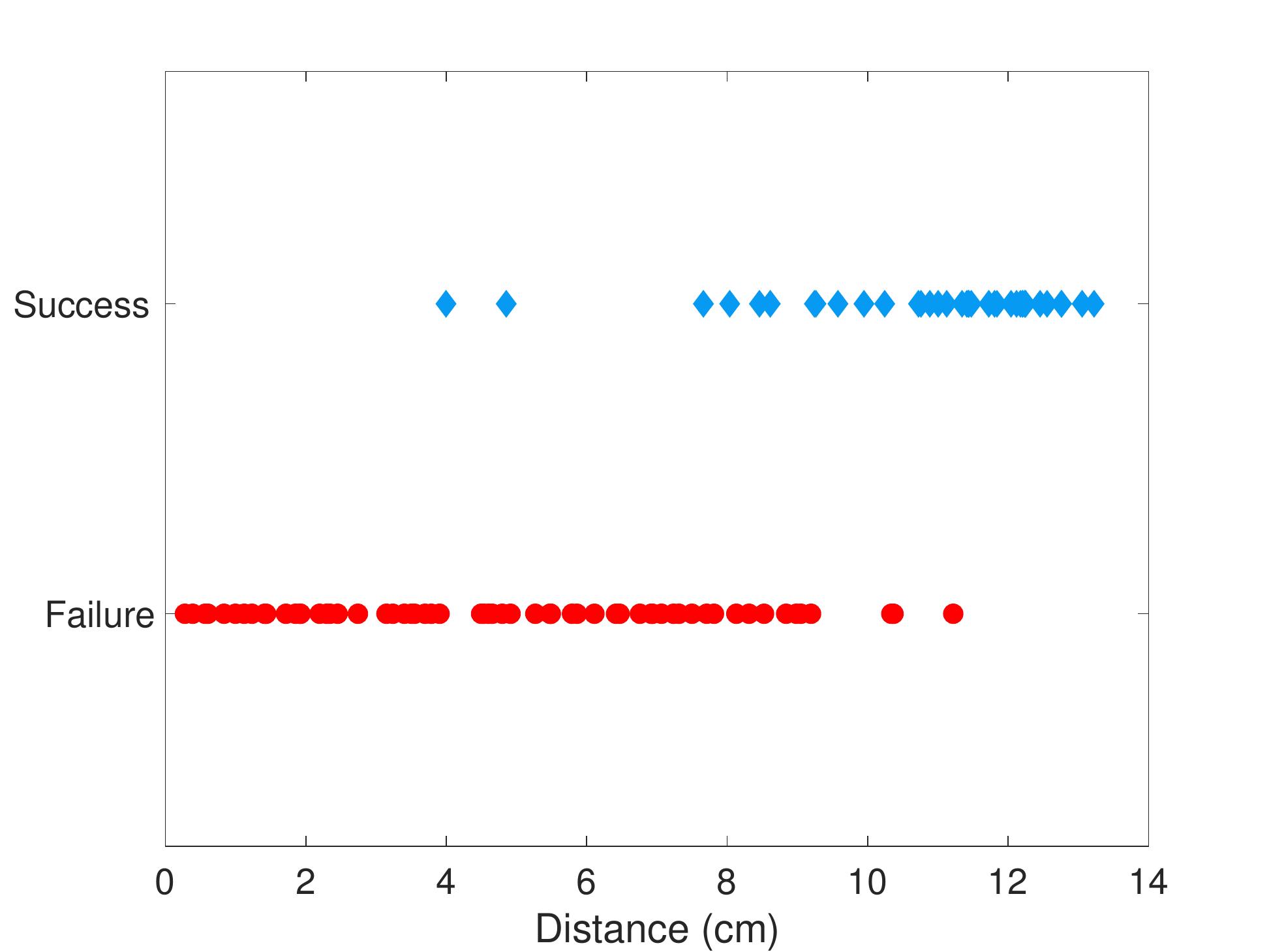}
\end{tabular}
\caption{Result of the brain-source localization experiment described
  in Section~\ref{sec:eeg}. The horizontal axis indicates the distance
  between the sources. In the left image success is declared if the
  three sources are accurately recovered for \emph{any} pattern of
  positive and negative amplitudes. In the right image success is
  declared if the sources are recovered for \emph{all} patterns of
  positive and negative amplitudes. Recovery succeeds if the
  \bdb{relative} $\ell_2$-norm recovery error is smaller than $10^{-4}$. All failures
  have an \bdb{absolute} error above $1.18 $. For comparison the $\ell_2$ norm of the
  true signal equals $\sqrt{3}$ ($\approx 1.73$) in all cases. }
\label{fig:eegresults}
\end{figure}

In order to test our hypothesis, we randomly choose superpositions of three sources located at approximately the same distance from each other, for a range of distances. For each superposition we generate 8 different measurement vectors by assigning every possible combination of positive and negative unit-norm amplitudes to each source. We then estimate the source locations by solving an $\ell_1$-norm minimization problem with equality constraints using CVX~\cite{cvx}. Figure~\ref{fig:eegresults} shows the results. When the separation is sufficiently large, exact recovery indeed occurs for all possible patterns of positive and negative amplitudes, as expected from our theoretical analysis. As the separation decreases, recovery fails for some of the patterns, as shown in the image on the right.  

\section{Conclusion and Future Work}

In this work, we establish that deterministic SNL problems can be solved via a tractable convex optimization program, as long as the parameters have a minimum separation in the parameter space with respect to the correlation structure of the measurement operator. As mentioned in Sections~\ref{sec:robustness}, our results can be used to establish some robustness guarantees at low-noise levels in terms of support recovery. Deriving more precise stability analysis in the spirit of~\cite{li2018approximate} is an interesting research direction. Another interesting open problem is characterizing the performance of reweighting techniques, which are commonly applied in practice to enhance solutions obtained from $\ell_1$-norm minimization~\cite{candes2008enhancing,mcmrf}. 

A drawback of the sparse-recovery framework discussed in this paper is computational complexity: for SNL problems in two or three dimensions it is very computationally expensive to solve $\ell_1$-norm minimization problems, even if we discretize the domain. In the last few years, algorithms that minimize nonconvex cost functions directly via gradient descent have been shown to provable succeed in solving underdetermined inverse problems involving randomized measurements~\cite{candes2015wirtinger,ma2017implicit}. As illustrated in Figure~\ref{fig:varproj}, nonlinear least-squares cost functions associated to deterministic SNL problems often have non-optimal local minima. An intriguing question is how to design cost functions for deterministic SNL problems that can be tackled directly by descent methods. 

Another computationally efficient alternative is to perform recovery using machine learning. Recent works suggest calibrating a feedforward network to output the model parameters using a training set of examples, with applications in point-source deconvolution~\cite{boyd2018deeploco} and super-resolution of line spectra~\cite{izacard2018learning}. Understanding under what conditions such techniques can be expected to yield accurate estimates is a challenging question for future research. 

\subsection*{Acknowledgements}
B.B.~is supported by the MacCracken Fellowship, and the Isaac Barkey and Ernesto Yhap Fellowship. C.P. is supported by NIH award NEI-R01-EY025673. C.F.~is supported by NSF award DMS-1616340. 
\begin{small}
\bibliographystyle{abbrv}
\bibliography{refs}
\end{small}

\appendix
\renewcommand{\appendixpagename}{Appendix}
\appendixpage


Throughout the appendix, we assume there is some compact interval
$I\subset\RR$ containing the support $\Theta$ of the true measure $\mu$.
In problem \eqref{pr:min_TV},
the variable $\tilde{\mu}$ takes values in the set of finite signed Borel
measures supported on $I$.
\section{Proof of \Cref{thm:dualcert}}
\label{sec:dualcertproof}
The following proof is standard, but given here for completeness.
\begin{proof}[\unskip\nopunct]
  Let $\nu$ be feasible for problem~\eqref{pr:min_TV}
  and define $h=\nu-\mu$.
  By taking the Lebesgue decomposition of $h$ with respect to
  $|\mu|$ we can write
  \begin{equation}
    h = h_{\Theta}+h_{\Theta^c},
  \end{equation}
  where $h_{\Theta}$ is absolutely continuous with respect to $|\mu|$,
  and $h_{\Theta^c}$, $|\mu|$ are mutually singular. In other
  words, the support of $h_\Theta$ is contained in~$\Theta$, and
  $h_{\Theta^c}(\Theta)=0$. 
  This allows us to write
  \begin{equation}
    h_{\Theta} = \sum_{t_j \in T} b_j\delta_{t_j},
  \end{equation}
  for some $b\in\RR^{|\Theta|}$. Set $\xi :=\sign(b)$, where we
  arbitrarily choose $\xi_j=\pm1$ if $b_j=0$. 
  By assumption there exists a corresponding
  $Q$ interpolating $\xi$ on $\Theta$.  Since $\mu$ and $\nu$ are feasible
  for problem~\eqref{pr:min_TV} we have
  $\int \phi(\theta)\,dh(\theta)=0$.  This implies
  \begin{equation}
    \label{eq:nutvbound}
      0 = q^T\int\phi(\theta)\,dh(\theta)\\
      = \int Q(\theta)\,dh(\theta)\\
      = \normTV{h_{\Theta}} + \int Q(\theta)\,dh_{\Theta^c}(\theta).
  \end{equation}
  Applying the triangle inequality, we obtain
  \begin{align}
    \normTV{\nu} & =
    \normTV{\mu+h_{\Theta}}+\normTV{h_{\Theta^c}} &&\text{(Mutually Singular)}\\
    & \geq  \normTV{\mu} + \normTV{h_{\Theta^c}}-\normTV{h_\Theta}&&
    \text{(Triangle Inequality)}\\
    & = \normTV{\mu} + \normTV{h_{\Theta^c}}+\int
      Q(\theta)\,dh_{\Theta^c}(\theta)&& \text{(Equation
        \eqref{eq:nutvbound})}\\
    & \geq \normTV{\mu} && \text{($|Q(\theta)|\leq1$)},
  \end{align}
  where the last inequality is strict if $\normTV{h_{\Theta^c}}> 0$ since
  $|Q(\theta)|<1$ for $\theta\in \Theta^c$.
  This establishes that $\mu$ is optimal for problem~\eqref{pr:min_TV},
  and that any other optimal solution must be supported on $\Theta$.
  \Cref{eq:nutvbound} implies that any feasible solution
  supported on $\Theta$ must be equal to $\mu$ (since
  $\normTV{h_\Theta}=0$),
  completing the proof of uniqueness.
\end{proof}
\section{Proof of \Cref{thm:schurbound}}
\label{sec:schurboundproof}
\begin{proof}[\unskip\nopunct]
  For any matrix $A\in\RR^{n\times n}$ such that $\normInf{ A} < 1$ the Neumann
  series $\sum_{j=0}^\infty A^j$ converges to $(\II-A)^{-1}$, which implies that $\II-A$ is invertible \cite{reed1981functional}.
  By the
  triangle inequality and the submultiplicativity  of the $\infty$-norm, this gives
  \begin{align}
    \label{eq:inversebound}
    \|(\II-A)^{-1}\|_\infty \leq \sum_{j=0}^\infty \|A\|_\infty^j =
    \frac{1}{1-\|A\|_\infty}.
  \end{align}
  Setting $A=-P^{(1,1)}$ and applying $\|P^{(1,1)}\|_\infty=\epsilon^{(1,1)}<1$ proves
  $\II+P^{(1,1)}$ is invertible.
  Let $\cC$ be the Schur complement of
  $\II+P^{(1,1)}$ in \eqref{eq:interpmat} so that
  \begin{equation}
    \cC = \II+P^{(0,0)}-P^{(1,0)}(\II+P^{(1,1)})^{-1}P^{(0,1)}.
  \end{equation}
  By the triangle inequality and \eqref{eq:inversebound} applied with
  $A=-P^{(1,1)}$ we obtain
  \begin{equation}
    \|\II-\cC\|_\infty \leq \epsilon^{(0,0)} + \frac{\epsilon^{(1,0)}\epsilon^{(0,1)}}{1-\epsilon^{(1,1)}}=c<1,
  \end{equation}
  proving $\cC$ is invertible.  As both $\II+P^{(1,1)}$ and its Schur
  complement $\cC$ are invertible, the matrix in
  \eqref{eq:interpmat} is also invertible (see e.g. \cite{zhang1999matrix}), which establishes the first claim.

  By applying blockwise Gaussian elimination we solve
  \eqref{eq:interpmat} in terms of $\cC$ to obtain
  \begin{align}
    \alpha &= \cC^{-1}\xi\label{eq:alphascur}\\
    \beta &= -(\II+P^{(1,1)})^{-1}P^{(0,1)}\alpha\label{eq:betaschur}.
  \end{align}
  Applying \eqref{eq:inversebound} and noting that $\normInf{\xi}=1$
  we obtain the required bounds on $\normInf{\alpha}$ and $\normInf{\beta}$.
  Finally, 
  \begin{equation}
    (\II-\cC)\alpha = \alpha - \xi,
  \end{equation}
which implies $\normInf{\alpha-\xi}\leq c\normInf{\alpha}$ and completes
  the proof.
\end{proof}
\section{Proof of \Cref{thm:cluster}}
\label{sec:clusterproof}
\begin{proof}[\unskip\nopunct]
  Throughout we use the fact that $\Theta$ having separation $\Delta>0$ (\Cref{def:separation}) implies $D_i^+\leq
  D_j^-$ for $j>i$.  When $i \leq k-1$, \eqref{eq:closer}
  implies that
  $$\theta \leq
  \frac{\sigma_iD_{i+1}^-+\sigma_{i+1}D_i^+}{\sigma_i+\sigma_{i+1}},$$
  where the righthand side is the $\sigma$-weighted average of
  $D_i^+$ and $D_{i+1}^-$.  If $\theta \geq D_i^+$ this gives
  $$\frac{D_{i+1}^--\theta}{\sigma_{i+1}} \geq
  \frac{D_{i+1}^--D_i^+}{\sigma_i+\sigma_{i+1}}
  \geq \frac{\max(\sigma_i,\sigma_{i+1})\Delta}{\sigma_i+\sigma_{i+1}}
  \geq \frac{\Delta}{2}.$$
If $\theta<D_i^+$ then 
\begin{equation}
  \frac{D_{i+1}^--\theta}{\sigma_{i+1}} >
  \frac{D_{i+1}^--D_i^+}{\sigma_{i+1}} \geq
  \frac{D_{i+1}^--D_i^+}{\max(\sigma_i,\sigma_{i+1})} \geq \Delta > \frac{\Delta}{2},
\end{equation}
by the separation conditions, proving \eqref{eq:cluster1}.
If $i+1 < j \leq |\Theta|$
then $d(\theta_{i+1},\theta_j)\geq \Delta(j-(i+1))$, so that
\begin{equation}
  \frac{D_j^--\theta}{\sigma_j} \geq \frac{D_j^--D_{i+1}^+}{\sigma_j}
  \geq \frac{D_j^--D_{i+1}^+}{\max(\sigma_j,\sigma_{i+1})} \geq \Delta(j-(i+1)),
\end{equation}
by the separation conditions applied to $\theta_j$ and $\theta_{i+1}$, which implies \eqref{eq:cluster2}.
Finally, if $j < i$ then
\begin{equation}
  \frac{\theta- D_j^+}{\sigma_j} \geq \frac{D_i^--D_j^+}{\sigma_j}
  \geq \frac{D_i^--D_j^+}{\max(\sigma_j,\sigma_i)} \geq \Delta(i-j),
\end{equation}
by the separation conditions applied to $\theta_j$ and $\theta_i$, which establishes \eqref{eq:cluster3}.
\end{proof}
\section{Proof of \Cref{lem:quadineq}}
\label{sec:quadineqproof}
\begin{proof}[\unskip\nopunct]
Multiplying through by $1-x^2$ (which is positive by assumption)
in \eqref{eq:specialineq} we obtain
\begin{equation}
  -x^3b + (2a+c)x^2 +xb - c < 0.
\end{equation}
The above inequality is implied by the simpler quadratic inequality
\begin{equation}
  (2a+c)x^2 +xb - c < 0,
\end{equation}
where the omitted term $-x^3b$ is always negative.
Since the inequality is satisfied for $x=0$ we
obtain the condition
\begin{equation}
  0 < x < \frac{-b + \sqrt{b^2+4(2a+c)c}}{2(2a+c)}.
\end{equation}
Translating this to a statement on $\Delta$ we obtain
\begin{equation}
  \Delta > 2\log\brac{\frac{2(2a+c)}{-b + \sqrt{b^2+4(2a+c)c}}}.
\end{equation}
\end{proof}

\end{document}